\documentclass{article}

\usepackage{amsmath,bm}
\usepackage{amsthm}
\usepackage{amssymb}
\usepackage{xcolor}
\usepackage{graphicx}
\usepackage[margin=1in]{geometry}
\usepackage[utf8]{inputenc}
\usepackage[colorlinks=true, allcolors=blue]{hyperref}

\newtheorem{theorem}{Theorem}[section]
\newtheorem{definition}[theorem]{Definition}
\newtheorem{proposition}[theorem]{Proposition}
\newtheorem{claim}[theorem]{Claim}
\newtheorem{corollary}[theorem]{Corollary}
\newtheorem{lemma}[theorem]{Lemma}

\newtheorem{remark}[theorem]{Remark}
\newtheorem{conjecture}{Conjecture}

\title{On the Mysteries of MAX NAE-SAT \thanks{A preliminary version of this paper appeared in SODA 2021.}}
\author{Joshua Brakensiek\thanks{University of California, Berkeley. Work was primarily completed while JB was affiliated with Stanford University, supported in part by an NSF Graduate Research Fellowship and a Microsoft Research PhD Fellowship.}\and Neng Huang\thanks{University of Michigan. Work was primarily completed while NH was affiliated with the University of Chicago, supported in part by NSF grant CCF:2008920}\and Aaron Potechin\thanks{University of Chicago, supported in part by NSF grant CCF:2008920}\and Uri Zwick\thanks{Blavatnik School of Computer Science, Tel Aviv University, Israel}}
\date{}

\newcommand{\change}[1]{{{#1}}}
\newcommand{\changeB}[1]{{{#1}}}
\newcommand{\strikeout}[1]{}

\newcommand{\GE}{\;\ge\;}
\newcommand{\LE}{\;\le\;}
\newcommand{\EQ}{\;=\;}

\newcommand{\third}{\frac{1}{3}}
\newcommand{\intd}{\int\displaylimits}

\newcommand{\eps}{\varepsilon}
\newcommand{\E}{\mathbb E}
\newcommand{\scD}{\mathcal D}

\newcommand{\Cov}{\mathrm{Cov}}
\newcommand{\Var}{\mathrm{Var}}
\newcommand{\U}{\mathrm{U}}

\newcommand{\ba}{\mathbf{a}}
\newcommand{\bb}{\mathbf{b}}
\newcommand{\be}{\mathbf{e}}
\renewcommand{\bf}{\mathbf{f}}
\newcommand{\bg}{\mathbf{g}}
\newcommand{\br}{\mathbf{r}}
\newcommand{\bv}{\mathbf{v}}
\newcommand{\bu}{\mathbf{u}}

\newcommand{\bx}{{\mathbf{x}}}

\newcommand{\beps}{{\bm{\epsilon}}}

\newcommand{\prob}{\mbox{\sl prob}}
\newcommand{\RELAX}{\mbox{\sl relax}}

\newcommand{\hard}{\text{hard}}

\newcommand{\cF}{\mathcal{F}}
\newcommand{\cC}{\mathcal{C}}

\newcommand{\NAE}{\mbox{\sl NAE}}
\newcommand{\round}{\operatorname{round}}

\newcommand{\even}{\mathrm{even}}
\newcommand{\odd}{\mathrm{odd}}

\newcommand{\MN}[1]{MAX NAE-$\{#1\}$-SAT}

\begin{document}

\maketitle

\begin{abstract}
MAX NAE-SAT is a natural optimization problem, closely related to its better-known relative MAX SAT. 
The approximability status of MAX NAE-SAT is almost completely understood if all clauses have the same size~$k$, for some $k\ge 2$. We refer to this problem as MAX NAE-$\{k\}$-SAT. For $k=2$, it is \change{a slight extension of} the celebrated MAX CUT problem.
For $k=3$, it is related to the MAX CUT problem in graphs that can be fractionally covered by triangles. 
For $k\ge 4$, it is known that an approximation ratio of $1-\frac{1}{2^{k-1}}$, obtained by choosing a random assignment, is optimal, assuming $P\ne NP$. For every $k\ge 2$, an approximation ratio of at least $\frac{7}{8}$ can be obtained for MAX NAE-$\{k\}$-SAT.
There was some hope, therefore, that there is also a $\frac{7}{8}$-approximation algorithm for MAX NAE-SAT, where clauses of all sizes are allowed simultaneously. 

Our main result is that there is \emph{no} $\frac{7}{8}$-approximation algorithm for MAX NAE-SAT, assuming the Unique Games conjecture (UGC). In fact, even for almost-satisfiable instances of MAX NAE-$\{3,5\}$-SAT (i.e., MAX NAE-SAT where all clauses have size $3$ or $5$), the best approximation ratio that can be achieved, assuming UGC, is at most $\frac{3(\sqrt{21}-4)}{2}\approx 0.8739$.  

Using calculus of variations, we extend the analysis of O’Donnell and Wu for MAX CUT to MAX NAE-$\{3\}$-SAT. We obtain an optimal algorithm, assuming UGC, for MAX NAE-$\{3\}$-SAT, slightly improving on previous algorithms. The approximation ratio of the new algorithm is about $0.9089$. This gives a full understanding of MAX NAE-$\{k\}$-SAT for every $k\ge 2$. Interestingly, the rounding function used by this optimal algorithm is the solution of an integral equation. 

We complement our theoretical results with some experimental results. We describe an approximation algorithm for almost-satisfiable instances of MAX NAE-$\{3,5\}$-SAT with a conjectured approximation ratio of 0.8728, and an approximation algorithm for almost-satisfiable instances of MAX NAE-SAT with a conjectured approximation ratio of 0.8698. We further conjecture that these are essentially the best approximation ratios that can be achieved for these problems, assuming the UGC. Somewhat surprisingly, the rounding functions used by these approximation algorithms are non-monotone step functions that assume only the values $\pm 1$.
\end{abstract}

\section{Introduction}

\subsection{Background}

In a seminal paper, Goemans and Williamson \cite{GW95} introduced the paradigm of obtaining an approximation algorithm for a constraint satisfaction problem (CSP) by first solving a Semidefinite Programming (SDP) relaxation of the problem, and then rounding the solution. They used this paradigm to obtain an $\alpha_{GW}\approx 0.87856$-approximation algorithm for MAX CUT and some approximation algorithms for MAX DI-CUT (maximum directed cut), MAX 2-SAT, and MAX SAT. Khot et al.\ \cite{KKMO07} showed that the Unique Games Conjecture (UGC) of Khot \cite{khot02} implies that no polynomial time algorithm can have an approximation ratio of $\alpha_{GW}+\varepsilon$ for MAX CUT, for any $\varepsilon>0$, thus showing that the Goemans-Williamson MAX CUT algorithm is probably optimal.

Improved approximation algorithms for MAX DI-CUT and MAX 2-SAT were obtained by Feige and Goemans \cite{FG95}, Matuura and Matsui \cite{MM03}, and then by Lewin et al.\ \cite{LLZ02} who gave a $0.94016$-approximation algorithm for MAX 2-SAT and a $0.87401$-approximation algorithm for MAX DI-CUT. Austrin \cite{Austrin07} \change{(see also Brakensiek et al.~\cite{BHZ24})} showed that the MAX 2-SAT algorithm of Lewin et al.\ \cite{LLZ02} is optimal, assuming the UGC. Austrin \cite{Austrin10} obtained some hardness results for MAX 2-AND, a generalization of the MAX DI-CUT problem. Very recently, the authors \cite{BHPZ22} obtained further improved approximation algorithms for MAX 2-AND and MAX DI-CUT and also obtained a separation between MAX CUT, MAX DI-CUT and MAX 2-AND, assuming UGC.

MAX NAE-SAT and MAX SAT are natural extensions of the MAX CUT and MAX 2-SAT problems. An instance of MAX SAT is composed of a collection of clauses. Each clause is a set (or, equivalently, a disjunction) of literals, where each literal is a variable or its negation. The goal is to assign Boolean values ($0$ or~$1$) to the variables so as to maximize the number of clauses that contain a literal that evaluates to~$1$. In the MAX NAE-SAT problem the goal is to maximize the number of clauses that contain at least one literal that evaluates to $1$ and at least one literal that evaluates to~$0$. In both MAX SAT and MAX NAE-SAT, clauses may be weighted. The approximability thresholds of MAX SAT and MAX NAE-SAT are not well understood, even under UGC. What makes them hard is that clauses are allowed to be of varying sizes. 

For an integer $k$ we let MAX $\{k\}$-SAT and MAX NAE-$\{k\}$-SAT be the versions of MAX SAT and MAX NAE-SAT in which all clauses are of size \emph{exactly}~$k$, and by MAX $[k]$-SAT and MAX NAE-$[k]$-SAT the versions in which clauses are of size \emph{at most} $k$. 
Note that MAX NAE-$\{2\}$-SAT is a natural generalization of MAX CUT in which negations are allowed. The Goemans-Williamson algorithm also achieves a ratio of $\alpha_{GW}\approx 0.87856$ for MAX NAE-$\{2\}$-SAT.

H{\aa}stad \cite{H01}, relying on the PCP theory developed by Arora et al.\ \cite{AS98,ALMSS98}, and others, proved that for every $k\ge 3$, obtaining an approximation ratio of $(1-\frac{1}{2^k})+\varepsilon$ for MAX $\{k\}$-SAT, for any $\varepsilon>0$, is NP-hard. Similarly, for every $k\ge 4$, obtaining a ratio of $(1-\frac{1}{2^{k-1}})+\varepsilon$ for MAX NAE-$\{k\}$-SAT
is also NP-hard. Intriguingly, these ratios are obtained, in expectation, by simply choosing a random assignment.

Thus, as mentioned in the abstract, while it is known that an approximation ratio of at least $\frac{7}{8}$ can be obtained for each \emph{individual} clause size, for both MAX SAT and MAX NAE-SAT, it was not known whether an approximation ratio of $\frac{7}{8}$ can be obtained for all clause sizes \emph{simultaneously} (see e.g., \cite{MM17}\footnote{\cite{MM17} states that the best known upper bound on the approximation ratio of MAX NAE-SAT is that of MAX CUT, $\approx 0.87856$, but a hardness of $7/8$ follows from the upper bound of $7/8$ for MAX NAE-$\{4\}$-SAT, which follows from H{\aa}stad's $7/8$'s hardness of MAX $\{3\}$-SAT~\cite{H01}.}). For MAX NAE-SAT, we resolve this problem by showing that no such approximation algorithm exists.

For MAX $[3]$-SAT (and MAX NAE-$[4]$-SAT), Karloff and Zwick \cite{KZ97} obtained a $\frac{7}{8}$-approximation algorithm whose computer assisted analysis appears in \cite{zwick02}. A $\frac{7}{8}$-approximation algorithm is not known even for MAX $\{1,4\}$-SAT, i.e., when all clauses are either of size $1$ or~$4$. (See Halperin and Zwick \cite{HZ01}.) For general MAX SAT and MAX NAE-SAT, Avidor et al.\ \cite{ABZ05}, improving results of Andersson and Engebretsen \cite{AE98}, Asano and Williamson \cite{AW02} and Zhang et al.\ \cite{ZYH04}, obtained a MAX SAT algorithm with a conjectured approximation ratio of $0.8434$ and a MAX NAE-SAT algorithm with a conjectured approximation ratio of $0.8279$. (We discuss below why most approximation ratios claimed for MAX SAT and MAX NAE-SAT are only conjectured.) For \emph{satisfiable} instances of MAX NAE-SAT, Zwick \cite{Zwick99a} gave an algorithm with a conjectured approximation ratio of $0.863$. (We improve on this conjectured ratio, see below.)

In a breakthrough result, Raghavendra \cite{R08,R09} (see also Brown-Cohen and Raghavendra \cite{BCR15}), showed that under UGC, the best approximation ratio achievable for a maximum constraint satisfaction problem is essentially \emph{equal} to the \emph{integrality ratio} of a natural SDP relaxation of the problem. Furthermore, Raghavendra \cite{R08,R09} showed \change{that an approximation ratio arbitrarily close to} the optimal approximation ratio can be obtained using a rounding procedure selected from a specific family of rounding functions. (For more on finding almost optimal rounding procedures, see  Raghavendra and Steurer \cite{RS09}.)

Raghavendra's result \cite{R08,R09} does \emph{not} solve the MAX SAT and MAX NAE-SAT \change{questions}, as it does not tell us what are the optimal approximation ratios achievable for these problems. It only tells us which SDP relaxations are sufficient to obtain the optimal results and that the optimal ratios are equal to the integrality ratios of these relaxations. This is certainly important and useful information, but it does \emph{not} tell us what the integrality ratios are.

\subsection{Our Results}

Our first and main result is that under UGC, there is \emph{no} $\frac{7}{8}$-approximation algorithm for MAX NAE-SAT. The \change{question} for MAX SAT is still open. Furthermore, assuming UGC, no polynomial time algorithm can achieve a ratio of more than $\frac{3(\sqrt{21}-4)}{2}\approx 0.8739$ even for \emph{almost satisfiable} instances of MAX NAE-$\{3,5\}$-SAT, i.e., instances of MAX NAE-SAT in which all clauses are of size~$3$ or~$5$, and there is an assignment that satisfies a $1-\varepsilon$ (weighted) fraction of all the clauses, for an arbitrarily small $\varepsilon>0$. We obtain the result by explicitly constructing instances of MAX NAE-$\{3,5\}$-SAT that exhibit an integrality ratio of $\frac{3(\sqrt{21}-4)}{2}\approx 0.8739$ \change{for the Basic SDP relaxation (see Section~\ref{S-basic})}. The result then follows from Raghavendra's result \cite{R08,R09}.
\change{\begin{theorem}
    For all $\eps > 0$, it is UG-hard to distinguish instances of MAX NAE-$\{3,5\}$-SAT which are $(1-\eps)$-satisfiable from instances which are not $\left(\frac{3(\sqrt{21}-4)}{2} + \eps\right)$-satisfiable. The result holds even when there's no negated literal in the instances.
\end{theorem}}

Our second result is an optimal approximation algorithm for MAX NAE-$\{3\}$-SAT, assuming UGC. The approximation ratio of this algorithm is $\approx 0.9089$, slightly improving on the previous approximation ratio of $\approx 0.90871$ for MAX NAE-$\{3\}$-SAT~\cite{Zwick99a}. This means that the optimal approximation ratios of MAX NAE-$\{k\}$-SAT, for every $k\ge 2$, are now known  (see Table~\ref{T-NAE-k}). The rounding function used by the optimal MAX NAE-$\{3\}$-SAT algorithm is the solution of an \emph{integral equation}. The integral equation is obtained using a \emph{Calculus of Variations} approach. The integral equation does not seem to have a closed-form solution, but the optimal rounding function can be approximated to any desired accuracy. We show that this can be done by solving a system of linear equations. A similar integral equation can be used to characterize the optimal rounding function for MAX CUT with a given completeness, giving an alternative description of the optimal rounding functions that are described by O'Donnell and Wu \cite{OW08}. 
\change{\begin{theorem}
There is an explicit integral equation defining an SDP rounding function of MAX NAE-$\{3\}$-SAT which achieves an approximation ratio of $\approx 0.9089$. Furthermore, the obtained approximation algorithm is optimal, assuming UGC.
\end{theorem}}
\begin{table}
    \centering
    \begin{tabular}{|c|c|c|c|}
    \hline
         MAX NAE-$\{k\}$-SAT & Optimal ratio & Algorithm & Hardness \\
    \hline\hline
         $k=2$ & $\change{\alpha_{GW}}\approx 0.8786$ & Goemans-Williamson \cite{GW95} & Khot et al.\ \cite{KKMO07} \\
         \hline
         $k=3$ & $\approx 0.9089$ & \change{This paper} & \change{This paper} \\
         \hline
         $k\ge 4$ & $1-\frac{1}{2^{k-1}}$ & Random assignment & H{\aa}stad \cite{H01} \\
         \hline
    \end{tabular}
    \caption{Optimal approximation ratios for MAX NAE-$\{k\}$-SAT.}
    \label{T-NAE-k}
\end{table}
\paragraph{\change{Experimental Results.}}We next experiment with approximation algorithms for MAX NAE-SAT, as well as some restricted versions of the problem. For a set $K\subseteq \{2,3,\ldots\}$, we let MAX NAE-$K$-SAT be the restriction of MAX NAE-SAT to instances in which each clause is of size~$k$, for some $k\in K$.

For MAX NAE-$\{3,5\}$-SAT we obtain an algorithm with a conjectured ratio of $0.872886$. This may indicate that our  upper $\frac{3(\sqrt{21}-4)}{2}\approx 0.8739$ upper bound on the best achievable ratio is not far from being tight. We conjecture that the optimal ratio is much closer, if not equal, to $0.872886$, but new techniques would be needed to prove it. For (almost) satisfiable instances of MAX NAE-SAT we obtain an approximation algorithm with a conjectured ratio of $0.869809$. We further conjecture that this is essentially the best approximation ratio that can be obtained for the problem. Interestingly, this ratio is obtained for clauses of size 3,7 and 8. We thus believe that for (almost) satisfiable instances, MAX NAE-$\{3,7,8\}$-SAT is as hard as MAX NAE-SAT.

The exact approximation ratios achievable for MAX NAE-$\{3,5\}$-SAT, MAX NAE-$\{3,7,8\}$-SAT and MAX NAE-SAT are not of great importance. What we believe \emph{is} important is the nature of the rounding procedures used to obtain what we believe are optimal, or close to optimal, results. All our algorithms, as well as most previous approximation algorithms, round the solution of the SDP relaxation using the RPR$^2$ (Random Projection followed by Randomized Rounding) technique introduced by Feige and Langberg \cite{FL06}. This technique employs a rounding function $f:(-\infty, \infty)\to[-1,1]$. (For more on the RPR$^2$ technique see Section~\ref{sub-RPR2}.) For MAX NAE-$\{3\}$-SAT, the optimal rounding function~$f$ is the solution of the integral equation mentioned above. What is intriguing about the best rounding functions we found for versions of MAX NAE-SAT that involve clauses of at least two different sizes is that they are \emph{step functions} that only attain the values $\pm 1$. We have some possible explanations for this phenomenon using Hermite expansion (see Section~\ref{sub-Hermite}).
\subsection{Technical Overview}\label{sub-tech}

\subsubsection*{RPR$^2$ and Raghavendra's Theorem.} 
Raghavendra \cite{R08,R09} showed that assuming the Unique Games Conjecture, to approximate \change{Boolean} CSPs it is sufficient to consider a Basic SDP which assigns a unit vector $\bv_i$ to each variable $x_i$. The idea behind these vectors is that the SDP ``thinks'' there is a distribution of solutions where for all $i$ and $j$, $\E[{x_i}x_j] = \bv_i \cdot \bv_j$. The value of the SDP is the expected number of constraints that the SDP ``thinks'' are satisfied by this distribution. We formally state the Basic SDP and make this precise in Section~\ref{S-prelim}.

In order to obtain an actual assignment from this Basic SDP, we need a \emph{rounding \change{procedure}} which takes these vectors $\{\bv_i: i \in [n]\}$ and outputs values $\{X_i: i \in [n]\}$ where each $X_i \in \{-1,1\}$.
For this, there is a canonical method known as \emph{multi-dimensional} RPR$^2$. Multi-dimensional RPR$^2$ involves taking a \emph{rounding function} $f : \mathbb R^d \to [-1, 1]$, and assigning $1$ to $x_i$ with probability \[\frac{f(\br_1 \cdot \bv_i, \hdots, \br_d \cdot \bv_i)+1}{2}\;,\] where $\br_1, \hdots, \br_d \sim \mathcal N(0, I_n)$ are randomly sampled $n$-dimensional Gaussian variables, common to all the variables. 

Raghavendra's theorem \cite{R08,R09} formally states that if \strikeout{you can demonstrate} \change{there is} a MAX CSP instance whose SDP value is at least $c$ but the best integral assignment to the $x_i$'s satisfies at most $s$ fraction of the constraints, then assuming the \change{Unique} Games Conjecture, distinguishing instances which have a $(c-\eps)$-fraction satisfying assignment from instances which do not have an $(s+\eps)$-fraction satisfying assignment is NP-hard. One can also prove that if this CSP instance is chosen to minimize $s$ (for a fixed value of $c$) then given any CSP with a $(c+\eps)$-fraction satisfying assignment one can find a $(s-\eps)$-fraction satisfying assignment \change{in polynomial time}. By a suitable minimax argument (see, e.g., Section 7 of~\cite{BCR15}), such an algorithm can be attained by applying multidimensional RPR$^2$ rounding, where the rounding function is sampled from a suitable distribution.

In our proofs, we analyze the performance of various RPR$^2$ rounding \strikeout{rules} \change{functions}, by looking at the \emph{low-degree moments} \change{of the functions}. For instance, the second moment when $d=1$ is
\[
F_2[f](\rho) \EQ \underset{\br \sim \mathcal N(0, I_n)}{\E} \left[f(\br \cdot \bu)f(\br \cdot \bv)\right], \text{ where $\bu, \bv$ \change{are} unit vectors with dot product $\rho$.}
\]
In particular, for NAE-$\{k\}$-SAT, one can show it suffices to look at moments of $2\ell$ variables for $2\ell \le k$.

\change{The second moment} $F_2[f](\rho)$ has a number of nice properties. It is an increasing function of $\rho$, it is an odd function, and it is convex on nonnegative inputs. These properties play a crucial role in our results. 

\subsubsection*{$\frac{3(\sqrt{21}-4)}{2}\approx 0.8739$ UGC-hardness of MAX NAE-SAT.} In Section~\ref{S-int-gap}, we show that assuming the Unique Games Conjecture, MAX NAE-SAT does not admit a $7/8$-approximation. In fact, we show \strikeout{for} \change{that even for} \emph{near-satisfiable} instances of MAX NAE-$\{3,5\}$-SAT that we cannot achieve a $7/8$-approximation. The key observation which leads to hardness, is that for NAE-$\{3\}$-SAT, a difficult triple $(\bv_1, \bv_2, \bv_3)$ of unit vectors to round has pairwise dot product $\bv_i \cdot \bv_j = -\third$ for all $i \neq j$. Likewise for NAE-$\{5\}$-SAT, a difficult quintuple $(\bv_1, \bv_2, \bv_3, \bv_4, \bv_5)$ has $\bv_i \cdot \bv_j =\third$ for $1 \le i < j\le 4$ and $\bv_i \cdot \bv_5 = 0$ for all $1 \le i \le 4.$ If we write out the expected value of \strikeout{our rounding rule} \change{the output using a rounding function~$f$}, we derive that
\[
\text{best integrality ratio} \LE \min_{p \in [0, 1]} \max_{\change{f}} \left[(1-p)\frac{3+3F_2\change{[f]}(\third)}{4} + p \frac{15-6F_2\change{[f]}(\third)-F_4\change{[f]}(\third)}{16}\right],
\] 
where $F_4\change{[f]}(\third) = F_4\change{[f]}(\third,\third,\third,\third,\third,\third)$ is the fourth moment when all of the pairwise biases are $\third$ and $1-p$ and~$p$ are the relative weights of the NAE-$\{3\}$-SAT and NAE-$\{5\}$-SAT clause types, respectively. We prove that $F_2\change{[f]}(\third) \in [0, \third]$ and $F_4\change{[f]}(\third) \ge F_2\change{[f]}(\third)^2$, \change{for any rounding function~$f$}. These together imply that for $p = \frac{\sqrt{3}}{21}$, the above expression is at most $\frac{3(\sqrt{21}-4)}{2} \approx 0.8739$.

In addition to this indirect argument that Raghavendra's optimal algorithm cannot achieve better than a $0.8739$ approximation, we also give an explicit integrality gap instance which achieves the same upper bound on the approximation ratio. 

\subsubsection*{An optimal algorithm for MAX NAE-$\{3\}$-SAT (assuming UGC).} In Section~\ref{S-COV}, we tackle the problem of finding an optimal approximation algorithm for MAX NAE-$\{3\}$-SAT. Our approach follows a similar template to that of \change{O'Donnell and Wu} \cite{OW08} for MAX CUT. \strikeout{Like} \change{As in}~\cite{OW08}, we consider the max-min problem of finding an RPR$^2$ rounding function $f$ which achieves at least the optimal approximation ratio for any distribution $\mathcal D$ of triples of unit vectors (pairs of unit vectors for MAX CUT). In the case of MAX CUT, \cite{OW08} showed that the ``most difficult'' \change{distribution} $\mathcal D$ is a distribution over pairs of identical unit vectors, and pairs of unit vectors with dot product $\rho < 0$. \strikeout{Likewise} \change{Similarly}, we show (see Theorem~\ref{NAE3SATdisttheorem}) that the most difficult \change{distribution} $\mathcal D$ come from triples of vectors that \strikeout{are} have pairwise dot products $(\rho, \rho, 1)$ or $(\rho_0, \rho_0, \rho_0)$, where $\rho \in (-1, 0]$ and $\rho_0 = \max(\rho, -\third)$ or $\rho_0 = 1$. 

Once we fix a distribution $\mathcal D$, the performance of \change{a rounding function} $f$ can be expressed using a double integral.  Using \emph{calculus of variations}, we \strikeout{can then produce} \change{obtain} an integral equation that must be satisfied by the optimal $f$. \change{The equation is a \emph{Fredholm integral equation of the second kind}. (See, e.g., \cite{polyanin2008handbook}.)}

If we discretize $f$ as a step function, the integral equation becomes a system of linear equations. (Cf.\ Section~6 of the full version of~\cite{OW08}). Thus, \change{as} \strikeout{like} in~\cite{OW08}, we can efficiently compute an arbitrarily close approximation of the optimal function for MAX NAE-$\{3\}$-SAT. We found that the optimal approximation ratio is approximately $0.9089$. We also observed that the optimal function is very close to an $s$-linear function, \change{i.e., a function} of the form $f(x) = \max(-1, \min(1, sx))$, but there is a nontrivial error term on the order of $O(x^3)$.

\subsubsection*{Conjectured near-optimal algorithm for MAX NAE-SAT \change{for almost satisfiable instances}.} 

In Section~\ref{S-approx}, we give better experimental algorithms for MAX NAE-SAT and various restrictions of it, for nearly satisfiable instances  ($1-\eps$ completeness, for $\eps > 0$ sufficiently small). Although the performance ratios of our algorithms are based on experimental conjectures (see Section~\ref{sub-conj}), we believe that these are essentially the best approximation ratios that can be obtained for these problems. As mentioned in the introduction, these algorithms use fairly surprising rounding functions: non-monotone step functions that only attain the values $\pm1$.
 
All the previous experimental \strikeout{work} \change{results} (including the previous state of the art results of Avidor et al. \cite{ABZ05}) only  considered \emph{monotone} rounding functions $f : \mathbb R \to [-1, 1]$. This was for good reason, as for MAX CUT, and our new results for MAX NAE-$\{3\}$-SAT, the optimal rounding functions are provably monotone. On the other hand, for even a simple extension like MAX NAE-$\{3,5\}$-SAT, it turns out that the optimal rounding function is very likely not monotone, but rather a non-monotone step function. In fact, we conjecture that the function 

\[
f(x) \EQ \begin{cases}
-1 & x < -\alpha\\
\phantom{-}1 & x \in [-\alpha, 0]\\
-1 & x \in (0, \alpha]\\
\phantom{-}1 & x > \alpha
\end{cases} \;,
\]
where $\alpha \approx 2.27519364977$ (see Figure~\ref{fig:double-step}) is the optimal rounding function for near-satisfiable instances of MAX NAE-$\{3,5\}$-SAT. We have further experimental results for near satisfiable instances of MAX NAE-SAT.

\section{Preliminaries}\label{S-prelim}

\subsection{MAX CSP\change{s} and \change{the} Basic SDP}\label{S-basic}

\begin{definition}[MAX CSP$(\cF)$] Let $\cF$ be a set of Boolean functions. Each $f\in \cF$ is a function $f:\{-1,1\}^{k(f)}\to\{0,1\}$, where $k(f)$ is the number of arguments of~$f$. An instance of MAX CSP$(\cF)$ is composed of a set of variables $V=\{x_1,x_2,\ldots,x_n\}$ and a collection $\cC$ of weighted clauses. Each clause $C\in \cC$ is a tuple $(w,f,i_1,i_2,\ldots,i_k,\change{s_1,\ldots,s_k})$, where $w>0$, $f\in \cF$, $k=k(f)$, $i_1,i_2,\ldots,i_k$ are distinct indices from $[n]=\{1,2,\ldots,n\}$, and $\change{s_1,s_2,\ldots,s_k}\in\{-1,1\}$. Such a clause represents the constraint $f(\change{s_1}x_{i_1},\change{s_2}x_{i_2},\ldots,\change{s_k}x_{i_k})=1$. (Here $\change{s_j}x_{i_j}$ denotes multiplication.) The goal is to find an assignment $\alpha:V\to\{-1,1\}$ that maximizes the sum of the weights of the satisfied constraints.
\end{definition}

The above definition defines Boolean constraint satisfaction problems. (For technical reasons, the values assumed by the variables are $-1$ and $+1$ rather than $0$ and~$1$.) The definition can be extended to other domains. We require the indices $i_1,i_2,\ldots,i_k$ to be distinct. (If some of the indices are equal, this is equivalent to having a constraint on a smaller number of variables.) We allow constraints to be applied to both variables and their negations. (This is equivalent to requiring the set $\cF$ of allowed constrains to be closed under negations of some of the arguments.) \change{We can also consider the version in which we do not allow variable negations, i.e., we do not have $s_1, \ldots, s_k$ in the above definition. This version is sometimes referred to as monotone MAX CSP, denoted by MAX CSP$^+$. Note that any MAX CSP$^+(\cF)$ instance is automatically a MAX CSP$(\cF)$ instance. } For more on general constraint satisfaction problems, and their maximization variants, see \cite{bulatov2000constraint,zhuk2017proof}.
\begin{definition}
For every integer $k \geq 2$, the \emph{Not-All-Equal} predicate on $k$ variables is defined as 
\[
\NAE_k(x_1, \ldots, x_k) \EQ \left\{\begin{array}{ll}
    0 & \text{if } x_1 = x_2 = \cdots = x_k, \\
    1 & \text{otherwise.} 
\end{array}\right.
\]
\end{definition}

We remark that the Not-All-Equal predicates are \emph{even} predicates: a collection of Boolean variables are not-all-equal if and only if their negations are not-all-equal.

In this paper, we only consider problems of the form MAX CSP$(\cF)$ where $\cF\subseteq \{\NAE_k \mid k\ge 2\}$\change{, and their monotone version  MAX CSP$^+(\cF)$}. For a set $K\subseteq \{2,3,\ldots\}$, we let MAX NAE-$K$-SAT be a shorthand for MAX CSP$(\{\NAE_k\mid k\in K\})$\change{, and we will refer to  MAX CSP$^+(\{\NAE_k\mid k\in K\})$ as monotone MAX NAE-$K$-SAT}.

For a \change{MAX} CSP instance with variable set $V = \{x_1, \ldots, x_n\}$ and clause set $\mathcal{C}$, we can define the Basic SDP as follows. We maintain a unit vector $\bv_i \in \mathbb R^{n+1}$ for each variable $x_i$ and a special unit vector $\bv_0$, and for each clause $C$ a probability distribution $p_C(\alpha)$ over $\mathcal{A}(C)$, the set of all assignments on variables in $C$. Here $z_i$ stands for a literal which is either $x_i$ or $-x_i$, and $\alpha(x_i)$ is the value $\alpha$ assigns to $x_i$. We use the notation $z_i \in C$ to denote that $z_i$ appears in the clause $C$.

\begin{align*}
&\quad \max \sum_{C \in \mathcal C}  w_C\left(\sum_{\alpha \in \mathcal{A}(C)}p_C(\alpha)C(\alpha)\right)\\
\forall i \in \{0, 1, 2, \ldots, n\},\ \ \ & \quad \bv_i \cdot \bv_i = 1\\
\forall C \in \mathcal C, \forall z_i, z_j \in C, \ \ \ & \quad \bv_i \cdot \bv_j = \sum_{\alpha \in \mathcal{A}(C)} \alpha(x_i)\alpha(x_j)p_C(\alpha)\\
\forall C \in \mathcal C, \forall z_i \in C, \ \ \ & \quad \bv_i \cdot \bv_0 = \sum_{\alpha \in \mathcal{A}(C)} \alpha(x_i)p_C(\alpha)\\
\forall C \in \mathcal C, \forall \alpha \in \mathcal{A}(C),\ \ \ & \quad p_C(\alpha) \ge 0.
\end{align*}

Note that the first two constraints imply that $\sum_{\alpha \in \mathcal{A}(C)} p_C(\alpha) = 1$. Let $b_i = \bv_0 \cdot \bv_i$ and $b_{i,j} = \bv_i \cdot \bv_j$ for $i \neq j \in [n]$. We call $b_i$ \emph{biases} and $b_{i,j}$ \emph{pairwise biases}. Informally speaking, $b_i$ \strikeout{is intended for} \change{represents} $\E[x_i]$, $b_{i,j}$ \strikeout{is intended for} \change{represents} $\E[x_ix_j]$, and the SDP constraints are saying that local assignments should agree with these biases and pairwise biases.
Note that for a \change{MAX} CSP \change{instance} with even predicates as in our case, $\bv_0$ (and therefore the biases) is not useful, as we can always combine a solution and its negation to get a new solution with 0 biases while preserving the objective value.

\begin{definition}[\change{Completeness, soundness and integality gap curve}]
For a \change{MAX} CSP instance $\Phi$, we define its \emph{completeness}, denoted $c(\Phi)$, to be its SDP value, and its \emph{soundness}, denoted $s(\Phi)$, to be the value of the optimal integral solution to $\Phi$. For a MAX CSP \change{(or MAX CSP$^+$)} problem $\Lambda$, define its \emph{integrality gap curve} to be the function $S_\Lambda(c): c \mapsto \inf\{s(\Phi) \mid \Phi \in \Lambda, c(\Phi) = c\}$.
\end{definition}

The Unique Games Conjecture, introduced by Khot~\cite{khot02}, is a central conjecture in the study of approximation algorithms. One version of the conjecture is as follows.

\begin{definition}
In the Unique Games problem with $R$ labels, we are given a graph $G = (V, E)$ and a set of permutations $\{\pi_e: [R] \to [R] \mid e \in E\}$. An assignment $\alpha: V \to [R]$ is said to satisfy an edge $e = (u, v)$ if $\pi_e(\alpha(u)) = \alpha(v)$. Our goal is to find an assignment that maximizes the number of satisfied edges.
\end{definition}

\begin{conjecture}[Unique Games Conjecture]\label{C-UGC}
For any constant $\varepsilon > 0$, there exists $R \in \mathbb{N}$ such that for Unique Games problem with $R$ labels, it is NP-hard to distinguish between the following two cases:
\begin{itemize}
    \item there exists an assignment that satisfies $(1- \varepsilon)$-fraction of edges;
    \item no assignment can satisfy more than $\varepsilon$-fraction of edges.
\end{itemize}
\end{conjecture}

We say that a problem is Unique Games hard if it is NP-hard assuming the Unique Games Conjecture. Raghavendra, in his breakthrough paper~\cite{R08,R09}, showed the following result which exhibited the close relation between the integrality gap curve of a CSP problem and its unique game hardness.

\begin{theorem}[Raghavendra~\cite{R08,R09}]
For a MAX CSP \change{(or MAX CSP$^+$)} problem $\Lambda$, Let $U_\Lambda(c)$ be the best polynomial-time computable integral value on instances with optimal value $c$ assuming the Unique Games Conjecture. Then we have
\begin{enumerate}
    \item For every constant $\eta > 0$ and $c \in [0, 1)$, $U_\Lambda(c) \leq S_{\Lambda}(c + \eta) + \eta$.
    \item For every constant $\eta > 0$ and $c \in (0, 1]$, there exists a polynomial time algorithm that on any input with completeness $c$, outputs an integral solution with value at least $U_{\Lambda}(c - \eta) - \eta$.
\end{enumerate}
\end{theorem}

We will describe Raghavendra's algorithm in the following subsection. For now we remark that this theorem essentially says that the integrality gap of the Basic SDP is the Unique Games hardness threshold of the \change{MAX} CSP \change{problem}, so in order to show Unique Games hardness, it suffices to construct an integrality gap instance for the Basic SDP.

\subsubsection{MAX NAE-SAT}

The approximation constant of MAX NAE-SAT has multiple interpretations in the literature~\cite{ABZ05,MM17}. The most common one is to assume that any instance $\Psi$ consists of an arbitrarily large clauses, in particular they grow as a function of the number of variables of the instance. Note that Raghavendra's result does not apply in this case, as the Basic SDP for an instance with clauses of length $\Omega(n)$ has exponential size.

A secondary interpretation is to consider the limiting behavior of MAX NAE-$[k]$-SAT as $k \to \infty$. In this case, each $k$ is a fixed value so the Basic SDP has polynomial size and Raghavendra's theorem applies.

It turns out these two views are essentially identical in the limit, and as such we assume the latter in the remainder of the main body of the paper. See Appendix~\ref{S-equiv} for more details.

\subsection{The \texorpdfstring{RPR$^2$}{RPR2} rounding technique}\label{sub-RPR2}

RPR$^2$ (Random Projection followed by Random Rounding) is an algorithm for rounding SDP solutions proposed by Feige and Langberg~\cite{FL06}. It generalizes hyperplane rounding and outward rotation techniques and has played an important role in designing SDP-based approximation algorithms. The RPR$^2$ technique chooses a function $f: \mathbb{R} \to [-1, 1]$ and performs the following action:
\begin{itemize}
    \item Choose a random $\br \sim N(\mathbf{0}, \mathbf{I}_n)$, where $n$ is the length of the SDP solution vectors and $N(\mathbf{0}, \mathbf{I}_n)$ is the $\change{n}$-dimensional standard normal distribution.
    \item For each variable $x_i$, compute the inner product $t_i = \br \cdot \bv_i$. (Random projection)
    \item For every $i$,  independently assign $x_i = 1$ with probability $\frac{1 + f(t_i)}{2}$, \change{and} $x_i = -1$ with probability $\frac{1 - f(t_i)}{2}$. (Random Rounding)
\end{itemize}
Some previously used functions for $f$ include the sign function (in hyperplane rounding) and piecewise linear functions.

In~\cite{R08,R09}, Raghavendra showed an algorithm which achieves the integrality gap of the Basic SDP for any CSP within an arbitrary precision, assuming UGC. His algorithm makes use of the following procedure\footnote{We are stating Raghavendra's algorithm for even CSPs in which the assignment $(x_1, \hdots, x_n)$ has the same weight as the assignment $(-x_1, \hdots, -x_n)$ for all assignments to $(x_1, \hdots, x_n)$ (such as MAX CUT and MAX NAE-SAT). Otherwise, Raghavendra considers a more general algorithm which incorporates $\bv_0 \cdot \bv_i$ for all $i$.} (denoted $Round_f$ in~\cite{R08,R09}), which is essentially a multi-dimensional version of RPR$^2$:
\begin{itemize}
    \item Choose a function $f: \mathbb{R}^d \to [-1, 1]$. Sample $d$ random normal vectors $\br^{(1)}, \ldots, \br^{(d)} \sim N(\mathbf{0}, \mathbf{I}_n)$. 
    \item For each variable $x_i$ and $1\leq j \leq d$, let $t_i^{(j)} = \bv_i \cdot \br^{(j)}$.
    \item Let $p_i = f(t_i^{(1)}, \ldots, t_i^{(d)})$. Assign $x_i = 1$ with probability $\frac{1 + p_i}{2}$ and $x_i = -1$ with probability $\frac{1 - p_i}{2}$.
\end{itemize}

When $d = 1$, $Round_f$ is the usual RPR$^2$ procedure. We will refer to this procedure as  RPR$^2_d$ if $f$ is a $d$-dimensional function. Raghavendra's algorithm runs $Round_f$ for every $f$ in a pre-computed $\varepsilon$-net and picks the best $f$. It also has a preprocessing ``smoothing'' step, which in our case corresponds to scaling the pairwise biases by a common factor. \change{On the other hand, given some SDP solution that is hard to round for any $\round_f$ procedure, we can also take an $\epsilon$-net of the $d$-dimensional unit sphere and construct an integrality gap instance for the Basic SDP. This idea and the argument in~\cite{RS09} imply the following lemma.
\begin{lemma}\label{lem:rpr2_ughardness}
    Let $c > s > 0$ and $\Phi$ be a (resp. monotone) MAX NAE-SAT instance with the following properties:
    \begin{itemize}
        \item $\Phi$ has an SDP solution with value at least $c$, and 
        \item for all $f$, the expected number of constraints satisfied by $Round_f$ on said SDP solution is at most $s$,
    \end{itemize} then it is UG-hard to distinguish (resp. monotone) MAX NAE-SAT instances with completeness at least $c - \epsilon$ from instances that are not $(\changeB{s} + \epsilon)$-satisfiable for any $\epsilon > 0$.
\end{lemma}
}

It turns out that for \strikeout{$Round_f$ procedures} \change{MAX NAE-SAT}, it suffices to consider odd $f$, due to the following lemma.

\begin{lemma}\label{lemma:odd_sufficient}
Let $f: \mathbb{R}^d \to [-1, 1]$ and $f'(x) = (f(x) - f(-x))/2$ be its odd part. For any \change{MAX} CSP, the worst case performance of $Round_{f'}$ is at least as good as $Round_f$.
\end{lemma}

\begin{proof}
Consider an arbitrary \change{MAX} CSP and let $\Phi$ be a worst-case instance for $Round_{f'}$. Observe that the $Round_{f'}$ procedure is equivalent to the following: independently for every variable $x_i$, with probability 1/2 apply the rounding function $f$ on $\bv_i$, and with probability 1/2 apply the rounding function $f$ on $-\bv_i$ and flip the result. Observe further that by replacing $\bv_i$ with $-\bv_i$ and flipping the outcome, we are essentially applying $f$ to a new instance with $-x_i$ in place of $x_i$ in $\Phi$. This implies that the value of $Round_{f'}$ on $\Phi$ is an average of $2^n$ values ($n$ is the number of variables in $\Phi$) where each value is $Round_f$ evaluated on an instance obtained by flipping some variables in $\Phi$. It follows that in some of these instances $Round_f$ has a value as bad as the value of $Round_{f'}$ on $\Phi$.
\end{proof}

\strikeout{In the later sections, we will only consider $Round_f$ with odd $f$}.

\subsection{Moment functions of rounding rules}\label{sub-moment}

Given a rounding rule \change{$Round_f$}, we can view its outputs as random variables $X_1, \ldots, X_n$, where $X_i$ is the value assigned to $x_i$. An important property of the RPR$^2$ technique, either one-dimensional or multi-dimensional, is that the $k$-wise moments of these random variables are determined entirely by the pairwise dot products of the vectors in the SDP solution. In other words, for any \strikeout{group} \change{set} of variables $\{X_{i_1}, \ldots, X_{i_k}\}$, $\E[X_{i_1}\cdots X_{i_k}]$ is a function \strikeout{on} \change{of} $b_{i_1,i_2}, b_{i_1,i_3}, \ldots, b_{i_{k-1},{i_k}}$. This inspires the following definition.

\begin{definition}[e.g.,~\cite{P19,HP19}]\label{def:f2f4}
For a rounding procedure $Round_f$, define $F_{k}[f](b_{1,2}, b_{1,3}, \ldots, b_{{k-1},{k}})$ to be the $k$-wise moment $\E[X_{1}\cdots X_{k}]$, where $X_{1}, \ldots, X_{k}$ are random variables obtained by applying $Round_f$ to vectors $\bv_{1}, \ldots, \bv_{k}$ such that $\bv_{i} \cdot \bv_{j} = b_{i,j}$ for every $1 \leq i < j \leq k$.
\end{definition}

We will usually omit the argument $f$ and write $F_k$ for $F_{k}[f]$ unless there are multiple rounding functions in question. \change{We will also write $F_k[f](x)$ instead of $F_k[f](x, x, \ldots, x)$ if the input biases are all equal.} The following observation is immediate:
\begin{proposition}
For every odd function $f$ and odd integer $k > 0$, $F_{k}[f] = 0$.
\end{proposition}
\begin{proof}
Let $X_1, \ldots, X_k$ be the sequence of random variables that $Round_f$ outputs on SDP vectors $\bv_1, \ldots, \bv_k$. \change{Taking the negation of each of these vectors, each $X_i$ is also negated because $f$ is odd}
while the pairwise biases stay the same. It follows that
\begin{align*}
    \E[X_1\cdots X_k] & \EQ F_k[f](\bv_1 \cdot \bv_2, \ldots, \bv_{k-1}\cdot \bv_k) \\ & \EQ F_k[f]((-\bv_1) \cdot (-\bv_2), \ldots, (-\bv_{k-1})\cdot (-\bv_k)) \\
    & \EQ\E[(-X_1)\cdots (-X_k)]\\
    & \EQ (-1)^k\E[X_1\cdots X_k]\;.
\end{align*}
The proposition then follows since $k$ is odd.
\end{proof}
The first non-trivial moment function $F_2$ has been studied in previous work and is relatively well understood. By definition, $F_2(x)$ is the expected value of the product of a pair of variables whose SDP solution vectors have inner product $x$. Observe that if $\br \sim N(0, I_n)$ is a standard normal vector and $\bv \in \mathbb{R}^n$ is a unit vector, then their inner product $\br \cdot \bv$ has the standard normal distribution $N(0, 1)$. Furthermore, if $\bv_i$ and $\bv_j$ are unit vectors with inner product $x$, then $\br \cdot \bv_i$ and $\br \cdot \bv_j$ are standard normal variables with correlation $x$. It follows that $F_2[f](x)$ is equal to the \emph{Gaussian noise stability} of $f$ at $x$, defined as $\E_{\bu, \bv}[f(\bu)f(\bv)]$ where $\bu, \bv$ is a pair of $d$-dimensional $x$-correlated standard Gaussians. \change{This quantity can also be expressed using the \emph{Gaussian noise operator} $\U$, where $\U_\eta f(\bu):= \E_{\br \sim N(0, I_d)}[f(\eta \bu + \sqrt{1 - \eta^2}\br)]$. }
\change{
\begin{lemma}
    Let $f: \mathbb{R}^d \to [-1, 1]$ and $x \in [0, 1]$. We have 
    \[
    F_2[f](x) \EQ \E_{\bu \sim N(0, I_d)}\left[f(\bu)\cdot\U_xf(\bu)\right] \EQ \E_{\bu \sim N(0, I_d)}\left[\left(\U_{\sqrt{x}}f(\bu)\right)^2\right]\;.
    \]
\end{lemma}
The second equality in the above lemma can be generalized to higher moments $F_k[f](x)$ with equal input biases.
\begin{lemma}\label{lemma:noisyf}
Let $f: \mathbb{R}^d \to [-1, 1], x \in [0, 1]$, $k \geq 2$. We have
\[
F_k[f](x) \EQ \E_{\bu \sim N(0, I_d)}\left[(\U_{\sqrt{x}} f(\bu))^k\right].
\]
\end{lemma}
\begin{proof}
    By definition, 
\[
F_k[f](x) \EQ \E_{\bu^{(1)}, \ldots, \bu^{(k)}}\left[f(\bu^{(1)})f(\bu^{(2)})\cdots f(\bu^{(k)})\right]\;,
\]
where $\bu^{(1)}, \ldots, \bu^{(k)}$ are $d$-dimensional Gaussian vectors such that each $\bu^{(i)} \sim N(0, I_d)$, and $\bu^{(i)}$ and $\bu^{(j)}$ are $x$-correlated for $1 \leq i < j \leq k$. One way to generate such a distribution is by having $k+1$ independent $d$-dimensional standard Gaussian vectors $\bv$, $\beps^{(1)}, \ldots, \beps^{(k)}$, and setting $\bu^{(i)} = \sqrt{\rho} \cdot \bv + \sqrt{1 - \rho} \cdot\beps^{(i)}$ for $1 \leq i \leq k$. It follows that
\begin{align*}
F_k[f](x) & \EQ \E_{\bu^{(1)}, \ldots, \bu^{(k)}}\left[f(\bu^{(1)})f(\bu^{(2)})\cdots f(\bu^{(k)})\right] \\
& \EQ \E_{\bv, \beps^{(1)}, \ldots, \beps^{(k)}}\left[\prod_{i = 1}^k f\left(\sqrt{x} \cdot \bv + \sqrt{1 - x} \cdot\beps^{(i)}\right)\right] \\
& \EQ \E_{\bv}\left[\prod_{i = 1}^k \E_{\beps^{(i)}}\left[f\left(\sqrt{x} \cdot \bv + \sqrt{1 - x} \cdot\beps^{(i)}\right)\right]\right] \\
& \EQ \E_\bv \left[\left(\U_{\sqrt{x}}f(\bv)\right)^k\right]\;. \qedhere
\end{align*}
\end{proof}
}

\strikeout{From here the standard analytic tools apply, and w}\change{W}e refer to~\cite{O14} for a more thorough treatment \change{of standard analytic tools for Gaussian noise stability.}. \strikeout{Here w}\change{W}e collect a few \change{more} facts \change{about $F_2$} that are crucial to the understanding of MAX NAE-$\{3\}$-SAT. These facts \strikeout{on $F_2$} can also be found in O'Donnell and Wu~\cite{OW08}. 

\begin{proposition}[Proposition 4.7 in \change{the full version} of~\cite{OW08}]\label{prop:F2power}
$F_2[f](x)$ is a power series in $x$ with nonnegative coefficients  and radius of convergence at least $1$. When $f$ is odd (resp. even), $F_2[f](x)$ has odd (resp. even) powers only. In particular, for odd $f$, $F_2[f](-x) = -F_2[f](x)$, $F_2[f]$ is convex on $[0, 1]$ and concave on $[-1, 0]$.
\end{proposition}

\begin{theorem}[Theorem 4.4 in \change{the full version of}~\cite{OW08}]\label{thm:gauss}
For every $f$, there exists a Gaussian rearrangement $f^*$ of $f$ which is a one-dimensional, odd, increasing function with the property that for every $x \in [0, 1]$, $F_2[f](x) \leq F_2[f^*](x)$ and $F_2[f](1) = F_2[f^*](1)$.
\end{theorem}

The power of these moment functions is best seen when combined with Fourier analysis. We have the following proposition on the Fourier expansion of Not-all-equal predicates:
\begin{proposition}
The Fourier expansion of $\NAE_k: \{-1, 1\}^k \to \{0, 1\}$ is given by
\[
\NAE_k(x_1, \ldots, x_k) \EQ \frac{1}{2^{k-1}}\left(2^{k-1} - 1 - \sum_{i_1 < i_2}x_{i_1}x_{i_2} -\!\!\!\sum_{i_1 < i_2 < i_3 < i_4} \!\!x_{i_1}x_{i_2}x_{i_3}x_{i_4} - ... - \!\!\!\sum_{i_1 < i_2 < \ldots <i_{2\lfloor k / 2 \rfloor}}\prod_{j = 1}^{2\lfloor k / 2 \rfloor}x_{i_j}\right).
\]
\end{proposition}
Using linearity of expectation, we can then express the expected value of a clause using these moment functions. For MAX NAE-$\{2\}$-SAT (MAX CUT) and MAX NAE-$\{3\}$-SAT, the expression only involves~$F_2$, since~$F_1$ and~$F_3$ are zero. Higher moment functions start to play a role in clauses with 4 or more variables. Unfortunately, \strikeout{however,} even $F_4$ seems to be much more difficult to understand than $F_2$. For example, while it is true that $F_2(x) \geq 0$ for nonnegative input $x$, this is untrue for $F_4$: there exists a rounding function and a set of nonnegative inputs on which $F_4$ is negative! (See Appendix~\ref{sub:F4negative}). This can be partially explained by the fact that higher moment functions have many more inputs which are closely intertwined since they are inner products of a group of vectors. In the following section we will prove a simple inequality relating $F_2$ and $F_4$, which simple as it is, already gives some exciting new insights.

\section{Integrality gap \change{and Unique Games hardness} for MAX NAE-SAT}\label{S-int-gap}

A longstanding open question is whether MAX NAE-SAT admits a $\frac{7}{8}$-approximation. Such an approximation is known to exist if every clause has at most $4$ variables or at least $4$ variables. In this section, we answer this question negatively in the general case, assuming the Unique Games Conjecture. \change{In fact, we will see that the answer is negative even for monotone MAX NAE-SAT.}
In our proof, we will focus on the following set of pairwise biases: $(-\frac13, -\frac13, -\frac13)$ for $\NAE_3$ and $(\frac13, \frac13, \frac13, \frac13, \frac13, \frac13, 0, 0, 0, 0)$ for $\NAE_5$. We show that these biases ``fool'' the Basic SDP (\change{have} completeness 1) but are in fact very difficult to round.

\begin{lemma}\label{lem:phi_completeness}
Let $\Phi$ be a {\rm MAX NAE-$\{3,5\}$-SAT} instance \change{with an SDP solution} whose $3$-clauses all have pairwise biases $(b_{1,2}, b_{1,3},\allowbreak b_{2,3}) = (-\frac13, -\frac13, -\frac13)$ and $5$-clauses all have pairwise biases $(b_{1,2}, b_{1,3}, \ldots, b_{4,5}) = (\frac13, \frac13, \frac13, \frac13, \frac13, \frac13, 0, 0, 0, 0)$, then $\Phi$ has completeness $1$.
\end{lemma}
\begin{proof}
It suffices to show that for every clause, there exists a distribution of satisfying assignments that agrees with the global (pairwise) biases.
\begin{itemize}
    \item $3$-clauses.
    The uniform distribution on $\{(1 , 1, -1), (1, -1, 1), (-1, 1, 1)\}$ has the same pairwise biases.
    \item $5$-clauses.
    The following distribution on satisfying assignments has the same pairwise biases.
    \begin{center}
    \renewcommand{\arraystretch}{1.35}
    \begin{tabular}{|c|r|r|r|r|r|}
      \hline Probability  & $X_{1}$ & $X_{2}$ & $X_{3}$ & $X_{4}$ & $X_{5}$ \\ \hline 
       $\frac16$  & $-1$ & 1 & 1 & 1 & 1 \\ \hline
       $\frac16$  & 1 & $-1$ & 1 & 1 & 1 \\ \hline
       $\frac16$  & 1 & 1 & $-1$ & 1 & 1 \\ \hline
       $\frac16$  & 1 & 1 & 1 & $-1$ & 1 \\ \hline
       $\frac13$  & 1 & 1 & 1 & 1 & $-1$ \\ \hline
    \end{tabular}
    \end{center}
\end{itemize}
\end{proof}

The following subsections \change{are organized as follows}. In Section~\ref{sub-implicit}, we show an implicit integrality gap \change{for MAX NAE-SAT} by considering the limitations of Raghavendra's rounding scheme on these biases. \change{This is then extended to Section~\ref{sec:extend_to_monotone} for instances of monotone MAX NAE-SAT.} In Section\change{~\ref{sub-explicit}}, we exhibit an explicit gap instance which has these biases but has no good integral solution.

\subsection{Implicit integrality gap through rounding limitation}\label{sub-implicit}

\begin{lemma}\label{lemma:moment_orthogonal}
\change{Let $f: \mathbb{R}^d \to [-1, 1]$}, $\bv_1, \ldots, \bv_k \in \mathbb{R}^n$ be a group of unit vectors with $\bv_i \cdot \bv_j = b_{i,j}$. If $b_{i, k} = 0$ for every $1 \leq i < k$, then
\[
F_k[f](b_{1, 2}, b_{1,3}, \ldots, b_{k-1, k}) \EQ \change{F_{k-1}[f](b_{1, 2}, b_{1,3}, \ldots, b_{k-2, k-1})\cdot\E_{\bu \sim N(0, I_d)}[f(\bu)]}\;.
\]
\end{lemma}
\begin{proof}
Let $X_1, \ldots, X_k$ be $\pm 1$ random variables obtained by running $Round_f$ on $\bv_1, \ldots, \bv_k$. Since $b_{i, k} = 0$ for every $1 \leq i < k$, we have that \change{for $\br \sim N(0, I_n)$, $\br \cdot \bv_k$ is independent from $\br \cdot \bv_1, \ldots, \br \cdot \bv_{k-1}$. This implies that $X_k$ is independent from $X_1, \ldots, X_{k-1}$. It follows that
\begin{align*}
F_k[f](b_{1, 2}, b_{1,3}, \ldots, b_{k-1, k}) &\EQ \E[X_1, \ldots X_k] \\
& \EQ \E[X_1, \ldots X_{k-1}]\E[X_k]  \\
& \EQ F_{k-1}[f](b_{1, 2}, b_{1,3}, \ldots, b_{k-2, k-1})\cdot\E_{\bu \sim N(0, I_d)}[f(\bu)]\;. \qedhere
\end{align*}} 
\end{proof}
\change{If $f$ is an odd function, then $\E_{\bu \sim N(0, I_d)}[f(\bu)]=0$ and we will have $F_k[f](b_{1, 2}, b_{1,3}, \ldots, b_{k-1, k}) = 0$.}

\begin{lemma}\label{lemma:f4f2}
\change{Let $f: \mathbb{R}^d \to [-1, 1], x \in [0, 1]$. We have $F_4[f](x) \geq F_2[f](x)^2.$}
\end{lemma}
\begin{proof}
\change{By Lemma~\ref{lemma:noisyf} as well as Jensen's inequality, we have
\[
F_4[f](x) \EQ \E_{\bu \sim N(0, I_d)}\left[(\U_{\sqrt{x}} f(\bu))^4\right] \GE \left(\E_{\bu \sim N(0, I_d)}\left[(\U_{\sqrt{x}} f(\bu))^2\right]\right)^2 \EQ F_2[f](x)^2\;. \qedhere
\]
}
\end{proof}

\begin{lemma}\label{NAE35ratiolemma}
\change{Let $\Phi$ be a MAX NAE-$\{3,5\}$-SAT instance with completeness 1, and $\mathcal{A}$ a distribution of assignments to $\Phi$.} If the following conditions hold \change{for some $F_2, F_4 \in [0, 1]$,} 
\begin{enumerate}
\item The \change{expected fraction of $3$-clauses satisfied by $\mathcal{A}$ is at most} $\frac{3+3F_2}{4}$,
\item The \change{expected fraction of $5$-clauses satisfied by $\mathcal{A}$ is at most} $\frac{15 - 6F_2 - F_4}{16}$,
\item $F_4 \geq F_2^2$,
\end{enumerate}
then \change{by possibly re-weighting the clauses in $\Phi$ we can obtain another instance $\Phi'$ with completeness 1} such that the \change{expected fraction of clauses satisfied by $\mathcal{A}$} on $\change{\Phi'}$ is at most $\frac{3(\sqrt{21}-4)}{2} < 0.8739$.
\end{lemma}
\begin{proof}
\change{Let $p \in [0, 1]$ be some parameter to be chosen later. We construct $\Phi'$ by taking the distribution where we choose a random $3$-clause from $\Phi$ with probability $1-p$ and choose a random $5$-clause from $\Phi$ with probability $p$ (here we think of weights on the clauses as probability weights).}
\change{Then the expected fraction of clauses in $\change{\Phi'}$ satisfied by $\mathcal{A}$ is at most} 
\[
(1-p)\frac{3 + 3F_2}{4} + p\frac{15 - 6F_2 - F_4}{16} \EQ \frac{12 + 3p + (12-18p)F_2 - pF_4}{16} \;.
\]
Since $F_4 \geq F_2^2$, this is at most
\begin{align*}
\frac{12 + 3p + (12-18p)F_2 - pF^2_2}{16} &\EQ \frac{12 + 3p + \frac{(6-9p)^2}{p} - p\left(F_2 - \frac{(6-9p)}{p}\right)^2}{16} \\
&\LE \frac{84p + \frac{36}{p}}{16} - 6 \;.
\end{align*}
Taking the derivative with respect to $p$, this is minimized when $\frac{1}{16}(84-\frac{36}{p^2}) = 0$, which happens when $p = \frac{3}{\sqrt{21}}$. When $p = \frac{3}{\sqrt{21}}$,
\[
\frac{84p + \frac{36}{p}}{16} - 6 \EQ \frac{12\sqrt{21} + 12\sqrt{21}}{16} - 6 \EQ \frac{3(\sqrt{21} - 4)}{2} \;,
\]
and this completes the proof.
\end{proof}

\begin{theorem}\label{theorem:non_monotone_35_sat}
A $0.8739$-approximation for \emph{MAX NAE-$\{3,5\}$-SAT} (clauses of size $3$ and $5$) is Unique Games hard, even when the instance has completeness $1-\eps$, for $\eps > 0$ arbitrarily small.
\end{theorem}

\begin{proof}
\change{Let $\Phi$ be an instance that satisfies the conditions in Lemma~\ref{lem:phi_completeness}. Note that such an instance always exists, since we can take one 3-clause and one 5-clause on disjoint variables. We first analyze how an arbitrary $Round_f$ scheme performs on the SDP solution described in Lemma~\ref{lem:phi_completeness}.}  Recall that we have the Fourier expansions:
\begin{align*}
    & \NAE_3(x_1, x_2, x_3) \EQ \frac{3 - x_1x_2 - x_1x_3 - x_2x_3}{4}\;, \\
    & \NAE_5(x_1, x_2, x_3, x_4, x_5) \EQ \frac{15 - \sum_{1 \leq i < j \leq 5}x_ix_j - \sum_{1 \leq i < j < k < l \leq 5}x_ix_jx_kx_l}{16}\;.
\end{align*}
\change{Let $f$ be any odd rounding function, and let $F_2(x)=F_2[f](x)$ and $F_4(x)=F_4[f](x)$.} \strikeout{Using this, we make the following observations:} \change{We now have:}
\begin{enumerate}
    \item If we have a $3$-clause $\NAE_3(x_1,x_2,x_3)$ where $b_{1,2} = b_{1,3} = b_{2,3} = -\frac{1}{3}$ then 
    \[
    \E[\NAE_3(X_1,X_2,X_3)] \EQ \change{\frac{3-3F_2(-\frac{1}{3})}{4} \EQ } \frac{3+3F_2(\frac{1}{3})}{4}\;.
    \]
    \change{In the second equality we used the fact that $f$ is odd and Proposition~\ref{prop:F2power}.}
    \item If we have a $5$-clause $\NAE_5(x_1,x_2,x_3,x_4,x_5)$ where $b_{1,2} = b_{1,3} = b_{1,4} = b_{2,3} = b_{2,4} = b_{3,4} = \frac{1}{3}$ and $b_{1,5} = b_{1,5} = b_{2,5} = b_{3,5} = b_{4,5} = 0$ then 
    \[
    \E[\NAE_5(X_1,X_2,X_3,X_4,X_5)] \EQ \frac{15-6F_2(\frac{1}{3}) - F_4(\frac{1}{3},\frac{1}{3},\frac{1}{3},\frac{1}{3},\frac{1}{3},\frac{1}{3})}{16} \;.
    \]
    \change{Here, all moments that contain the 5th variable evaluate to 0 due to Lemma~\ref{lemma:moment_orthogonal}.}
\end{enumerate}
\change{We can now apply Lemma~\ref{NAE35ratiolemma}, with $F_2 = F_2(1/3), F_4 = F_4(1/3)$, $\mathcal{A}$ being the distribution of assignments induced by $Round_f$, to obtain another instance $\Phi'$ such that the expected satisfied fraction by $Round_f$ on $\change{\Phi'}$ is at most $\frac{3(\sqrt{21}-4)}{2} < 0.8739$.} \strikeout{The theorem now follows from the following lemma.}
\change{Since $f$ is an arbitrary odd function, the theorem now follows from Lemma~\ref{lem:rpr2_ughardness} and Lemma~\ref{lemma:odd_sufficient}.}
\end{proof}

\subsection{Extension to Monotone MAX NAE-SAT}\label{sec:extend_to_monotone}

In this subsection, we \strikeout{briefly remark that } \change{extend} the previous upper bound argument \strikeout{holds also for} \change{to} monotone MAX NAE-SAT, i.e., the version where we do not allow negated variables in any of the clauses. For monotone MAX NAE-SAT, the only difference is that now the rounding function $f$ does not need to be odd. 

\change{
We will use the same distribution of configurations as before. It is easy to see that the completeness stays the same. Given any rounding function $f$, for the performance of $Round_f$ on the 3-clause, we now have
\begin{equation}\label{eq:monotone_3}
\frac{3 - 3F_2[f]\left(-\frac{1}{3}\right)}{4}. 
\end{equation}
For the 5-clause, we now have
\begin{equation}\label{eq:monotone_5}
\frac{15 - 6F_2[f]\left(\frac{1}{3}\right) - 4F_2[f](0) - F_4[f]\left(\frac{1}{3}\right) - 4F_4[f]\left(\frac{1}{3}, \frac{1}{3}, \frac{1}{3}, 0, 0, 0\right)}{16}
\end{equation}
We will show that both (\ref{eq:monotone_3}) and (\ref{eq:monotone_5}) increase if we replace $f$ with its odd part $f^{\odd}$, defined by $\bv \mapsto \frac{f(\bv) - f(-\bv)}{2}$. This is spelled out in the following lemma. 
\begin{lemma}\label{lem:monotono_nae_diff}   
    Let $f: \mathbb{R}^d \to [-1, 1]$ and $\rho \in [0, 1]$. Then the following inequalities hold:
    \begin{itemize}
        \item[(a)] $\frac{3 - 3F_2[f]\left(-\rho\right)}{4} \leq \frac{3 - 3F_2[f^\odd]\left(-\rho\right)}{4}$. 
        \item[(b)] $\frac{15 - 6F_2[f]\left(\rho\right) - 4F_2[f](0) - F_4[f](\rho) - 4F_4[f](\rho, \rho, \rho, 0, 0, 0)}{16} \leq \frac{15 - 6F_2[f^\odd]\left(\rho\right) - F_4[f^\odd](\rho)}{16}$.
    \end{itemize}
\end{lemma}
The proof of Lemma~\ref{lem:monotono_nae_diff} is deferred to Appendix~\ref{app:proof_monotone}. Using Lemma~\ref{lem:monotono_nae_diff}, we can easily derive the same  bound for monotone MAX NAE-SAT.
}

\begin{theorem}
Assuming UGC, it is NP-hard to approximate monotone MAX NAE-SAT within an approximation ratio of $\frac{3\sqrt{21} - 12}{2} \approx 0.8739$.
\end{theorem}
\begin{proof}
\change{
Fix an arbitrary rounding function $f$. Using the same construction as in Theorem~\ref{theorem:non_monotone_35_sat}, we know by Lemma~\ref{lem:monotono_nae_diff} that $Round_f$ satisfies at most $\frac{3 - 3F_2[f^\odd]\left(-1/3\right)}{4} = \frac{3 + 3F_2[f^\odd]\left(1/3\right)}{4} $ of the 3-clauses and at most $\frac{15 - 6F_2[f^\odd]\left(1/3\right) - F_4[f^\odd](1/3)}{16}$ of the 5-clauses. Now we can apply Lemma~\ref{NAE35ratiolemma} with $F_2[f^\odd]\left(1/3\right)$ and $F_4[f^\odd]\left(1/3\right)$, and it immediately follows that $Round_f$ achieves at most $\frac{3(\sqrt{21}-4)}{2} < 0.8739$ on this distribution. Since $f$ is arbitrary, we obtain the desired claim by applying Lemma~\ref{lem:rpr2_ughardness}.
}
\end{proof}

\subsection{Explicit integrality gap}\label{sub-explicit}

In this subsection, we explicitly construct a family of gap instances \change{for MAX NAE-SAT} whose integrality ratio tends to $\frac{3(\sqrt{21}-4)}{2} < 0.8739$. \change{Note that in the proof of Lemma~\ref{lem:rpr2_ughardness} an $\epsilon$-net of the unit ball is used as the variable set, but here we will provide a much simpler construction. } Let $\{\be_i\mid i \in [n]\}$ be the canonical basis of $\mathbb{R}^n$. Consider the subset of $\mathbb{R}^n$ containing vectors that have exactly three nonzero coordinates, each being $1/\sqrt{3}$ or $-1/\sqrt{3}$, namely
\[ \change{V_n} \EQ \left\{ \left. \frac{b_1\be_i + b_2\be_j + b_3\be_k}{\sqrt{3}} \;\right|\; b_1, b_2, b_3 \in \{-1, 1\} , 1\le i<j<k\le n \right\}\;.\]

To every $\bv \in \change{V_n} $, we assign a Boolean variable $x_\bv \in \{-1, 1\}$ such that $x_{-\bv} = -x_\bv$. Our goal is to define a \strikeout{CSP} \change{MAX NAE-SAT} instance $\Phi$ with variables $x_\bv$ such that assigning $\bv$ to $x_\bv$ is an SDP solution with perfect completeness, while any integral solution has value at most $\frac{3(\sqrt{21} - 4)}{2}$ as $n$ tends to infinity.

\begin{definition} We define $\mathcal{C}_3$ to be the set of $3$-clauses of the form $\NAE(x_{{\bv}_1},x_{{\bv}_2},x_{{\bv}_3})$ where 
\begin{enumerate}
    \item ${\bv_1} = \frac{1}{\sqrt{3}}({s_1}\be_{i_1} - {s_2}\be_{i_2} + {s_4}\be_{i_4})$
    \item ${\bv_2} = \frac{1}{\sqrt{3}}({s_2}\be_{i_2} - {s_3}\be_{i_3} + {s_5}\be_{i_5})$
    \item ${\bv_3} = \frac{1}{\sqrt{3}}({s_3}\be_{i_3} - {s_1}\be_{i_1} + {s_6}\be_{i_6})$
\end{enumerate}
for some distinct indices $i_1,\ldots,i_6 \in [n]$ and signs $s_1,\ldots,s_6 \in \{-1,1\}$.
\end{definition}
\begin{definition}
We define $\mathcal{C}_5$ to be the set of $5$-clauses of the form $\NAE(x_{{\bv}_1},x_{{\bv}_2},x_{{\bv}_3},x_{{\bv}_4},x_{{\bv}_5})$ where 
\begin{enumerate}
    \item For all $j \in \{1,2,3,4\}$, ${\bv_j} = \frac{1}{\sqrt{3}}({s_1}\be_{i_1} + {s_{2j}}\be_{i_{2j}} + {s_{2j+1}}\be_{i_{2j+1}})$
    \item ${\bv_5} = \frac{1}{\sqrt{3}}({s_{10}}\be_{i_{10}} + {s_{11}}\be_{i_{11}} + {s_{12}}\be_{i_{12}})$
\end{enumerate}
for some distinct indices $i_1,\ldots,i_{12} \in [n]$ and signs $s_1,\ldots,s_{12} \in \{-1,1\}$.
\end{definition}
\begin{remark}
These sets of clauses are designed so that the pairwise biases for the $3$-clauses are $(-\frac{1}{3},-\frac{1}{3},-\frac{1}{3})$ and the pairwise biases for the $5$-clauses are $(\frac{1}{3},\frac{1}{3},\frac{1}{3},\frac{1}{3},\frac{1}{3},\frac{1}{3},0,0,0,0)$.
\end{remark}
\begin{definition}
Let \change{$\Phi_n$} be the {\rm MAX NAE-$\{3,5\}$-SAT}-instance with variable set $\{X_\bv \mid \bv \in \change{V_n} \}$ and clause set $\mathcal{C}_3 \cup \mathcal{C}_5$, where every clause in $\mathcal{C}_3$ has weight $\frac{1-\frac{3}{\sqrt{21}}}{|\mathcal{C}_3|}$ and every clause in $\mathcal{C}_5$ has weight $\frac{3}{\sqrt{21}|\mathcal{C}_5|}$.
\end{definition}
\begin{theorem}
For any integral solution to \change{$\Phi_n$}, the weight of the satisfied clauses is at most $\frac{3(\sqrt{21}-4)}{2} + O(\frac{1}{n})$.
\end{theorem}

\change{Our strategy for proving this theorem is to define analogs of moment functions for assignments to this particular instance, so that the observations from the previous sections still apply and the same argument goes through. A detailed proof can be found in Appendix~\ref{appendix:explicit_gap}. We remark that by dropping the requirement that $x_{-\bv} = -x_\bv$ the same construction can be used to give a gap instance for monotone MAX NAE-SAT as well.}

\section{Calculus of Variations for MAX CUT and MAX NAE-\texorpdfstring{$\{3\}$}{3}-SAT}\label{S-COV}

\change{In this section we uses calculus of variations to show that the optimal RPR$^2$ rounding function for MAX NAE-$\{3\}$-SAT satisfies a \emph{Fredholm integral equation of the second kind}.}
To help illustrate our techniques, we also perform this analysis for the simpler case of MAX-CUT, giving an alternative description for the optimal rounding function found by O'Donnell and Wu~\cite{OW08}.

\subsection{Locally Optimal Distributions}
Recall from Section~\ref{sub-moment} that the expected value of an $\NAE_2$ clause with pairwise bias $b_{1,2}$ is $\frac{1 - F_2(b_{1,2})}{2}$ and that of an $\NAE_3$ clause with pairwise biases $b_{1,2}, b_{1,3}, b_{2,3}$ is \strikeout{given by} $\frac{3 - F_2(b_{1,2}) - F_2(b_{1,3}) - F_2(b_{2,3})}{4}$. In this subsection, we use the fact that $F_2$ is odd and convex on $[0, 1]$ to derive the hardest distribution for MAX NAE-$\{3\}$-SAT. The idea is to maximize the sum of $F_2$ while preserving the sum of the pairwise biases. A similar analysis was done for MAX CUT by O'Donnell and Wu~\cite{OW08}. We include this analysis here for completeness and to help illustrate the techniques.
\subsubsection{Facts about odd functions which are convex for \texorpdfstring{$x \geq 0$}{x >= 0}}
For our analysis, we need some facts about odd functions which are convex for $x \geq 0$.
\begin{proposition}\label{averagingprop}
If $F_2(x)$ is convex for $x \geq 0$, $0 \leq x_1 \leq x_2 \leq x_3 \leq x_4$, and $x_2 + x_3 = x_1 + x_4$ then $F_2(x_2) + F_2(x_3) \leq F_2(x_1) + F_2(x_4)$.
\end{proposition}
\begin{proof}
If $x_1 = x_2 = x_3 = x_4$ then the result is trivial so we can assume that $x_1 < x_4$. Writing $x_2 = a{x_1} + (1-a){x_4}$ and $x_3 = b{x_1} + (1-b){x_4}$ where $a,b \in [0,1]$, we have that 
\[
x_2 + x_3 \EQ (a+b)x_1 + (2 - a - b)x_4 \EQ x_1 + x_4 + (1 - a - b)(x_4 - x_1) \;.
\] 
Since $x_2 + x_3 = x_1 + x_4$, we must have that $b = 1-a$. Since $F_2$ is convex we have that 
\begin{enumerate}
\item $F_2(x_2) \leq aF_2(x_1) + (1-a)F_2(x_4)$.
\item $F_{\change{2}}(x_{\change{3}}) \leq bF_2(x_1) + (1-b)F_2(x_4) = (1-a)F_2(x_1) + aF_2(x_4)$.
\end{enumerate}
Adding these inequalities together we have that $F_2(x_2) + F_2(x_3) \leq F_2(x_1) + F_2(x_4)$, as needed.
\end{proof}
Our main tool is the following lemma:
\begin{lemma}\label{keymodifyinglemma}
If $F_2(x)$ is an odd function which is convex for $x \geq 0$ then 
\begin{enumerate}
\item For all $x,x'$ such that $-x \geq |x'|$, $2F_2\left(\frac{x + x'}{2}\right) \geq  F_2(x) + F_2(x')$.
\item For all $x,x'$ such that $x \geq |x'|$ and all $y \geq 0$, $F_2(x+y) + F_2(x' - y) \geq F_2(x) + F_2(x')$.
\end{enumerate}
\end{lemma}
\begin{proof}
If $-x \geq |x'|$ then there are two cases to consider:
\begin{enumerate}
\item If $x' \leq 0$ then since $F_2(x)$ is convex for $x \geq 0$, $2F_2\left(-\frac{x + x'}{2}\right) \leq  F_2(-x) + F_2(-x')$. Using the fact that $F_2(-x) = -F_2(x)$ and rearranging, $2F_2\left(\frac{x + x'}{2}\right) \geq  F_2(x) + F_2(x')$.
\item If $x' \geq 0$ then since $F_2(x)$ is convex for $x \geq 0$, $F_2(-(x+x')) + F_2(x') \leq F_2(-x) + F_2(0)$ and $2F_2\left(-\frac{x + x'}{2}\right) \leq  F_2(-(x+x')) + F_2(0)$. Adding these two inequalities together, using the fact that $F_2(-x) = -F_2(x)$, and rearranging, $2F_2\left(\frac{x + x'}{2}\right) \geq  F_2(x) + F_2(x')$.
\end{enumerate}
Similarly, if $|x'| \leq x $ then there are three cases to consider:
\begin{enumerate}
\item If $x' \leq 0$ then since $F_2(x)$ is convex for $x \geq 0$, $F_2(x) + F_2(y-x') \leq F_2(-x') + F_2(x+y)$. Rearranging and using the fact that $F_2(-x) = -F_2(x)$, $F_2(x+y) + F_2(x' - y) \geq F_2(x) + F_2(x')$.
\item If $0 \leq x' \leq y$ then since $F_2(x)$ is convex for $x \geq 0$, $F_2(x) + F_2(x') \leq F_2(0) + F_2(x + x')$ and $F_2(x + x') + F_2(y-x') \leq F_2(x+y) + F_2(0)$. Adding these two inequalities together, using the fact that $F_2(-x) = -F_2(x)$, and rearranging, $F_2(x+y) + F_2(x' - y) \geq F_2(x) + F_2(x')$.
\item If $x' \geq y$ then since $F_2(x)$ is convex for $x \geq 0$, $F_2(x) + F_2(x') \leq F_2(x + y) + F_2(x'-y)$.
\end{enumerate}
\end{proof}

\subsubsection{The Hardest Distributions for MAX CUT/MAX-NAE-\{2\}-SAT}
\begin{theorem}[\cite{OW08}]\label{thm:maxcuthard}
For {\rm MAX CUT/MAX-NAE-\{2\}-SAT} with any completeness \change{and for any SDP rounding scheme}, the \change{distributions which minimize the approximation ratio} are supported on $\{-\rho,1\}$ for some $\rho \in [0,1]$.
\end{theorem}
\begin{proof}
To prove this theorem, we show that any distribution which is not of this form can be improved. For this, we use a sequence of lemmas.
\begin{lemma}\label{optimizingNAE2SATdist}
For {\rm MAX CUT/MAX-NAE-\{2\}-SAT}, 
\begin{enumerate}
\item If there are two constraints with pairwise biases $x$ and~$x'$ where $|x'| \leq -x \leq 1$ then replacing $x$ and~$x'$ with two copies of $\frac{x + x'}{2}$ does not affect the SDP value and can only decrease the performance of the rounding scheme.
\item If there are two constraints with pairwise biases $x$ and~$x'$ where $0 < |x'| \leq x \leq 1$ then replacing $x$ and~$x'$ with $x + x' - 1$ and $1$ does not affect the SDP value and can only decrease the performance of the rounding scheme.
\end{enumerate}
\end{lemma}
\begin{proof}
The SDP value \change{of $\frac{1}{2}\cdot\left(\frac{1-x}{2} + \frac{1-x'}{2}\right)$}is not affected by these adjustments. To show that the performance of the rounding scheme \change{of $\frac{1}{2}\cdot\left(\frac{1-F_2(x)}{2} + \frac{1-F_2(x')}{2}\right)$} can only decrease, we make the following observations:
\begin{enumerate}
\item The first statement follows directly from the first statement of Lemma \ref{keymodifyinglemma}.
\item The second statement follows from applying the second statement of Lemma \ref{keymodifyinglemma} to $x$ and~$x'$ with $y = 1 - x$. \qedhere
\end{enumerate}
\end{proof}
It is not hard to verify that the only distributions which cannot be improved using this lemma are distributions where the pairwise biases are supported on $\{-\rho,1\}$ for some $\rho \in [0,1]$. 
\end{proof}

\subsubsection{Possible Pairwise Biases}\label{sub-possible}
In order to analyze the hardest distribution of triples of pairwise biases for MAX-NAE-\{3\}-SAT, we need to consider which triples of pairwise biases are possible for a set of three $\pm{1}$ variables (regardless of any constraints).
\begin{lemma}[e.g., \cite{FG95}]
For distributions on three $\pm{1}$ variables $x_1,x_2,x_3$, the polytope of possible pairwise biases is given by the following inequalities:
\begin{enumerate}
\item $b_{1,2} + b_{1,3} + b_{2,3} \geq -1$
\item $b_{1,2} \geq b_{1,3} + b_{2,3} - 1$
\item $b_{1,3} \geq b_{1,2} + b_{2,3} - 1$
\item $b_{2,3} \geq b_{1,2} + b_{1,3} - 1$
\end{enumerate}
\end{lemma}
\begin{remark}
Observe that the inequalities $-1 \leq b_{1,2} \leq 1$, $-1 \leq b_{1,3} \leq 1$, and $-1 \leq b_{2,3} \leq 1$ are implied by these inequalities. To see this, note that adding the first two inequalities and dividing by $2$ gives the inequality $b_{1,2} \geq -1$. Adding the second and third inequality and dividing by $2$ gives the inequality $b_{2,3} \leq 1$. By symmetry, the other inequalities can be derived in a similar way.
\end{remark}
\begin{proof}
The possible pairwise biases for integral assignments are as follows.
\begin{enumerate}
\item $(+1,+1,+1)$
\item $(+1,-1,-1)$
\item $(-1,+1,-1)$
\item $(-1,-1,+1)$
\end{enumerate}
We want to describe the convex hull of these points. Note that these points satisfy all of these equalities, so any point in their convex hull satisfies these inequalities as well.

We now need to show that if these inqualities are satisfied then there is a distribution which matches these pairwise biases. To do this, assume that the inequalities are satisfied and take the distribution where 
\begin{enumerate}
\item $b_{1,2} = b_{1,3} = b_{2,3} = 1$ with probability $c_{+++} = \frac{b_{1,2} + b_{1,3} + b_{2,3} + 1}{4}$
\item $b_{1,2} = 1$ and $b_{1,3} = b_{2,3} = -1$ with probability $c_{+--} = \frac{b_{1,2} - b_{1,3} - b_{2,3} + 1}{4}$
\item $b_{1,3} = 1$ and $b_{1,2} = b_{2,3} = -1$ with probability $c_{-+-} = \frac{b_{1,3} - b_{1,2} - b_{2,3} + 1}{4}$
\item $b_{2,3} = 1$ and $b_{1,2} = b_{1,3} = -1$ with probability $c_{--+} = \frac{b_{2,3} - b_{1,2} - b_{1,3} + 1}{4}$
\end{enumerate}
Since the inequalities are satisfied, we have that $c_{+++}, c_{+--}, c_{-+-}, c_{--+} \in [0,1]$. Now observe that 
\begin{enumerate}
\item $c_{+++} + c_{+--} + c_{-+-} + c_{--+} = 1$
\item $\E[x_{1}x_{2}] = c_{+++} + c_{+--} - c_{-+-} - c_{--+} = b_{1,2}$
\item $\E[x_{1}x_{3}] = c_{+++} - c_{+--} + c_{-+-} - c_{--+} = b_{1,3}$
\item $\E[x_{2}x_{3}] = c_{+++} - c_{+--} - c_{-+-} + c_{--+} = b_{2,3}$
\end{enumerate}
so this distribution matches the given pairwise biases.
\end{proof}
\subsubsection{The Hardest Distributions for MAX-NAE-\{3\}-SAT}
\begin{theorem}\label{NAE3SATdisttheorem}
For {\rm MAX-NAE-\{3\}-SAT} with any completeness, the hardest distributions are of the following forms.
\begin{enumerate}
\item The distribution of the pairwise biases for the constraints is supported on \\
$\{(-\rho,-\rho,-\rho),(-\rho,-\rho,1)\}$ for some $\rho \in [0,\frac{1}{3}]$.
\item The distribution of the pairwise biases for the constraints is supported on \\
$\{(-\frac{1}{3},-\frac{1}{3},-\frac{1}{3}),(-\rho,-\rho,1)\}$ for some $\rho \in [\frac{1}{3},1]$.
\item The distribution of the pairwise biases for the constraints is supported on \\
$\{(-\rho,-\rho,1),(1,1,1)\}$ for some $\rho \in [\change{0},1]$.
\end{enumerate}
\end{theorem}
\begin{remark}
The reason why we cannot have triples of pairwise biases $(-\rho,-\rho,-\rho)$ where $\rho > \frac{1}{3}$ is because for any triple $(b_{1,2},b_{1,3},b_{2,3})$ of pairwise biases, $b_{1,2} + b_{1,3} + b_{2,3} \geq -1$.
\end{remark}
\begin{proof}
To prove this theorem, we show that any distribution which is not of this form can be improved. For this, we use the following lemmas.
\begin{lemma}
For a single {\rm MAX NAE-\{3\}-SAT} constraint, 
\begin{enumerate}
\item If the pairwise biases are $x,x',x''$ where $|x'| \leq -x \leq 1$ then replacing $x$ and~$x'$ with two copies of $\frac{x + x'}{2}$ does not affect the SDP value and can only decrease the performance of the rounding scheme.
\item If the pairwise biases are $x,x',x''$ where $\max{\{|x'|,|x''|\}} \leq x \leq 1$, and $x'' < x'$ then replacing $x$ and~$x'$ with $x + x' - x''$ and $x''$ does not affect the SDP value and can only decrease the performance of the rounding scheme. Observe that $x'' \geq x + x' - 1$ so $x + x' - x'' \leq 1$.
\item If the pairwise biases are $x,x',x''$ where $x'' = x'$ and $|x'| < x \leq 1$ then replacing $x,x',x''$ with $1$, $x' + \frac{x-1}{2}$, and $x'' + \frac{x-1}{2}$ does not affect the SDP value and can only decrease the performance of the rounding scheme. Observe that since $x > |x'|$, $x' + \frac{x-1}{2} = \frac{x + x'}{2} + \frac{x'-1}{2} > \frac{x'-1}{2} \geq -1$.
\end{enumerate}
\end{lemma}
\change{\begin{proof}
The SDP value of $\frac{3-x-x'-x''}{4}$ is not affected by these adjustments. To show that the performance of the rounding scheme of $\frac{3-F_2(x) -F_2(x')-F_2(x'')}{4}$ can only decrease, we make the following observations:
\begin{enumerate}
\item The first statement follows directly from the first statement of Lemma \ref{keymodifyinglemma}.
\item The second statement follows from applying the second statement of Lemma \ref{keymodifyinglemma} to $x$ and~$x'$ with $y = x' - x''$.
\item The third statement follows by applying the second statement of Lemma \ref{keymodifyinglemma} twice. First we apply it to $x$ and $x'$ with $y = \frac{1-x}{2}$. We then apply it a second time with $\frac{x+1}{2}$ and $x''$ with $y = \frac{1-x}{2}$.
\qedhere
\end{enumerate}
\end{proof}}

\change{\begin{remark}\label{makingprogressremark}
Technically, we need to make sure that we are making progress when we apply Lemma \ref{keymodifyinglemma}. To ensure this, we apply Lemma \ref{keymodifyinglemma} so that 
\begin{enumerate}
\item All $x$ stay in the range $[-1,1]$.
\item We never change the value of any $x$ to $-1$ unless it was $-1$ already.
\item Whenever we apply the second statement (or more precisely, make a sequence of up to $3$ applications of the second statement), we change at least one $x$ which is less than $1$ to $1$.
\end{enumerate}
Under these conditions, our applications of Lemma \ref{keymodifyinglemma} reduce $\sum_{x}{h(x)}$ for the following potential function~$h(x)$:
\[
h(x) \EQ \frac{x^2}{8} - 1_{x = 1} + 1_{x = -1} \;.
\]
This implies that we make progress and can converge to a distribution where Lemma \ref{keymodifyinglemma} cannot be applied.
\end{remark}}
\begin{remark}
Whenever case $2$ holds and $x + x' - x'' < 1$, we immediately apply case $3$ on $(x + x' - x'',x'',x'')$. This sequence changes $x$ to $1$, satisfying the conditions of Remark \ref{makingprogressremark}.
\end{remark}
\begin{corollary}
The hardest distributions have triples of pairwise biases of the form $(x,x,x)$ where $x \in [-\frac{1}{3},0]$ or of the form $(x',x',1)$ where $x' \in [-1,1]$.
\end{corollary}
In order to prove Theorem \ref{NAE3SATdisttheorem}, we need to show that there can only be one value of $x$ and only one value of~$x'$ which is not equal to $1$. For this, we use the following lemma.
\begin{lemma}
For {\rm MAX-NAE-\{3\}-SAT}, 
\begin{enumerate}
\item If there are two constraints with pairwise biases $(x,x,x)$ and $(x',x',x')$ where $x,x' \in [-\frac{1}{3},0]$ then replacing these triples of pairwise biases with two copies of $(\frac{x + x'}{2},\frac{x + x'}{2},\frac{x + x'}{2})$ does not affect the SDP value and can only decrease the performance of the rounding scheme.
\item If there are two constraints with pairwise biases $(x,x,1)$ and $(x',x',1)$ where $|x'| \leq -x \leq 1$ then replacing these triples of pairwise biases with two copies of $(\frac{x + x'}{2},\frac{x + x'}{2},1)$ does not affect the SDP value and can only decrease the performance of the rounding scheme.
\item If there are two constraints with pairwise biases $(x,x,1)$ and $(x',x',1)$ where $|x'| \change{\ \le\ } x \leq 1$ then replacing these triples of pairwise biases with $(1,1,1)$ and $(x + x' - 1,x + x' - 1,1)$ does not affect the SDP value and can only decrease the performance of the rounding scheme.\qedhere
\end{enumerate}
\end{lemma}
\begin{proof}
Again, the SDP value is linear in $x$ and~$x'$ so it is not affected by these adjustments. To show that the performance of the rounding scheme can only decrease, we make the following observations:
\begin{enumerate}
\item The first and second statements follow directly from the first statement of Lemma \ref{keymodifyinglemma}.
\item The third statement follows from applying the second statement of Lemma \ref{keymodifyinglemma} to $x$ and~$x'$ with $y = 1-x$. \qedhere
\end{enumerate}
\end{proof}
This implies that we can take the triples of pairwise biases to be supported on \\
$\{(x,x,x),(x',x',1),(1,1,1)\}$ for some $x \in [-\frac{1}{3},0]$ and some $x' \in [-1,\change{0}]$. We now show that we can either take $x = x' \in [-\frac{1}{3},0]$ or take $x = -\frac{1}{3}$ and $x' \in [-1,-\frac{1}{3}]$. To show this, we use the following lemma.
\begin{lemma}
For {\rm MAX-NAE-\{3\}-SAT}, 
\begin{enumerate}
\item If there are two constraints \change{of equal weight} with pairwise biases $(x,x,x)$ and $(x',x',1)$ where $x \in [-\frac{1}{3},0]$, $x' \in [-1,0]$, and $\frac{3x + 2x'}{5} \geq -\frac{1}{3}$ then replacing these triples of pairwise biases with \\
$(\frac{3x + 2x'}{5},\frac{3x + 2x'}{5},\frac{3x + 2x'}{5})$ and $(\frac{3x + 2x'}{5},\frac{3x + 2x'}{5},1)$ does not affect the SDP value and can only decrease the performance of the rounding scheme.
\item If there are two constraints \change{of equal weight} with pairwise biases $(x,x,x)$ and $(x',x',1)$ where $x \in [-\frac{1}{3},0]$, $x' \in [-1,-\frac{1}{3}]$, and $\frac{3x + 2x'}{5} \leq -\frac{1}{3}$ then replacing these triples of pairwise biases with 
$(-\frac{1}{3},-\frac{1}{3},-\frac{1}{3})$ and $(x'+\frac{3x + 1}{2},x' + \frac{3x + 1}{2},1)$ does not affect the SDP value and can only decrease the performance of the rounding scheme.
\end{enumerate}
\end{lemma}
\begin{proof}[Proof sketch]
\strikeout{This can be shown using the fact that $F_2$ is an odd function and $F_2(x)$ is convex for $x \geq 0$.}
\change{
Recall that an $\NAE_3$ clause with pairwise biases $b_{1,2}, b_{1,3}, b_{2,3}$ has an SDP value of $\frac{3-b_{1,2}-b_{1,3}-b_{2,3}}{4}$ is satisfied by the rounding scheme with a probability of $\frac{3 - F_2(b_{1,2}) - F_2(b_{1,3}) - F_2(b_{2,3})}{4}$.
\begin{itemize}
\item It is straightforward to verify that $(\frac{3x + 2x'}{5},\frac{3x + 2x'}{5},\frac{3x + 2x'}{5})$ and $(\frac{3x + 2x'}{5},\frac{3x + 2x'}{5},1)$ satisfy the triangle inequalities and that the SDP value does not change:
\[
    \frac{1}{2}\cdot \frac{3-3x}{4} + \frac{1}{2}\cdot \frac{3-2x'-1}{4} = \frac{1}{2} \cdot \frac{3 - 3(3x+2x')/5)}{4} + \frac{1}{2} \cdot \frac{3 - 2(3x+2x')/5-1}{4}.
\]
To verify that the rounding probability can only decrease, it suffices to verify that
\[
    \frac{1}{2}\cdot \frac{3-3F_2(x)}{4} + \frac{1}{2}\cdot \frac{3-2F_2(x') - F_2(1)}{4} \ge \frac{1}{2} \cdot \frac{3 - 3F_2((3x+2x')/5)}{4} + \frac{1}{2} \cdot \frac{3 - 2F_2((3x+2x')/5)-F_2(1)}{4}.
\]
This is equivalent to showing that $5F_2((3x+2x')/5) \ge 3F_2(x) + 2F_2(x'),$ which follows from the fact that $F_2$ is concave on $[-1, 0]$ (Proposition~\ref{prop:F2power}).
\item As in the previous part, it is straightforward to verify that $(-\frac{1}{3},-\frac{1}{3},-\frac{1}{3})$ and $(x'+\frac{3x + 1}{2},x' + \frac{3x + 1}{2},1)$ satisfy the triangle inequalities and that the SDP value does not change:
\[
    \frac{1}{2}\cdot \frac{3-3x}{4} + \frac{1}{2}\cdot \frac{3-2x'-1}{4} = \frac{1}{2} \cdot 1 + \frac{1}{2} \cdot \frac{3 - (2x' + 3x + 1) - 1}{4}.
\]
To verify that the rounding probability can only decrease, it suffices to verify that
\[
    \frac{1}{2}\cdot \frac{3-3F_2(x)}{4} + \frac{1}{2}\cdot \frac{3-2F_2(x')-F_2(1)}{4} \ge \frac{1}{2} \cdot \frac{3 - 3F_2(-1/3)}{4} + \frac{1}{2} \cdot \frac{3 - 2F_2((2x' + 3x + 1)/2) - F_2(1)}{4}.
\]
This is equivalent to showing that $3F_2(-1/3) + 2F_2((1+3x+2x')/2) \ge 3F_2(x) + 2F_2(x'),$ which follows from the fact that $F_2$ is concave on $[-1, 0]$ (Proposition~\ref{prop:F2power}). \qedhere
\end{itemize}
}
\end{proof}
Finally, we observe that if we have both $(x,x,x)$ and $(1,1,1)$ then we can get rid of one of them. To see this, observe that as far as the SDP value and the performance of the rounding scheme are concerned, having constraints with pairwise biases $(x,x,x)$ and $(1,1,1)$ with weights $\frac{2}{3}$ and $\frac{1}{3}$ respectively is the same as having a constraint with pairwise biases $(x,x,1)$. Thus, if we have both $(x,x,x)$ and $(1,1,1)$ we can convert them into $(x,x,1)$ until one of them is exhausted.
\end{proof}

\subsubsection{Optimality of One-dimensional Rounding Function}
Having derived the hardest distribution for MAX NAE-$\{3\}$-SAT, we conclude this subsection with the following optimality result.

\begin{definition}
Let  $f : \mathbb R^k \to [-1, 1]$ be an RPR$^2_k$ rounding function. 
\begin{itemize}
    \item For a distribution $\scD$ of biases, define $s_2(f, \scD) :=  \frac{1-\underset{b \sim \scD}{\E}[F_2[f](b)]}{2}$.
    \item For a distribution $\scD$ of triples of biases, define $s_3(f, \scD) :=  \frac{3-\underset{(b_{1,2}, b_{1,2}, b_{2,3}) \sim \scD}{\E}[F_2[f](b_{1,2}) + F_2[f](b_{1,3}) + F_2[f](b_{2,3})]}{4}$.
\end{itemize}
\end{definition}

\begin{claim}[\cite{OW08}]\label{claim:nae3-one}
For any $f: \mathbb R^k \to [-1, 1]$, there exists a monotone, odd, one-dimensional function $g : \mathbb R \to [-1, 1]$ such that 
\begin{itemize}
    \item $s_2(g, \scD_{\hard,2}) \ge s_2(f, \scD_{\hard,2})$ for any distribution of biases $\scD_{\hard,2}$ described by Theorem~\ref{thm:maxcuthard}.
    \item $s_3(g, \scD_{\hard,3}) \ge s_3(f, \scD_{\hard,3})$ for any distribution of triples of biases $\scD_{\hard,3}$ described by Theorem~\ref{NAE3SATdisttheorem}.
\end{itemize}
\end{claim}

\begin{proof}
If we apply Theorem~\ref{thm:gauss} and let $g = f^*$ be the Gaussian rearrangement of $f$, then $g$ is a monotone, odd, one-dimensional function. Furthermore, $F_{2}[g](1) = F_{2}[f](1)$ and $F_{2}[g](x) \le F_{2}[f](x)$ for all $x \in [-1, 0]$. By Theorem~\ref{thm:maxcuthard} and Theorem~\ref{NAE3SATdisttheorem}, any bias that comes up in a hardest distribution either for MAX CUT or MAX NAE-$\{3\}$-SAT is in $[-1, 0] \cup \change{\{1\}}$. Thus, $\E_{b \sim \scD_{\hard,2}}[F_{2}[g](b)] \leq \E_{b \sim \scD_{\hard,2}}[F_2[f](b)]$ and 
\begin{align*}
    & \E_{(x_{12}, x_{13}, x_{23}) \sim \scD_{\hard,3}}[F_{2}[g](x_{12})+F_{2}[g](x_{13})+F_{2}[g](x_{23})] \\
    \LE\; & \E_{(x_{12}, x_{13}, x_{23}) \sim \scD_{\hard,3}}[F_{2}[f](x_{12})+F_{2}[f](x_{13})+F_{2}[f](x_{23})] \;.
\end{align*}
Therefore, $s_2(g, \scD_{\hard,2}) \ge s_2(f, \scD_{\hard,2})$ and $s_3(g, \scD_{\hard,3}) \ge s_3(f, \scD_{\hard,3})$. 
\end{proof}

\subsection{A variational approach to MAX CUT} 

Let $f : \mathbb R\to [-1, 1]$ be an odd, monotone RPR$^2$ rounding function for MAX CUT. By Theorem~\ref{thm:maxcuthard}, for a given completeness $c \in [1/2, 1]$, we know the hardest distributions $\mathcal D$ are a combination of equal vectors and vectors with dot product $\rho$. Let $\alpha \in [0, 1]$ be the relative frequency of these two dot products, that is $c = \alpha \cdot \frac{1-\rho}{2}$. The performance of the rounding scheme will be

\begin{align*}
s_2(f, \scD) &\EQ (1-\alpha) \cdot \frac{1-F_2(1)}{2} + \alpha \cdot \frac{1 - F_2(\rho)}{2}\\
&\EQ \frac{1}{2} - \frac{1-\alpha}{2}\int_{\mathbb R}f(x)^2\phi(x)\,dx - \frac{\alpha}{2} \int_{\mathbb R^2}f(x)f(y) \phi_{\rho}(x, y)\,dx dy\;,
\end{align*}
where $\phi(x)=\frac{1}{\sqrt{2\pi}}\exp(-x^2/2)$ is the density of a standard normal variable, and 
\begin{align*}
\phi_{\rho}(x, y) &\EQ \frac{\exp\left(-\frac{1}{2}(x, y)^T\Sigma^{-1}(x, y)\right)}{2\pi \sqrt{1-\rho^2}} 
\EQ
\frac{\exp
\left(-\frac{x^2-2\rho x y + y^2}{2(1-\rho^2)}\right)
}{
2\pi\sqrt{1-\rho^2}}
\end{align*}
is the density function of a two-dimensional normal random variable with mean $\mu=(0,0)$  covariance matrix  
$
\Sigma =\begin{pmatrix}1&\rho\\\rho&1\end{pmatrix}$ so that $\Sigma^{-1}=\frac{1}{1-\rho^2}
\begin{pmatrix}1&-\rho\\-\rho&1\end{pmatrix}$.

Thus, maximizing $s_2(f, \scD)$ is equivalent to minimizing
\[ L(f) \EQ \frac{1-\alpha}{\alpha} \int_{-\infty}^\infty f(x)^2\phi(x)dx + \int_{\mathbb R^2}f(x)f(y) \phi_{\rho}(x, y)\,dx dy \;.
\]
\strikeout{We know that $f$ is monotone and that it attains the values $-1$ and $1$ when $|x|$ is sufficiently large.} Let $a := \sup \{x : f(x) < 1\}$. Since $f$ is increasing and odd, we know that $a \ge 0$ and that $|f(x)| < 1$ on $(-a, a)$. Note that it may be the case that $a = \infty$.

We prove in Appendix~\ref{app:minimizer} that there is an `optimal' $f$. That is, there exists a \change{function} $f$ which globally minimizes $L(f)$ among all valid RPR$^2$ rounding functions. This means that any adjustment to $f$, not violating the condition $\|f\|_{\infty} \le 1$ increases the functional. More precisely, consider any $a' \in (0, a)$ and any measurable $h : (-a', a') \to [-1, 1]$ (not necessarily odd nor monotone). Note that for some $\eps$ sufficiently small, for all $s \in [-\eps, \eps]$, $\|f+s h\|_{\infty} \le 1$. Then, in order for there to be a local minimum of $L$ at $s = 0$, we must have that
\begin{align*}
0 &\EQ \frac{1}{2}\left.\frac{\partial}{\partial s}L(f + sh)\right|_{s=0}\\
&\EQ \frac{1-\alpha}{\alpha}\int_{-a'}^{a'} f(y)h(y)\phi(y)\,dy +
\int_{-\infty}^\infty \int_{-a'}^{a'} f(x)h(y) \phi_{\rho}(x,y)\,dx dy\\
&\EQ \int_{-a'}^{a'} h(y) \left( 
\frac{1-\alpha}{\alpha}f(y)\phi(y) +
\int_{-\infty}^\infty f(x) \phi_{\rho}(x,y)\,dx
\right)\,dy \;.
\end{align*}
Since $h$ is an arbitrary measurable function on $(-a', a')$, we have that for almost every $y \in (-a', a')$,
\begin{align}
\frac{1-\alpha}{\alpha}f(y)\phi(y) +
\int_{-\infty}^\infty f(x) \phi_{\rho}(x,y)\,dx \EQ 0\;.\label{eq:fred-maxcut}
\end{align}
Since $a' \in (0, a)$. The above holds for all $\change{y} \in (-a, a)$. Further define
\[K(x, y)\;:=\;\frac{\phi_\rho(x,y)}{\phi(y)} \;=\; \frac{ \exp\left(-\frac{(x-\rho y)^2}{2(1-\rho^2)}\right)
}
{\sqrt{2\pi(1-\rho^2)}} \;. \]
and
\[ g(y) \;:=\; \frac{\alpha}{1-\alpha} \left(\int_{-\infty}^{-a} K(x,y)\,dx - \int_{a}^{-\infty} K(x,y)\,dx \right) \;. \]

Then, (\ref{eq:fred-maxcut}) becomes
\begin{align}
f(y) +
\frac{\alpha}{1-\alpha}\int_{-a}^a f(x) K(x,y)\,dx \;=\; g(y) \quad,\quad -a\le y\le a \;.\label{eq:fredholm-maxcut}
\end{align}

This is a Fredholm integral equation of the second kind (see~\cite{polyanin2008handbook}).

\subsubsection{Comparing with Feige and Langberg}

Let $s_2(f, \scD)|y$ be the expected size of the cut produced by using the rounding \change{function}~$f$ conditioned on the value of~$y$, i.e., 
\[ s_2(f, \scD)|y \EQ \frac{1}{2} - \frac{1-\alpha}{2} f(y)^2 -\frac{\alpha}{2}\int_{-\infty}^\infty f(x)f(y) \frac{\phi_\rho(x,y)}{\phi(y)}\,dx\;.   \]

Feige and Langberg \cite{FL06} argue that if $f$ is optimal, then for every $y$ we have $s_2(f, \scD)|y\ge\frac12$, and if $-1<f(y)<1$, then $s_2(f, \scD)|y=\frac12$. Their intuitive argument is that if $s_2(f, \scD)|y<\frac12$ then the cut produced by the rounding procedure is not \emph{locally optimal} in expectation and can thus be improved.

If $-1<f(y)<1$ then $s_2(f, \scD)|y=\frac12$ which, by moving sides and dividing by $f(y)$, is equivalent to 
\[ \frac{1-\alpha}{\alpha}f(y) +
\int_{-\infty}^\infty f(x) \frac{\phi_\rho(x,y)}{\phi(y)}\,dx \EQ 0 \;.
\]
This is exactly the integral equation we got above.

\subsection{Analogous Condition for MAX NAE-\texorpdfstring{$\{3\}$}{3}-SAT}

Let $f : \mathbb R\to [-1, 1]$ be an odd, increasing RPR$^2$ rounding function for MAX NAE-$\{3\}$-SAT. By Theorem~\ref{NAE3SATdisttheorem}, for a given completeness $c \in [3/4, 1]$, we also know the hardest distributions $\scD$ have $(\rho_0, \rho_0, \rho_0)$ with \change{probability} $\alpha$ and $(1, \rho, \rho)$ with probability $1-\alpha$, where $\rho_0$ is either $\max(\rho, -\third)$ or $\rho_0 = 1$. 

We first consider the case $\rho_0 = \max(\rho, -\third)$. The performance of the rounding scheme is

\begin{align*}
s_3(f, \scD) &\EQ (1-\alpha) \cdot \frac{3-F_2(1)-2F_2(\rho)}{4} + \alpha \cdot \frac{3 - 3F_2(\rho_0)}{4}\\
&\EQ \frac{3}{4} - \frac{1-\alpha}{4}\int_{\mathbb R}f(x)^2\phi(x)\,dx- \frac{1-\alpha}{2}\int_{\mathbb R^2}f(x)f(y)\phi_{\rho}(x, y)\,dx dy - \frac{3\alpha}{4} \int_{\mathbb R^2}f(x)f(y) \phi_{\rho_0}(x, y)\,dx dy\;.
\end{align*}

Thus, maximizing $s_3(f, \scD)$ is equivalent to minimizing
\[ L(f) \EQ \frac{1-\alpha}{3\alpha}\int_{\mathbb R}f(x)^2\phi(x)\,dx+ \frac{2-2\alpha}{3\alpha}\int_{\mathbb R^2}f(x)f(y)\phi_{\rho}(x, y)\,dx dy + \int_{\mathbb R^2}f(x)f(y) \phi_{\rho_0}(x, y)\,dx dy\;.
\]
\strikeout{We know that $f$ is monotone and that it attains the values $-1$ and $1$ when $|x|$ is sufficiently large.} Let $a := \sup \{x : f(x) < 1\}$. Since $f$ is increasing and odd, we know that $a \ge 0$ and that $|f(x)| < 1$ on $(-a, a)$. Note that it may be the case that $a = \infty$.

We prove in Appendix~\ref{app:minimizer} that there is an `optimal' $f$. That is, there exists a \change{function} $f$ which globally minimizes $L(f)$ among all valid RPR$^2$ rounding functions. This means that any adjustment to $f$, not violating the condition $\|f\|_{\infty} \le 1$ increases the functional. More precisely, consider any $a' \in (0, a)$ and any measurable $h : (-a', a') \to [-1, 1]$ (not necessarily odd nor monotone). Note that for some $\eps$ sufficiently small, for all $s \in [-\eps, \eps]$, $\|f+s h\|_{\infty} \le 1$. Then, in order for there to be a local minimum of $L$ at $s = 0$, we must have that
\begin{align*}
0 &\EQ \frac{1}{2}\left.\frac{\partial}{\partial s}L(f + sh)\right|_{s=0}\\
&\EQ \frac{1-\alpha}{3\alpha}\intd_{-a'}^{a'} f(y)h(y)\phi(y)\,dy +
\frac{2-2\alpha}{3\alpha}\intd_{-\infty}^\infty \intd_{-a'}^{a'} f(x)h(y) \phi_{\rho}(x,y)\,dx dy+\intd_{-\infty}^\infty \int_{-a'}^{a'} f(x)h(y) \phi_{\rho_0}(x,y)\,dx dy\\
&\EQ \intd_{-a'}^{a'} h(y) \left( 
\frac{1-\alpha}{3\alpha}f(y)\phi(y) +\frac{2-2\alpha}{3\alpha}\intd_{\infty}^{\infty}f(x)\phi_{\rho}(x,y)\,dx
\intd_{-\infty}^\infty f(x) \phi_{\rho_0}(x,y)\,dx
\right)\,dy \;.
\end{align*}
Since $h$ is an arbitrary measurable function on $(-a', a')$, we have that for almost every $y \in (-a', a')$,
\begin{align}
\frac{1-\alpha}{3\alpha}f(y)\phi(y) +\frac{2-2\alpha}{3\alpha}\int_{\infty}^{\infty}f(x)\phi_{\rho}(x,y)\,dx
\int_{-\infty}^\infty f(x) \phi_{\rho_0}(x,y)\,dx = 0\;.\label{eq:fred-nae3}
\end{align}
Since $a' \in (0, a)$. The above holds for all $\change{y} \in (-a, a)$. Further define
\[K(x, y)\EQ \frac{2-2\alpha}{3\alpha}\cdot \frac{\phi_\rho(x,y)}{\phi(y)} + \frac{\phi_{\rho_0}(x, y)}{\phi(y)} \EQ \frac{2-2\alpha}{3\alpha}\cdot \frac{ \exp\left(-\frac{(x-\rho y)^2}{2(1-\rho^2)}\right)}
{\sqrt{2\pi(1-\rho^2)}}+\frac{ \exp\left(-\frac{(x-\rho_0 y)^2}{2(1-\rho_0^2)}\right)}
{\sqrt{2\pi(1-\rho_0^2)}} \;. \]
and
\[ g(y) \;=\; \frac{3\alpha}{1-\alpha} \left(\int_{-\infty}^{-a} K(x,y)\,dx - \int_{a}^{-\infty} K(x,y)\,dx \right) \;. \]

Then, (\ref{eq:fred-nae3}) becomes
\begin{align}
f(y) +
\frac{3\alpha}{1-\alpha}\int_{-a}^a f(x) K(x,y)\,dx \EQ g(y) \quad,\quad -a\le y\le a \;.\label{eq:fredholm-na3e}
\end{align}

This is \change{again} a Fredholm integral equation of the second kind.

In the other case, where $\rho_0 = 1$, we have that
\begin{align*}
s_3(f, \scD) &\EQ (1-\alpha) \cdot \frac{3-F_2(1)-2F_2(\rho)}{4} + \alpha \cdot \frac{3 - 3F_2(\rho_0)}{4}\\
&\EQ \frac{3}{4} - \frac{1-\alpha}{4}\int_{\mathbb R}f(x)^2\phi(x)\,dx- \frac{1-\alpha}{2}\int_{\mathbb R^2}f(x)f(y)\phi_{\rho}(x, y)\,dx dy - \frac{3\alpha}{4} \int_{\mathbb R}f(x)^2 \phi(x)\,dx\;\\
&\EQ \frac{3}{4} - \frac{1-4\alpha}{4}\int_{\mathbb R}f(x)^2\phi(x)\,dx- \frac{1-\alpha}{2}\int_{\mathbb R^2}f(x)f(y)\phi_{\rho}(x, y)\,dx dy\;.
\end{align*}
By applying a similar argument, we have $a := \sup \{x : |f(x)| < 1\}$ and for all $y \in (-a, a)$,
\begin{align}
\frac{1-4\alpha}{2-2\alpha}f(y)\phi(y) +
\int_{-\infty}^\infty f(x) \phi_{\rho}(x,y)\,dx \EQ 0\;.\label{eq:fred-nae3sat2}
\end{align}

We can then define
\[K(x, y)\;:=\;\frac{\phi_\rho(x,y)}{\phi(y)} \EQ \frac{ \exp\left(-\frac{(x-\rho y)^2}{2(1-\rho^2)}\right)
}
{\sqrt{2\pi(1-\rho^2)}} \;. \]
and
\[ g(y) \;:=\; \frac{1-4\alpha}{2-2\alpha} \left(\int_{-\infty}^{-a} K(x,y)\,dx - \int_{a}^{-\infty} K(x,y)\,dx \right) \;. \]

Then, (\ref{eq:fred-nae3sat2}) becomes
\begin{align}
f(y) +
\frac{1-4\alpha}{2-2\alpha}\int_{-a}^a f(x) K(x,y)\,dx \;=\; g(y) \quad,\quad -a\le y\le a \;.\label{eq:fredholm-nae3sat2}
\end{align}

\subsection{Solving integral equations}

\subsubsection{Iterative techniques}
A Fredholm integral equation of the second kind has the following standard form:
\[ f(x) - \lambda
\int_{a}^b K(x,y)f(y)\,dy \;=\; g(x) \quad,\quad a\le x\le b \;.
\]
Perhaps the simplest way of solving such integral equations that do not have a closed-form solution is the method of successive approximations. We construct a sequence of functions $\{f_n(x)\}$ that hopefully converge to a solution of the integral equation. The sequence $\{f_n(x)\}$ is defined as follows:
\[ f_0(x) \;=\; g(x) \;.\]
\[ f_n(x) \;=\; g(x) + \lambda \int_{a}^b K(x,y)f_{n-1}(y)\,dy \quad,\quad n>1\;, \]

\subsubsection{Linear system solving}

Since we have more structure in our kernels for MAX CUT and MAX NAE-$\{3\}$-SAT, we can compute a discretization of $f$ by solving a linear system. Recall that we sought to maximize a functional of the form
\begin{align}
L(f) = c - \lambda_1 \int_{\mathbb R} f(x)^2\phi(x)\,dx - \lambda_2 \int_{\mathbb R^2} f(x)M(x, y)f(y)\,dy,\label{eq:loss-abstract}
\end{align}
for some suitable function $M(x, y)$. Let $N \ge 1$ be a positive integer and partition $\mathbb R$ into $-\infty = a_0 < a_1 < \cdots < a_N = \infty$ such that for all $i \in \{1, 2, \hdots, N\}$,
\[
\int_{a_{i-1}}^{a_i} \phi(x)\,dx \EQ \frac{1}{N}\;.
\]

Assume that $f$ is piecewise constant, taking on value $f_i$ in the interval $(a_{i-1}, a_i)$. We let $\bf$ denote the vector of $f_i$'s. \change{Since any odd and monotone function can be approximated in the $\ell_\infty$ norm by piecewise-constant odd and monotone functions,} we may \change{assume} that $f$ is odd and monotone. \change{For} all $i, j \in \{1, \hdots, N\}$\change{,} define
\[
\hat{M}_{i,j} \EQ \int_{a_{j-1}}^{a_j}\int_{a_{i-1}}^{a_i} M(x, y)\,dx\,dy\;.
\]
Then (\ref{eq:loss-abstract}) becomes
\[
L(\bf) \EQ c - \frac{\lambda_1}{N} \sum_{i=1}^N f_i^2 \change{- \lambda_2}\sum_{i=1}^{N}\sum_{j=1}^Nf_i \hat{M}_{i,j} f_j \EQ c - \bf^T\change{\left(\frac{\lambda_1}{N} I + \lambda_2 \hat{M}\right)}\bf\;.
\]

For all $i$ such that $|f_i| < 1$, we must have that
\[
0 \EQ -\frac{\partial}{\partial f_i}L(\bf) \EQ \change{\frac{2\lambda_1}{N}} f_i + \lambda_2 \sum_{j=1}^n (\hat{M}_{i,j} + \hat{M}_{j\change{,}i}) f_j\;.
\]
In our situations, $\hat{M}$ is a symmetric matrix, so we have for all $i$ such that $|f_i| < 1$,
\[
f_i + \frac{\change{N} \lambda_2}{\lambda_1} \sum_{j=1}^n \hat{M}_{i,j}f_j \EQ 0\;.
\]
Let $\lambda := \frac{\change{N} \lambda_2}{\lambda_1}$ and let $i_a$ be the largest index such that $f_i = -1$. Then, we can define for $i \in \{i_a+1, \hdots, N-i_a\}$,
\[
g_i \;:=\; \lambda \sum_{j=1}^{i_a} \hat{M}_{i,j} - \hat{M}_{i,N-j+1}\;.
\]
Let $\bf'$ be $\bf$ restricted to $\{i_a+1, \hdots, N-i_a\}$ and $\hat{M}'$ be $\hat{M}$ restricted to $\{i_a+1, \hdots, N-i_a\}^2$. We then have the linear system
\[
(I + \lambda\hat{M}')\bf' \EQ \bg\;. 
\]  

This is a discrete Fredholm equation of the second kind. We can directly solve the linear system.\footnote{Note that $I + \lambda \hat{M}'$ can only be non-invertible for $O(N)$ values of $\lambda$. In the experiments (next section), this was only problematic when $\alpha=1$ where $I + \lambda \hat{M'}$ has rank $1$. A discussion about invertibility of similar linear systems for MAX CUT is given in Section 6 of~\cite{OW08}.} Note that we need to guess the value of $i_a$, so we need to solve multiple linear systems. Empirically, we can use a binary search to find the value of $i_a$.

For any $\eps > 0$, if we pick $N = O(1/\eps)$ and try $\alpha$ and $\rho$ in a grid of size $O(1/\eps)$. For each choice, we shall get an optimal step function for that distribution. By trying all of these functions on any input and taking the best expected result (e.g.,~\cite{R09}), we can compute optimal approximation factor to within an additive $\eps^{O(1)}$.\footnote{Proof of this follows by suitably modifying the arguments of~\cite{OW08}.}

\subsection{Experimental Results}\label{sub-exper}

Using MATLAB, we implemented a search to find the optimal $\hat{\bf}$, for various choices of $\alpha \in [0, 1]$, $\rho \in [-1, 0]$, and (for MAX NAE-$\{3\}$-SAT) $\rho_0 \in \{\max(-\third, \rho), 0\}$. For MAX CUT, our results reproduced the results found by \cite{OW08}.

For MAX NAE-$\{3\}$-SAT, we started with a coarse search. In particular, we did a grid search over $500^2$ values of $\alpha$ and $\rho$ and considered step functions with $N=100$ steps. From this, we computed a numerical approximation of the completeness/soundness tradeoff curve (Figure~\ref{fig:nae-3-sat}). Using these calculations, we estimated that the approximation ratio of MAX NAE-$\{3\}$-SAT is 0.9089 to four digits of precision.

\begin{figure}[t]
\begin{center}
\includegraphics[width=2.5in]{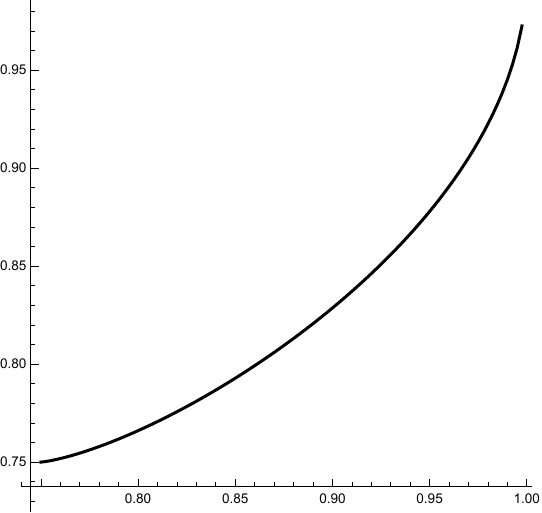}
\hspace*{1cm}
\begin{tabular}[b]{ccc}
Completeness & Soundness & Ratio\\
\hline
0.750 & 0.750000 & 1.000000 \\
0.775 & 0.756203 & 0.975746 \\
0.800 & 0.766051 & 0.957564 \\
0.825 & 0.778354 & 0.943460 \\
0.850 & 0.792829 & 0.932740 \\
0.875 & 0.809500 & 0.925143 \\
0.900 & 0.828659 & 0.920732 \\
0.925 & 0.850978 & 0.919976 \\
0.950 & 0.877926 & 0.924132 \\
0.975 & 0.913468 & 0.936890 \\
\\
\\
\end{tabular}
\end{center}
\caption{
\change{A plot and table showing the tradeoff between completeness ($x$-axis) and soundness ($y$-axis) for MAX NAE-$\{3\}$-SAT. (For a similar tradeoff for MAX CUT, see p.~339 of \cite{OW08}.)}}
\label{fig:nae-3-sat}
\end{figure}

With a more fine-grained search around the hardest points, we found that the most difficult point is $\alpha \approx 0.7381$, $\rho \approx -0.7420$, and $\rho_0 = -\third$. At this point, we computed a step function with $600$ steps which attains an approximation ratio of $\approx 0.9089169$.  The plot of the rounding function which attains this ratio is shown in Figure~\ref{fig:opt} (left). 

For a fixed rounding function $f$, if we assume $\rho_0 = 1$, then \[s_3(\mathcal D) - 0.908916 \cdot c_3(f, \mathcal D)\] is \change{an} affine function in $\alpha$ and a convex function in $\rho$. This is also the case if $\rho_0 = -\third$ or $\rho_0 = \rho$. Thus, for each of $\alpha \in \{0, 1\}$ and $\rho_0 \in \{1, -\third, \rho\}$, we did a ternary search to check that $f$ achieves an approximation of at least $0.908916$ on all hard distributions (up to a numerical error of $10^{-9}$ in the choice of $\rho$).

\begin{figure}[ht]
\centering
\includegraphics[height=2.5in]{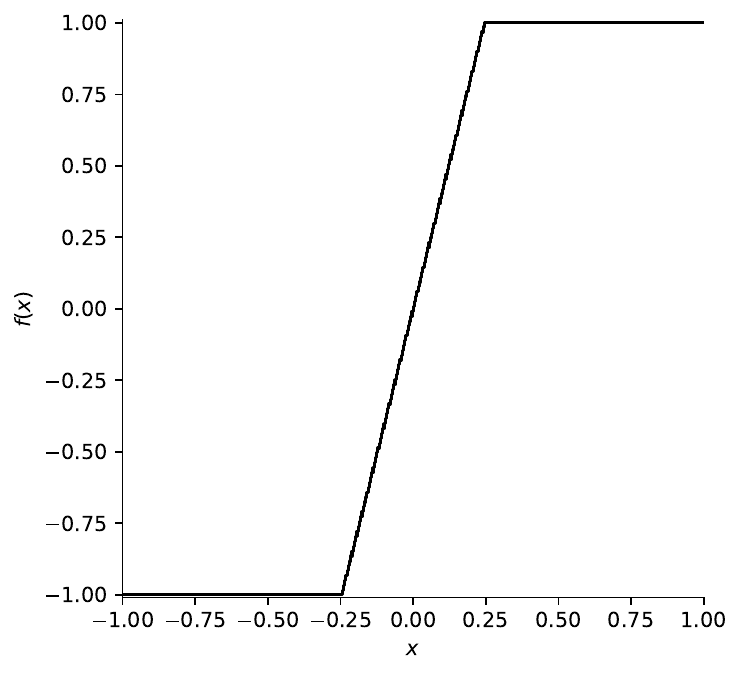}
\includegraphics[height=2.5in]{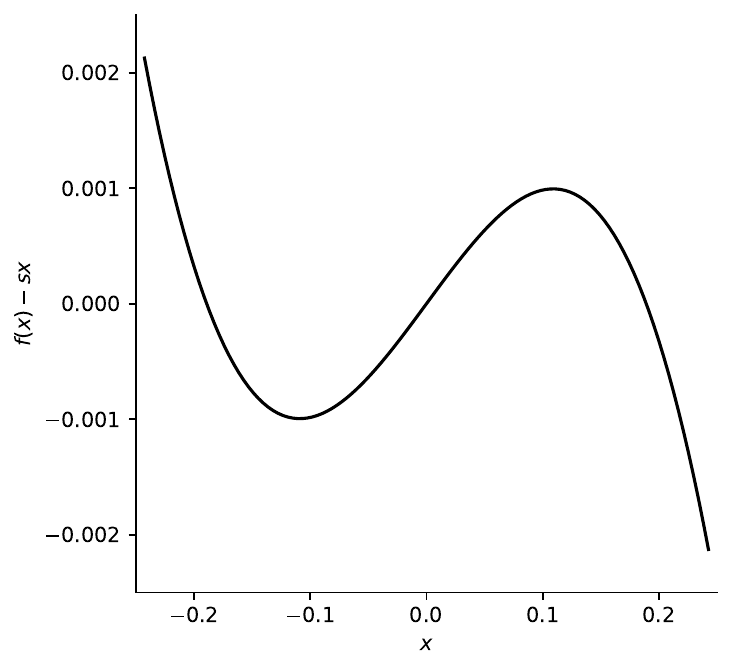}
\caption{(left) Near-optimal rounding function for MAX NAE-$\{3\}$-SAT in terms of approximation factor. (right) Approximate deviation of the near-optimal rounding function for MAX NAE-$\{3\}$-SAT from the best-fit $s$-linear function.}
\label{fig:opt}
\end{figure}

This step function is essentially indistinguishable from a \emph{$s$-linear} function of the form $\max(\min(sx, 1), -1)$. Using linear regression, we found that the best fit is $s \approx 4.072132$. Figure~\ref{fig:opt} (right) shows the deviation of~$f$ from an $s$-linear function in the region where $f$ is strictly between $-1$ and $1$. Notice that the curve is essentially cubic, but the deviation from $s$-linear is smaller than $10^{-3}$.

\section{Approximation algorithms for satisfiable MAX NAE-\texorpdfstring{$K$}{K}-SAT}\label{S-approx}

In this section we experiment with approximation algorithms for MAX NAE-SAT, as well as some restrictions of it, such as MAX NAE-$\{3,5\}$-SAT and MAX NAE-$\{3,7,8\}$-SAT. We focus in this paper on approximation algorithms for satisfiable, or almost satisfiable, instances of these problems. When considering satisfiable, or almost satisfiable, instances, we may assume that there are no clauses of size~$2$.

\subsection{On the difficulty of proving rigorous approximation bounds}\label{sub-difficult}

Our proposed approximation algorithms round the solution of the basic SDP using the RPR$^2$ method with a carefully chosen rounding function $f:[-\infty,\infty]\to [-1,1]$. 
The approximation ratio $\alpha_K(f)$ of an algorithm for \change{satisfiable instances of} MAX NAE-$K$-SAT, for some finite set~$K$, that uses RPR$^2$ with rounding function~$f$ is 
\[ \alpha_K(f) \;=\; \min_{k\in K} \min_{\begin{array}{c}\scriptstyle\bv_1,\bv_2,\ldots,\bv_k \\[-4pt] \scriptstyle
\RELAX(\bv_1,\bv_2,\ldots,\bv_k)=1\end{array}} \prob_f(\bv_1,\bv_2,\ldots,\bv_k) \;,\]
where $\bv_1,\bv_2,\ldots,\bv_k$ are assumed to be unit vectors that can be written as a convex combination of integral solutions, i.e., $\RELAX(\bv_1,\bv_2,\ldots,\bv_k)=1$, and $\prob_f(\bv_1,\bv_2,\ldots,\bv_k)$ is the probability that rounding the vectors using~$f$ yields an assignment that satisfies the corresponding $\NAE_k$ clause.

Even for a given~$f$, this is a fairly difficult optimization problem. For $k=9$, for example, this is essentially a 36-dimensional problem. (Here $36=\binom{9}{2}$.) What makes the problem even harder is that even the computation of $\prob_f(\bv_1,\bv_2,\ldots,\bv_k)$ for given vectors $\bv_1,\bv_2,\ldots,\bv_k$ is a non-trivial task, as it essentially amounts to computing a $k$-dimensional integral.

In view of these difficulties, the approximation ratios of the algorithms we consider for \change{satisfiable instances of} MAX NAE-$K$-SAT, where $\max K>3$, are only conjectured. We believe that we know which configuration $\bv_1,\bv_2,\ldots,\bv_k$ attains the minimum in the above expression, and thus determines the approximation ratio, but we are not able, at the moment, to prove it rigorously.

It is, in principal, possible to make the analysis of the proposed approximation algorithms rigorous, as done for MAX $[3]$-SAT and MAX NAE-$[4]$-SAT in \cite{zwick02}, but this would require a lot of effort. We note that even for MAX $[2]$-SAT and MAX DI-CUT, this is a non-trivial task. (See Sj{\"o}gren \cite{Sjogren09} \change{and Brakensiek et al.~\cite{BHPZ22,BHZ24}}.)

Even though the approximation ratios that we obtain are only conjectured, we believe that they are useful guides for further theoretical investigations of the MAX NAE-SAT and MAX SAT problems.

\subsection{Conjectures regarding optimal rounding procedures}\label{sub-conj}

The experiments we did with approximation algorithms for satisfiable, or almost satisfiable, instances of MAX NAE-$K$-SAT, for various sets $K$, lead us to make the following conjectures:

\begin{conjecture}\label{C-Sym}
The hardest configuration $\bv_1,\bv_2,\ldots,\bv_k$, where $k\in K$ and $k\ge 3$, for the optimal $RPR^2$ rounding function $f_K$ for satisfiable, or almost satisfiable, instances of {\rm MAX NAE-$K$-SAT} is the symmetric configuration in which $\bv_i\cdot \bv_j=1-\frac{4}{k}$, for every $i\ne j$.
\end{conjecture}

Our intuition for this conjecture is that this configuration comes from taking the uniform distribution over satisfying assignments where all but one of the $X_i$ are the same, which we expect are the hardest satisfying assignments to distinguish from the unsatisfying assignments.

A conjecture similar to Conjecture~\ref{C-Sym} was also made in Avidor et al.\ \cite{ABZ05}. Note that for $K=\{3\}$, the  conjecture is true and is a corollary of Theorem~\ref{NAE3SATdisttheorem}, as the only hard point with completeness~$1$ is $(-\frac13, -\frac13, -\frac13)$. When $K=\{3,4\}$ or $K = \{4\}$, the conjecture is also true as hyperplane rounding gives a $\frac{7}{8}$ approximation.

However, for $K$ which contain larger $k$, the situation is more subtle. When $k \geq 5$, this point is not the hardest configuration for all rounding functions because it is not a local maximum for $\sum_{i < j \in [k]}{F_2(\bv_i\cdot \bv_j)}$. That said, we conjecture that for the optimal rounding function, any potential increase for $F_2$ is offset by decreases to $F_{4},F_{6},\ldots$, so this is still the hardest configuration for the optimal rounding function.

\begin{conjecture}\label{C-RPR2}
The optimal rounding procedure for \emph{satisfiable}, or almost satisfiable, instances of {\rm MAX NAE-$K$-SAT} is the (one-dimensional) $RPR^2$ procedure with an appropriate rounding function $f=f_K$.
\end{conjecture}

Our intuition for this conjecture is that if the hardest configuration is the symmetric configuration in which $\bv_i\cdot \bv_j=1-\frac{4}{k}$, for every $i\ne j$ then these vectors can be split into a common component and a component which is orthogonal to everything else. For more details, see Section \ref{sub-prob}. Our conjecture is that one dimensional rounding schemes are most effective for interacting with this common component. 

Note that Conjecture~\ref{C-RPR2} does not follow from Ragavendra \cite{R08,R09} and UGC, as Ragavendra \cite{R08,R09} uses a high-dimensional version of $RPR^2$.

\begin{conjecture}\label{C-Step}
Furthermore, if $|K|\ge 2$ and $\min K= 3$, then the optimal $RPR^2$ rounding function $f_K$ for {\rm MAX NAE-$K$-SAT} is a \emph{step function} that only assumes the values $+1$ and $-1$.
\end{conjecture}

Our intuition for this conjecture is that for $RPR^2$ rounding functions, we can describe its performance in terms of Hermite coefficients (for details of how this works for $F_4$, see Appendix \ref{HermiteF4Appendix}). If it is the case that it's most important to optimize the first few Hermite coefficients, this is accomplished by a $\pm{1}$ step function. For more details, see Section \ref{sub-Hermite}. We also give an alternative heuristic argument in support of this conjecture in Appendix \ref{sub-heuristic}.

We note that the requirement $\min K= 3$ in Conjecture~\ref{C-Step} is important. If $\min K>3$, then it follows from H{\aa}stad \cite{H01} that the optimal rounding procedure is simply choosing a random assignment, which is equivalent to using $RPR^2$ with the function $f(x)=0$, which does not assume only $\pm 1$ values.

\subsection{Hermite expansions of rounding functions}\label{sub-Hermite}

We give a motivation for an improved rounding scheme which does well for MAX NAE-$\{3,5\}$-SAT with perfect completeness. A tool used by \cite{OW08} in studying RPR$^2$ rounding functions is the \emph{Hermite expansion}.\footnote{\cite{OW08} consider the Hermite expansion for multivariate polynomials as well, but we only consider single variable polynomials in this section.} Define the $n$th normalized Hermite polynomial to be~\cite{godsil1981hermite}
\[
H_n(x) \EQ \frac{1}{\sqrt{n!}}\sum_{\ell = 0}^{\lfloor n/2 \rfloor} (-1)^\ell m_\ell(K_n)x^{n - 2\ell}\quad, \quad n = 0, 1, 2, \ldots
\]
Here $m_\ell(K_n)$ is the number of $\ell$-matchings in the complete graph of order $n$. 
\change{The first few polynomials are:
\[ H_0(x) \EQ 1 \quad,\quad H_1(x)\EQ x \quad,\quad H_2(x) \EQ \frac{1}{\sqrt{2}}(x^2-1) \quad,\quad H_3(x) \EQ \frac{1}{\sqrt{6}}(x^3-3x) \;. \]
The first few Hermite polynomials of odd degree are shown in Figure~\ref{F-Hermite}.
The polynomials $H_n(x)$ form} an orthonormal basis with respect to the Gaussian measure on $\mathbb R$. That is,
\[
\int_{-\infty}^{\infty} H_i(x)H_j(x)\phi(x)\,dx \EQ  \begin{cases}
1 & i = j\\
0 & i \neq j
\end{cases}\quad i,j = 0, 1, 2, \ldots,
\]  
where $\phi(x) = \frac{1}{\sqrt{2\pi}}e^{-x^2/2}$.

\begin{figure}[ht]
\centering
\includegraphics[width=3in]{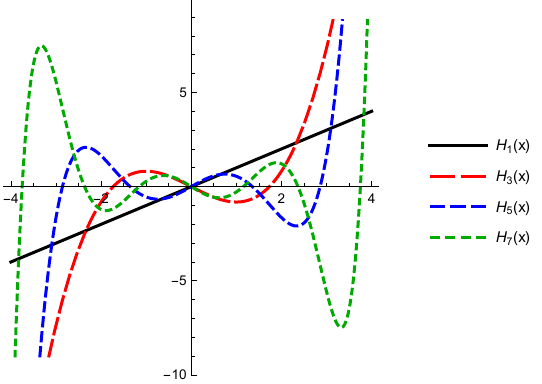}
\caption{The first few normalized Hermite polynomials of odd degree.}\label{F-Hermite}
\end{figure}

Note that any rounding scheme can be described by its Hermite expansion:

\[
f(x) \EQ \sum_{i=1\text{ odd}}^{\infty} c_i H_i(x) \;,
\]
where,
\[
c_i \EQ\langle f(x), H_i(x)\rangle \;:=\; \int_{-\infty}^{\infty} f(x)H_i(x) \phi(x) \,dx\;.
\]

For a positive integer $k \ge 1$, let 
\[
P_k \;:=\; \{\;\bigl(\langle f(x), H_1(x)\rangle, \langle f(x), H_3(x)\rangle, \ldots, \langle f(x), H_{2k-1}(x)\rangle\bigr)\; \mid\; \forall f : \mathbb R \to [-1, 1]\;\}\;.
\]

Note that $P_k$ is a convex set. See Figure~\ref{fig:P2} for a picture of $P_2$.

\begin{claim}
Every $(c_1, \hdots, c_k) \in P_k$ is attained by a step function with at most $2k^2$ steps.
\end{claim}

\begin{proof}
Since $P_k$ is convex, every point in $P_k$ can be expressed as a convex combination of $k$ extreme points. Thus, it suffices to show that each extreme point of $P_k$ can be attained by a step function with at most $2k$ steps. 

Fix an extreme point $(c_1, c_2, \hdots, c_k) \in P_k$. By definition of being an extreme point there is exists $(\alpha_1, \hdots, \alpha_k)$ such that $\sum_{i=1}^k \alpha_ix_i$ is maximized in $P_k$ at $(c_1, c_2, \hdots, c_k)$. That is, we seek to maximize
\[
\max_{f : \mathbb R \to [-1, 1]} \sum_{i=1}^k \langle f(x), \alpha_iH_{2i-1}(x)\rangle \EQ \max_{f : \mathbb R \to [-1, 1]}  \left\langle f(x), \sum_{i=1}^k\alpha_iH_{2i-1}(x)\right\rangle\;.
\]  
This expression is maximized when $f = \operatorname{sign}\left(\sum_{i=1}^k\alpha_iH_{2i-1}(x)\right)$. Since $\sum_{i=1}^k\alpha_iH_{2i-1}(x)$ is a degree $2k-1$ polynomial, its sign function has at most $2k$ steps.
\end{proof}

\begin{figure}
\centering
\includegraphics[width=7cm]{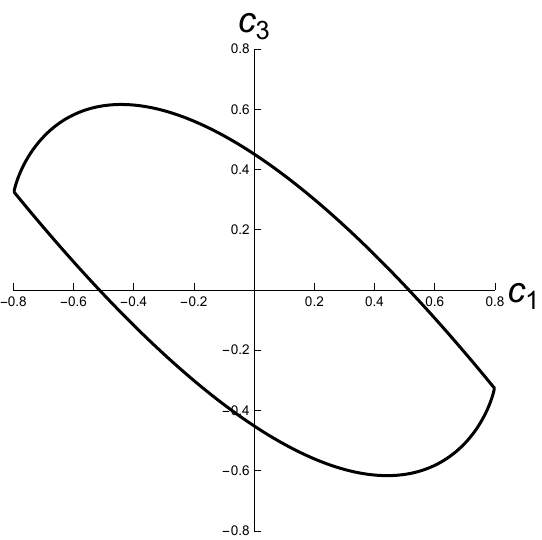}
\caption{The tradeoff between $c_1$ and $c_3$ for extreme RPR$^2$ rounding functions. The \strikeout{marked points} \change{lines shown} are the boundary of~$P_2$.}
\label{fig:P2}
\end{figure}

Thus, if we conjecture that optimizing the first few Hermite coefficients suffices to optimize the approximation ratio of MAX NAE-SAT to high precision, it suffices to look at step functions.

Using MATLAB, we searched through step functions achieving points in $P_2$, and as a result, we came across the following RPR$^2$ rounding function for almost satisfiable instances of MAX NAE-$\{3,5\}$-SAT:
\[
f_\alpha(x) \EQ \begin{cases}
-1 & x < -\alpha\\
\phantom{-}1 & x \in [-\alpha, 0]\\
-1 & x \in (0, \alpha]\\
\phantom{-}1 & x > \alpha
\end{cases} \;,
\]
where $\alpha \approx 2.27519364977$. See Figure~\ref{fig:double-step}.

\begin{figure}[ht]
\centering
\includegraphics[width=12cm]{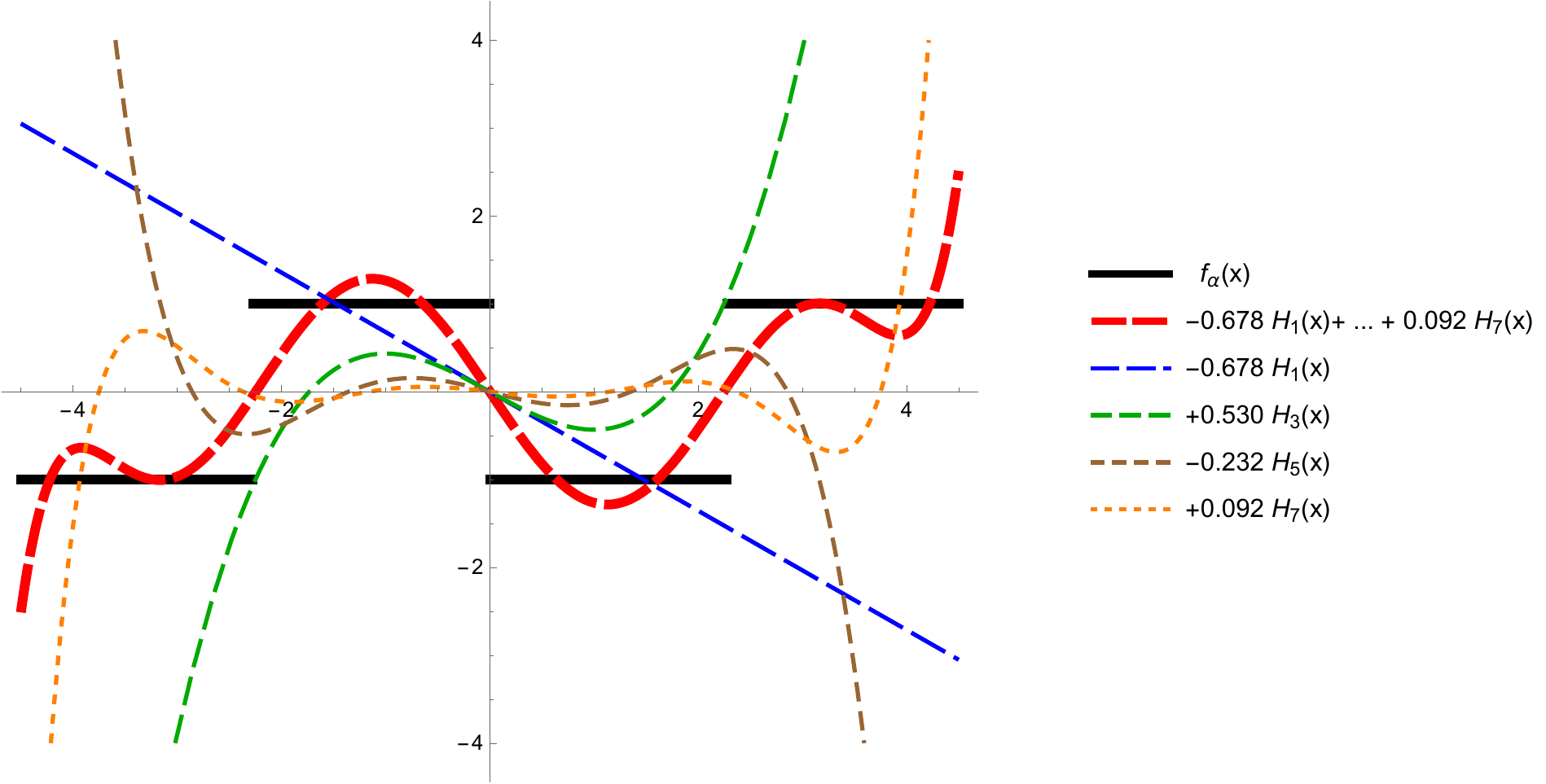}
\caption{The conjectured optimal rounding function $f_\alpha(x)$ for \change{almost-satisfiable instances of} MAX NAE-$\{3,5\}$-SAT. Plotted alongside are first four terms of its Hermite expansion, and the sum of these terms.}
\label{fig:double-step}
\end{figure}

Through numerical experiments, this rounding function achieves an \change{approximation ratio} of approximately $0.87286$ on \change{almost-satisfiable} instances. For more on experiments with step functions, see Section~\ref{sub-step}.

\subsection{Computing probabilities of symmetric configurations}\label{sub-prob}

Consider a collection of $k$ unit vectors $\bv_1,\bv_2,\ldots,\bv_k$ such that $\bv_i\cdot\bv_j=\rho$, for every $i\ne j$. We say that such a configuration is \emph{symmetric}. Let $p_f(k,\rho)=\prob_f(\bv_1,\bv_2,\ldots,\bv_k)$ be the probability that a $\NAE_k$ clause corresponding to these vectors $\bv_1,\bv_2,\ldots,\bv_k$ is satisfied when using rounding function $f$, then:
\[ \displaystyle
p_f(k,\rho) \;=\; \frac{1}{2^{k-1}}\left((2^{k-1}-1)-\sum_{i \text{ even }}\!\!\binom{k}{i} F_i(\rho)\right) \;,\]
where $F_i(\rho)=F_i(\rho,\rho,\ldots,\rho)$. Recall that $F_i$ is the $i$-wise moment function defined in Section~\ref{sub-moment}.

When the pairwise biases are all equal, \change{by Lemma~\ref{lemma:noisyf}, we have
\[
F_{2\ell}(\rho) \EQ \E_{x \sim N(0, 1)}\left[\left(\U_{\sqrt{\rho}}f\left(x\right)\right)^{2\ell}\right] \EQ \int\limits_{-\infty}^\infty (\U_{\sqrt{\rho}}f(x))^{2l} \phi(x)\, dx\;.
\]
This means that} we can evaluate $F_{2\ell}$, where $\ell$ is a positive integer, using two integrals instead of an $2\ell$-dimensional integral. \strikeout{The key idea is to decompose the covariance into a weighted sum of the all-one matrix $J$ and the identity matrix $I$. Assume that we have $2\ell$ unit vectors $\bv_1, \hdots, \bv_{2\ell} \in \mathbb R^n$ with inner product $\rho > 0$ between each pair of them. Then, if we take a random gaussian vector $\br \sim N(0, I_{n})$ and let $x_i = \br \cdot \bv_i$ be the inner product between $\br$ and $\bv_i$, then $x_i \sim N(0, 1)$ and $\Cov[x_i, x_j] = \rho$ for $i \neq j$. The covariance matrix can thus be written as $\rho J + (1-\rho) I$.}

\strikeout{The same distribution can be obtained by the following procedure: first sample $x \sim N(0, \rho)$, and then sample~$2\ell$ independent variables $\varepsilon_i \sim N(0, 1 - \rho)$ where $i = 1, 2, \hdots, 2\ell$, and finally let $x_i = x + \varepsilon_i$. Using this, $F_{2\ell}$ can be computed as follows:}
We note that if $f$ has the Hermite expansion $\sum_{i = 0}^\infty c_iH_i$, then $\U_\eta f$ has the Hermite expansion $\sum_{i = 0}^\infty c_i \eta^i H_i$.

Further observe that
\[ \displaystyle
p_f(k,\rho) = \frac{1}{2^{k-1}}\left((2^{k-1}-1)-\!\!\!\sum_{i \text{ even }}\!\!\binom{k}{i} F_i(\rho)\right) =
1 - \frac{1}{2^k} \int\limits_{-\infty}^\infty \biggl(
\left(1+(\U_{\sqrt{\rho}}f)(x)\right)^k + \left(1-(\U_{\sqrt{\rho}}f)(x)\right)^k
\biggr)\phi(x)\,dx
.\]

\subsubsection*{Formula for step functions}

For every $\ba=(a_1,a_2,\ldots,a_\ell)$, where $0<a_1<a_2<\cdots<a_\ell$, and $-1\le b_0,b_1,\ldots,b_\ell\le 1$, let $f(x)=f_{\ba,\bb}(x)$ be the function such that $f(x)=b_i$, if $a_{i}\le x<a_{i+1}$, for $i=0,1,\ldots,\ell$, where $a_0=0$ and $a_{\ell+1}=\infty$. Also, $f(x)=-f(-x)$, if $x<0$. We say that $f_{\ba,\bb}(x)$ is a $(\ell+1)$-step function, only counting steps to the right of the origin.
It is easy to check that
\[ (\U_{\rho}f)(x) \;=\; \sum_{i=0}^\ell b_i \left(
\Phi\biggl(\frac{a_{i+1}-\rho x}{\sqrt{1-\rho^2}}\biggr) -
\Phi\biggl(\frac{a_i-\rho x}{\sqrt{1-\rho^2}}\biggr) -
\Phi\biggl(\frac{-a_i-\rho x}{\sqrt{1-\rho^2}}\biggr) +
\Phi\biggl(\frac{-a_{i+1}-\rho x}{\sqrt{1-\rho^2}}\biggr)
\right) 
\;,\]
where $\Phi(x)=\int_{-\infty}^x \phi(y)\,dy$ is the cumulative distribution function of the standard normal distribution.

\subsection{Experiments with step functions}\label{sub-step}

For every $\ba=(a_1,a_2,\ldots,a_\ell)$, where $0<a_1<a_2<\cdots<a_\ell$, and $-1\le b_0,b_1,\ldots,b_\ell\le 1$, let $f(x)=f_{\ba,\bb}(x)$, as above, be the function such that $f(x)=b_i$, if $a_{i}\le x<a_{i+1}$, for $i=0,1,\ldots,\ell$, where $a_0=0$ and $a_{\ell+1}=\infty$. Also, $f(x)=-f(-x)$, if $x<0$. We say that $f_{\ba,\bb}(x)$ is a $(\ell+1)$-step function, only counting steps to the right of the origin. We also let $f_\ba(x)$ be the function $f_\ba(x)=f_{\ba,\bb}(x)$, where $b_i=(-1)^{i+1}$, for $i=0,1,\ldots,\ell$. Note that hyperplane rounding uses a $1$-step function $f_{(),(1)}(x)$ with $\ell=0$ and $b_0=1$.

Motivated by the results of Sections~\ref{sub-Hermite} (see also Appendix~\ref{sub-heuristic}), we did extensive numerical experiments with step functions. For a given set $K$, and a given number $k+1$ of steps, we solved a numerical optimization problem in which the variables are $a_1,a_2,\ldots,a_\ell$ and $b_0,b_1,\ldots,b_\ell$. The objective function maximized was $\alpha_K(f_{\ba,\bb})=\min_{k\in K} \alpha_k(f_{\ba,\bb})$, where $\ba=(a_1,a_2,\ldots,a_\ell)$ and $\bb=(b_0,b_1,\ldots,b_\ell)$. When computing $\alpha_k(f_{\ba,\bb})$ we considered only the conjectured hardest configuration of Conjecture~\ref{C-Sym}. As this configuration is symmetric, the probability $\alpha_k(f_{\ba,\bb})=\prob_{f_{\ba,\bb}}(\bv_1,\ldots,\bv_k)=p_{f_{\ba,\bb}}(k,1-\frac{4}{k})$, where $\bv_i\cdot \bv_j=1-\frac{4}{k}$, for $i\le j$, when $k\ge 3$, can be numerically computed using the  formula, which involves integration, given in Section~\ref{sub-prob}.

Most of the experiments were carried out in Mathematica using numerical optimization tools, taking advantage of the possibility of performing numerical calculations with arbitrary precision. The optimization problem of finding the optimal $a_1,a_2,\ldots,a_\ell$ and $b_0,b_1,\ldots,b_\ell$ is a fairly difficult optimization problem, as the objective function $\alpha_K(f_{\ba,\bb})=\min_{k\in K} \alpha_k(f_{\ba,\bb})$ is far from being a convex function. However, as the number of variables is relatively small, we were able to repeat the optimization attempts many times, from different initial points. This gives us some confidence that the best step functions found are close to being the optimal ones.

\strikeout{No counterexamples to the other two conjectures of Section~\ref{sub-conj} were found.} Most of our experiments were devoted to confirming Conjecture~\ref{C-Step}, i.e., that the optimal $RPR^2$ rounding function~$f_K$ for MAX NAE-$K$-SAT, where $|K|\ge 2$ and $\min K\ge 3$, is a $\pm 1$ step function, relying on Conjecture~\ref{C-Sym} stating that the worst configuration $\bv_1,\ldots,\bv_k$, where $k\in K$, for the optimal rounding function~$f_K$ is the symmetric configuration with $\bv_i\cdot\bv_j=1-\frac{4}{k}$, for every $i\ne j$.

The best step functions we found for satisfiable, or almost satisfiable, instances of MAX NAE-$K$-SAT, where $K=\{3,k\}$, for $k=5,6,\ldots,8$, and for $K=\{3,7,8\}$ are given in Table~\ref{T-Step}. What is interesting in these experiments is that when the number of allowed steps is high enough, the optimal step function is a $\pm1$ function, and allowing more steps does not seem to help. 

For example, for $K=\{3,5\}$ the best $2$-step function found is already a $\pm1$ function, identical to the one found in Section~\ref{sub-exper}. We could not improve on this function by allowing more steps.

For $K=\{3,6\}$, the best step function found is a $3$-step functions. For $K=\{3,7\}$ and $K=\{3,8\}$, the best step functions used are $4$-step functions.
For $K=\{3,9\}$ and $K=\{3,10\}$, not shown in the table, the best best results are obtained using $5$-step functions.

Again, we could not improve on these functions by allowing more steps. It is possible, that tiny improvements can be obtained by allowing more steps, but the improvements obtained in the approximation ratio, if any, are likely to be less than $10^{-9}$.

\begin{table}[t]
\begin{center}

\begin{tabular}{|c||c||c|c|}
\hline
$K=\{3,5\}$  & $a_1$ & $b_0$ & $b_1$ \\
\hline\hline
0.870978418 & & 0.863471455 & \\
\hline
0.872886331 &  2.275193649 & $-1$ & $1$ \\
\hline
\end{tabular}

\bigskip

\begin{tabular}{|c||c|c||c|c|c|}
\hline
$K=\{3,6\}$  & $a_1$ & $a_2$ & $b_0$ & $b_1$ & $b_2$\\
\hline\hline
0.869020196 & & & 0.856454637 & & \\
\hline
0.870806446 &  2.251163925 &  & $-1$ & $1$ & \\
\hline
0.870806482 & 2.251064988 & 4.502131583 & $-1$ & $1$ & $-1$ \\
\hline
\end{tabular}

\bigskip

\begin{tabular}{|c||c|c|c||c|c|c|c|}
\hline
$K=\{3,7\}$  & $a_1$ & $a_2$ & $a_3$ & $b_0$ & $b_1$ & $b_2$ & $b_3$  \\
\hline\hline
0.868331573 & & & & 0.853973417 & & & \\
\hline
0.86967887 &  1.617354199 & & 
&  $-1$ &  $-0.443504607$ & & \\
\hline
0.869818822 & 1.955864822 & 2.288418785 & & $-1$ & $1$ & $-1$ & \\
\hline
0.869818822 & 1.955862161 & 2.288413620 & 5.658697297 & $-1$ & $1$ & $-1$ & $1$ \\
\hline
\end{tabular}

\bigskip

\begin{tabular}{|c||c|c|c||c|c|c|c|}
\hline
$K=\{3,8\}$  & $a_1$ & $a_2$ & $a_3$ & $b_0$ & $b_1$ & $b_2$ & $b_3$  \\
\hline\hline
0.868384155 & & & & 0.854163133 & & & \\
\hline
0.869708575 &  1.342323152 & & &  $-1$ &  $-0.637982114$ & & \\
\hline
0.869954386 & 1.783234209 & 2.015766438 & & $-1$ & $1$ & $-1$ & \\
\hline
0.869954931 & 1.782430334 & 2.014523521 & 4.492762885  & $-1$ & $1$ & $-1$ & $1$ \\
\hline
\end{tabular}
\bigskip

\begin{tabular}{|c||c|c|c||c|c|c|c|}
\hline
\{3,7,8\}  & $a_1$ & $a_2$ & $a_3$ & $b_0$ & $b_1$ & $b_2$ & $b_3$  \\
\hline\hline
0.868331573 & & & & 0.853973417 & & & \\
\hline
0.869649096 &  1.486111761 & & &  $-1$ &  $-0.550842608$ & & \\
\hline
0.869809386 &  1.914108264 & 2.216226101 & & $-1$ & $1$ & $-1$ & \\
\hline
0.869809394 &  1.914115410 & 2.216234256 & 5.228184560 & $-1$ & $1$ & $-1$ & $1$\\
\hline
\end{tabular}

\end{center}
\caption{Best step functions for MAX NAE-$K$-SAT with a given number of steps. The first column gives the \change{conjectured} approximation ratio. The other columns give the vectors $\ba$ and $\bb$.}\label{T-Step}
\end{table}

Let $\alpha_K$ be the best ratio found for MAX NAE-$K$-SAT. We currently have
\[ \frac{7}{8}\;=\;\alpha_{3,4} \;>\; \alpha_{3,5} \;>\; \alpha_{3,6} \;>\; \alpha_{3,7} \;<\; \alpha_{3,8} \;<\; \cdots \]
Thus, if we just look at the mixture of two clause sizes, it seems that $\{3,7\}$ is the hardest.

However, it seems that $\alpha_{3,7,8}<\alpha_{3,7}<\alpha_{3,8}$. The best rounding function found for \MN{3,7,8} is a $4$-step $\pm1$ function. (It is likely that a really tiny improvement can be obtained by adding a fifth step, but this was not attempted.) The best functions found for $K=\{3,5\}$ and $K=\{3,7,8\}$ are shown in Figure~\ref{F-35-378-fun}.

\begin{figure}[t]
\centering
\hspace*{0.75cm}
\includegraphics[width=7.25cm]{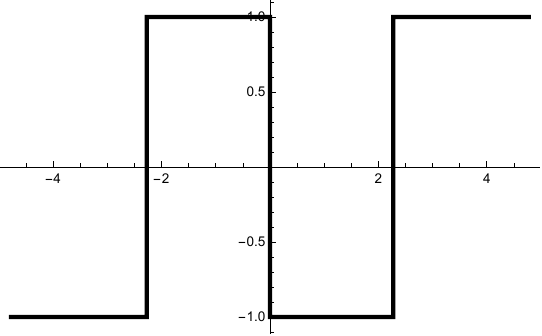}
\hspace*{0.75cm}
\includegraphics[width=7.25cm]{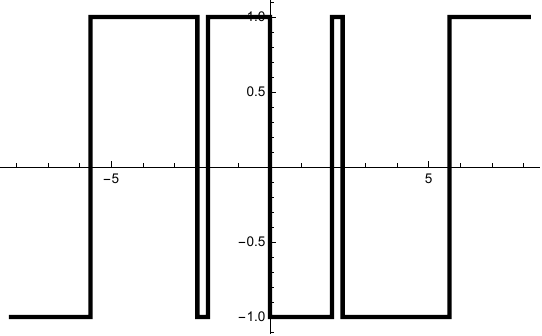}
\caption{The best rounding functions \change{found} for $K=\{3,5\}$ (left) and for $K=\{3,7,8\}$ (right).}\label{F-35-378-fun}
\end{figure}

When using the best function found for $K=\{3,7,8\}$, the satisfaction probabilities of clauses of size 3,7 and~8 are the same, up to numerical error. The satisfaction probability of every clause size $k>4$, $k\ne7,8$ is higher. (As mentioned, We rely on Conjecture~\ref{C-Sym}.) This seems to suggest that approximating satisfiable, or almost satisfiable, instances of MAX NAE-$\{3,7,8\}$-SAT is as hard as approximating satisfiable, or almost satisfiable, instances of MAX NAE-SAT. We thus conjecture that $0.869809$ is the best approximation ratio that can be obtained for almost satisfiable instances of MAX NAE-SAT.

An approximation ratio of $0.863$ for satisfiable instances of MAX NAE-SAT is conjectured in Zwick~\cite{Zwick99a}. This conjectured approximation ratio is obtained using outward rotation, which is a special case of $RPR^2$ rounding. Our new conjectured value of $0.869809$ is not much larger than the previously conjectured result. However, we believe that the rounding function for $K=\{3,7,8\}$ shown in Table~\ref{T-Step} and Figure~\ref{F-35-378-fun} is essentially the best rounding function for satisfiable, or almost satisfiable, instances of MAX NAE-SAT.

If Conjecture~\ref{C-Step} is true, i.e., the optimal rounding function is a step function, an interesting theoretical question would be whether the optimal function has a finite or
or an infinite number of steps. We speculate that a finite number of steps is enough.

We have made attempts to verify that the best step functions we found are at least local maxima, and that they cannot be improved by adding more steps.

Figure~\ref{F-35} shows the effect of trying to add a second breakpoint, i.e., a third step, for $K=\{3,5\}$.  The first breakpoint is held fixed at $a_1\approx 2.275193649773$ and the probabilities for clauses of sizes 3 and 5 are shown as functions of $a_2$. (The approximation ratio obtained using the best 2-step function is subtracted from these probabilities. Both graphs show the same functions.) The two probabilities do cross each other around $a_2\approx 7$, but the two functions are increasing at that point, suggesting that the optimal value of $a_2$ is $+\infty$, i.e., it is not beneficial to add a second breakpoint and a third step.

\begin{figure}[t]
\centering
\includegraphics[width=7.7cm]{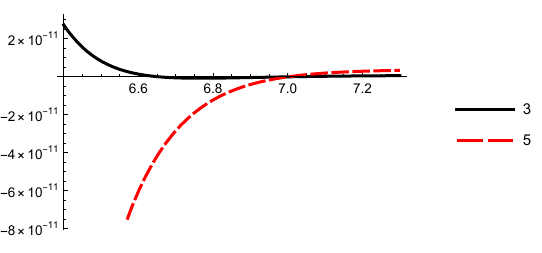}
\hspace*{0.75cm}
\includegraphics[width=7.7cm]{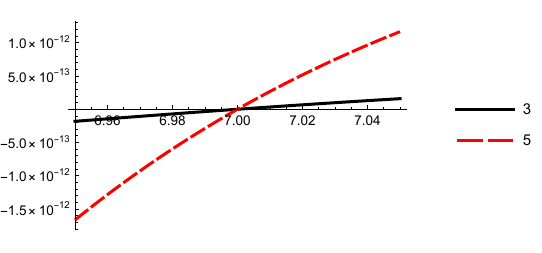}
\caption{The effect of trying to add a second breakpoint for $\{3,5\}$. The two functions shown in the two figures are the satisfaction \change{relative} probabilities of clauses of size 3 and 5, assuming Conjecture~\ref{C-Sym}, as a function of~$a_2$. \change{All probabilities are relative to the satisfaction probability of the model when $a_2 \approx 7$.} The graph \change{on} the right is a close-up on $a_2\approx 7$. 
}\label{F-35}
\end{figure}

Graphs showing the effect of adding the fourth step \change{around $a_3 \approx 5.228184560$ and the} fifth step around $a_4\approx 9.279617526380$, for $K=\{3,7,8\}$, are shown in Figure~\ref{F-378}. In the graph on the left, $a_1$ $a_2$ are fixed to the optimal setting for a 3-step function, and the probabilities for clauses of sizes 3,7 and 8 are shown as a function of $a_3$. The graph on the right fixes $a_1,a_2$ and $a_3$ and shows the probabilities as a function of~$a_4$. It is possible that a really tiny improvement may be obtained by adding a fifth step.

\begin{figure}[t]
\centering
\includegraphics[width=7.7cm]{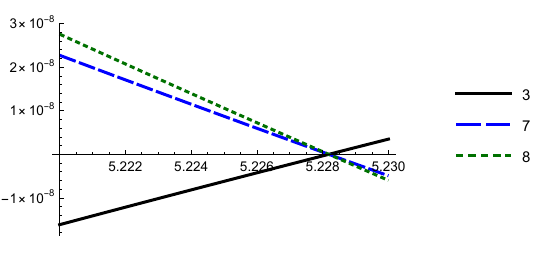}
\hspace*{0.75cm}
\includegraphics[width=7.7cm]{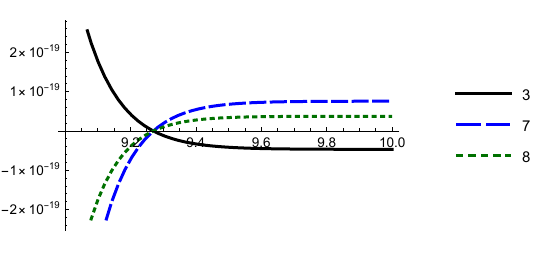}
\caption{The effect of adding the third and fourth breakpoints for $K=\{3,7,8\}$ \change{(c.f., Figure~\ref{F-35}}).} \change{Probabilities are relative to adding a breakpoint at $a_3 \approx 5.2282$ and $a_4 \approx 9.2796$, respectively.}\label{F-378}
\end{figure}

\section{Concluding Remarks and Open Questions}\label{S-concl}

We presented the first improved hardness of approximation result for MAX NAE-SAT in nearly two decades, showing that no $\frac78$-approximation algorithm can be obtained for the problem, assuming UGC. We also presented an optimal algorithm for MAX NAE-$\{3\}$-SAT, again assuming UGC. Finally, we presented an algorithm for (almost) satisfiable instances of MAX NAE-SAT that we conjecture to be nearly optimal. 

What we find striking is the contrast between MAX CUT and MAX NAE-$\{3\}$-SAT, in which the optimal rounding functions are smooth (except for near-perfect completeness, where we get hyperplane rounding), monotone functions obtained as a solution of an integral equations, and MAX NAE-$\{3,5\}$-SAT, MAX NAE-$\{3,7,8\}$-SAT and NAX NAE-SAT in which, at least for almost satisfiable instances, the apparent optimal rounding functions are non-monotone step functions that only assume $\pm1$ values, with a surprisingly small number of steps. 

Many problems remain. The obvious open problems are proving, or disproving, Conjectures~\ref{C-Sym}, \ref{C-RPR2} and~\ref{C-Step}. (Not to mention Conjecture~\ref{C-UGC}\ldots)

It would also be interesting to extend our results to approximating general MAX NAE-SAT, without assuming that instances are (almost) satisfiable. Initial investigations seem to suggest that more surprises lurk there.

More technically, we understand $F_2$ relatively well and used it to analyze \change{NAE} clauses of length 2 and~3. However, higher moments functions are not well understood and several natural questions can be asked about them. For example, we know that $F_2(x)$ is convex in $x$ for $x \ge 0$. Does $F_{2\ell}(x):=F_{2\ell}(x, \ldots, x)$ have the same property for $\ell \ge 2$? Having a better understanding of these higher moments can lead to a more rigorous analysis for clauses of length greater than 3.

Finally, it would be interesting to see if some of the ideas introduced in this paper can be used to decide whether there is a $\frac{7}{8}$-approximation algorithm for MAX SAT. What makes MAX SAT potentially easier than MAX NAE-SAT is that we can take advantage of individual biases. On the other hand, the search for optimal rounding functions becomes harder.

\section*{Acknowledgments}
We thank Ryan O'Donnell for sharing the code used to generate the experimental results in \cite{OW08}. \change{We also thank the anonymous reviewers for numerous useful comments. Some plots in this paper were made using Numpy~\cite{harris2020array} and Matplotlib~\cite{Hunter:2007}.}

\bibliographystyle{plain}
\bibliography{refs}

\begin{thebibliography}{10}

\bibitem{AE98}
Gunnar Andersson and Lars Engebretsen.
\newblock Better approximation algorithms for set splitting and not-all-equal
  {SAT}.
\newblock {\em Information Processing Letters}, 65(6):305--311, 1998.

\bibitem{ALMSS98}
Sanjeev Arora, Carsten Lund, Rajeev Motwani, Madhu Sudan, and Mario Szegedy.
\newblock Proof verification and the hardness of approximation problems.
\newblock {\em Journal of the ACM}, 45(3):501--555, 1998.

\bibitem{AS98}
Sanjeev Arora and Shmuel Safra.
\newblock Probabilistic checking of proofs: A new characterization of np.
\newblock {\em Journal of the ACM}, 45(1):70--122, 1998.

\bibitem{AW02}
Takao Asano and David~P Williamson.
\newblock Improved approximation algorithms for max sat.
\newblock {\em Journal of Algorithms}, 42(1):173--202, 2002.

\bibitem{Austrin07}
Per Austrin.
\newblock Balanced {MAX 2-SAT} might not be the hardest.
\newblock In {\em Proc.\ of 39th STOC}, pages 189--197, 2007.

\bibitem{Austrin10}
Per Austrin.
\newblock Towards sharp inapproximability for any 2-{CSP}.
\newblock {\em SIAM Journal on Computing}, 39(6):2430--2463, 2010.

\bibitem{ABZ05}
Adi Avidor, Ido Berkovitch, and Uri Zwick.
\newblock Improved approximation algorithms for {MAX} {NAE-SAT} and {MAX}
  {SAT}.
\newblock In {\em Approximation and Online Algorithms, Third International
  Workshop, {WAOA} 2005}, volume 3879 of {\em Lecture Notes in Computer
  Science}, pages 27--40. Springer, 2005.

\bibitem{BHPZ22}
Joshua Brakensiek, Neng Huang, Aaron Potechin, and Uri Zwick.
\newblock Separating max 2-and, max di-cut and max cut.
\newblock In {\em 2023 IEEE 64th Annual Symposium on Foundations of Computer
  Science (FOCS)}, pages 234--252. IEEE, 2023.

\bibitem{BHZ24}
Joshua Brakensiek, Neng Huang, and Uri Zwick.
\newblock Tight approximability of {MAX 2-SAT} and relatives, under {UGC}.
\newblock In {\em Proceedings of the 35th SODA}, pages 1328--1344.

\bibitem{BCR15}
Jonah Brown{-}Cohen and Prasad Raghavendra.
\newblock Combinatorial optimization algorithms via polymorphisms.
\newblock {\em CoRR}, abs/1501.01598, 2015.

\bibitem{bulatov2000constraint}
Andrei~A Bulatov, Andrei~A Krokhin, and Peter Jeavons.
\newblock Constraint satisfaction problems and finite algebras.
\newblock In {\em International Colloquium on Automata, Languages, and
  Programming}, pages 272--282. Springer, 2000.

\bibitem{FG95}
Uriel Feige and Michel Goemans.
\newblock Approximating the value of two prover proof systems, with
  applications to {MAX 2SAT} and {MAX DICUT}.
\newblock In {\em Proceedings Third Israel Symposium on the Theory of Computing
  and Systems}, pages 182--189. IEEE, 1995.

\bibitem{FL06}
Uriel Feige and Michael Langberg.
\newblock The {RPR${}^2$} rounding technique for semidefinite programs.
\newblock {\em Journal of Algorithms}, 60(1):1--23, 2006.

\bibitem{godsil1981hermite}
Chris~D Godsil.
\newblock Hermite polynomials and a duality relation for matchings polynomials.
\newblock {\em Combinatorica}, 1(3):257--262, 1981.

\bibitem{GW95}
Michel~X Goemans and David~P Williamson.
\newblock Improved approximation algorithms for maximum cut and satisfiability
  problems using semidefinite programming.
\newblock {\em Journal of the ACM}, 42(6):1115--1145, 1995.

\bibitem{HZ01}
Eran Halperin and Uri Zwick.
\newblock Approximation algorithms for {MAX 4-SAT} and rounding procedures for
  semidefinite programs.
\newblock {\em Journal of Algorithms}, 40(2):184--211, 2001.

\bibitem{harris2020array}
Charles~R. Harris, K.~Jarrod Millman, St{\'{e}}fan~J. van~der Walt, Ralf
  Gommers, Pauli Virtanen, David Cournapeau, Eric Wieser, Julian Taylor,
  Sebastian Berg, Nathaniel~J. Smith, Robert Kern, Matti Picus, Stephan Hoyer,
  Marten~H. van Kerkwijk, Matthew Brett, Allan Haldane, Jaime~Fern{\'{a}}ndez
  del R{\'{i}}o, Mark Wiebe, Pearu Peterson, Pierre G{\'{e}}rard-Marchant,
  Kevin Sheppard, Tyler Reddy, Warren Weckesser, Hameer Abbasi, Christoph
  Gohlke, and Travis~E. Oliphant.
\newblock Array programming with {NumPy}.
\newblock {\em Nature}, 585(7825):357--362, September 2020.

\bibitem{H01}
Johan H{\aa}stad.
\newblock Some optimal inapproximability results.
\newblock {\em Journal of the ACM}, 48(4):798--859, 2001.

\bibitem{HP19}
Neng Huang and Aaron Potechin.
\newblock On the approximability of presidential type predicates.
\newblock {\em arXiv preprint arXiv:1907.04451}, 2019.

\bibitem{Hunter:2007}
J.~D. Hunter.
\newblock Matplotlib: A 2d graphics environment.
\newblock {\em Computing in Science \& Engineering}, 9(3):90--95, 2007.

\bibitem{isserlis1918formula}
Leon Isserlis.
\newblock On a formula for the product-moment coefficient of any order of a
  normal frequency distribution in any number of variables.
\newblock {\em Biometrika}, 12(1/2):134--139, 1918.

\bibitem{KZ97}
Howard Karloff and Uri Zwick.
\newblock A 7/8-approximation algorithm for {MAX 3SAT}?
\newblock In {\em Proc.\ of 38th FOCS}, pages 406--415. IEEE, 1997.

\bibitem{khot02}
Subhash Khot.
\newblock On the power of unique 2-prover 1-round games.
\newblock In {\em FOCS 2002}, pages 767--775, 2002.

\bibitem{KKMO07}
Subhash Khot, Guy Kindler, Elchanan Mossel, and Ryan O’Donnell.
\newblock Optimal inapproximability results for max-cut and other 2-variable
  {CSP}s?
\newblock {\em SIAM Journal on Computing}, 37(1):319--357, 2007.

\bibitem{leoni2017first}
Giovanni Leoni.
\newblock {\em A first course in Sobolev spaces}.
\newblock American Mathematical Soc., 2017.

\bibitem{LLZ02}
Michael Lewin, Dror Livnat, and Uri Zwick.
\newblock Improved rounding techniques for the {MAX 2-SAT} and {MAX DI-CUT}
  problems.
\newblock In {\em International Conference on Integer Programming and
  Combinatorial Optimization}, pages 67--82. Springer, 2002.

\bibitem{MM17}
Konstantin Makarychev and Yury Makarychev.
\newblock {Approximation Algorithms for CSPs}.
\newblock In Andrei Krokhin and Stanislav Zivny, editors, {\em The Constraint
  Satisfaction Problem: Complexity and Approximability}, volume~7 of {\em
  Dagstuhl Follow-Ups}, pages 287--325. Schloss Dagstuhl--Leibniz-Zentrum fuer
  Informatik, Dagstuhl, Germany, 2017.

\bibitem{MM03}
Shiro Matuura and Tomomi Matsui.
\newblock New approximation algorithms for {MAX 2SAT} and {MAX DICUT}.
\newblock {\em Journal of the Operations Research Society of Japan},
  46(2):178--188, 2003.

\bibitem{OW08}
Ryan O'Donnell and Yi~Wu.
\newblock An optimal {SDP} algorithm for {Max-Cut}, and equally optimal long
  code tests.
\newblock In {\em Proc.\ of 40th STOC}, pages 335--344, 2008.
\newblock Fuller version:
  {https://www.cs.cmu.edu/$\sim$odonnell/papers/optimal-max-cut.pdf}.

\bibitem{O14}
Ryan O’Donnell.
\newblock {\em Analysis of Boolean Functions}.
\newblock Cambridge University Press, USA, 2014.

\bibitem{polyanin2008handbook}
Andrei~D. Polyanin and Alexander~V. Manzhirov.
\newblock {\em Handbook of integral equations}.
\newblock CRC press, 2008.

\bibitem{P19}
Aaron Potechin.
\newblock On the approximation resistance of balanced linear threshold
  functions.
\newblock In {\em Proc.\ of 51st STOC}, pages 430--441, 2019.

\bibitem{R08}
Prasad Raghavendra.
\newblock Optimal algorithms and inapproximability results for every {CSP}?
\newblock In {\em Proc.\ of 40th STOC}, pages 245--254, 2008.

\bibitem{R09}
Prasad Raghavendra.
\newblock {\em Approximating {NP}-hard Problems - Efficient Algorithms and
  their Limits}.
\newblock PhD thesis, University of Washington, 2009.

\bibitem{RS09}
Prasad Raghavendra and David Steurer.
\newblock How to round any {CSP}.
\newblock In {\em Proc.\ of 50th FOCS}, pages 586--594. IEEE, 2009.

\bibitem{Sjogren09}
Henrik Sj{\"o}gren.
\newblock Rigorous analysis of approximation algorithms for {MAX 2-CSP}.
\newblock Master's thesis, KTH Royal Institute of Technology, 2009.

\bibitem{ZYH04}
Jiawei Zhang, Yinyu Ye, and Qiaoming Han.
\newblock Improved approximations for max set splitting and max {NAE SAT}.
\newblock {\em Discrete Applied Mathematics}, 142(1-3):133--149, 2004.

\bibitem{zhuk2017proof}
Dmitriy Zhuk.
\newblock A proof of {CSP} dichotomy conjecture.
\newblock In {\em Proc.\ of 58th FOCS}, pages 331--342. IEEE, 2017.

\bibitem{Zwick99a}
Uri Zwick.
\newblock Outward rotations: {A} tool for rounding solutions of semidefinite
  programming relaxations, with applications to {MAX} {CUT} and other problems.
\newblock In {\em Proc.\ of 31th STOC}, pages 679--687. {ACM}, 1999.

\bibitem{zwick02}
Uri Zwick.
\newblock Computer assisted proof of optimal approximability results.
\newblock In {\em Proceedings of the 13th SODA}, pages 496--505, 2002.

\end{thebibliography}

\appendix

\section{Hermite Decomposition of \texorpdfstring{$F_4$}{F4}}\label{HermiteF4Appendix}

In this appendix, we express $F_4$ in terms of the Hermite coefficients of the rounding function, as is previously done for $F_2$. We recall the following theorem from probability theory.
\begin{theorem}[Isserlis, e.g., \cite{isserlis1918formula}]
Let $X_1, X_2, \ldots, X_{2n}$ be jointly gaussian random variables such that $\E[X_i] = 0$ for every $i \in [2n]$, then
\begin{align*}
    & \E[X_1X_2\cdots X_{2n}] \EQ \sum_{M}\prod_{\{i, j\} \in M}\E[X_iX_j]\;, \\
    & \E[X_1X_2\cdots X_{2n-1}] \EQ 0\;.
\end{align*}
Here the summation $\sum_M$ runs over every perfect matching $M$ of the complete graph $K_{2n}$.
\end{theorem}

\begin{corollary}
Let $i, j, k, l$ be nonnegative integers, $X_1, X_2, X_3, X_4$ be jointly gaussian zero-mean random variables such that $\E[X_sX_t] = \rho_{st}$ for every $1 \leq s < t \leq 4$, then
\[
\E[H_i(X_1)H_j(X_2)H_k(X_3)H_l(X_4)] \EQ \sum_{\substack{a,b,c,d,e,f:\\a+b+c = i \\a+d+e = j \\b+d+f = k\\c+e+f = l}}\frac{\sqrt{i!j!k!l!}}{a!b!c!d!e!f!} \cdot \rho_{12}^a\rho_{13}^b\rho_{14}^c\rho_{23}^d\rho_{24}^{e}\rho_{34}^f\;.
\]
\end{corollary}
\begin{proof}
By linearity of expectation, we have
\begin{align*}
    & \E[H_i(X_1)H_j(X_2)H_k(X_3)H_l(X_4)] \\
  \EQ & \sum_{t_1 = 0}^{\lfloor i/2 \rfloor}\sum_{t_2 = 0}^{\lfloor j/2 \rfloor}\sum_{t_3 = 0}^{\lfloor k/2 \rfloor}\sum_{t_4 = 0}^{\lfloor l/2 \rfloor}\E\left[\textstyle\frac{(-1)^{t_1} m_{t_1}(K_i)X_1^{i - 2t_1}\cdot(-1)^{t_2} m_{t_2}(K_j)X_2^{j - 2t_2}\cdot(-1)^{t_3} m_{t_3}(K_k)X_3^{k - 2t_3}\cdot(-1)^{t_4} m_{t_4}(K_l)X_4^{l - 2t_4}}{\sqrt{i!j!k!l!}}\right] \\
  \EQ & \frac{1}{\sqrt{i!j!k!l!}}\sum_{t_1 = 0}^{\lfloor i/2 \rfloor}\sum_{t_2 = 0}^{\lfloor j/2 \rfloor}\sum_{t_3 = 0}^{\lfloor k/2 \rfloor}\sum_{t_4 = 0}^{\lfloor l/2 \rfloor}\textstyle(-1)^{t_1 + t_2 + t_3 + t_4}m_{t_1}(K_i)m_{t_2}(K_j)m_{t_3}(K_k)m_{t_4}(K_l)\E\left[ X_1^{i - 2t_1} X_2^{j - 2t_2} X_3^{k - 2t_3}X_4^{l - 2t_4}\right].
\end{align*}
We can apply Isserlis' theorem to $\E\left[ X_1^{i - 2t_1} X_2^{j - 2t_2} X_3^{k - 2t_3}X_4^{l - 2t_4}\right]$ and express it in terms of matchings. If we look at the expression so obtained for $\E[H_i(X_1)H_j(X_2)H_k(X_3)H_l(X_4)]$ combinatorially, we are doing the following\strikeout{thing}. We have a complete graph whose vertex set is partitioned into 4 parts, $V_1, V_2, V_3, V_4$ with $|V_1| = i, |V_2| = j, |V_3| = k, |V_4| = l$. We first pick some partial matching $M_1$  consisting of edges whose two endpoints are in the same part, and we get a factor of $(-1)^{|M_1|}$. We then match up the remaining vertices arbitrarily and get a partial matching $M_2$. Each edge in $M_2$ with one endpoint in $V_s$ and one endpoint in $V_t$ will contribute a factor of $\E[X_sX_t]$. $M_1 \cup M_2$ is a perfect matching of our complete graph, and note that a perfect matching may be counted multiple times in this procedure since $M_2$ is arbitrary. 

We claim that if a perfect matching $M$ contains an edge whose endpoints are from the same part, then it is counted multiple times and the contributions from each time sum up to zero. Indeed, if $M$ contains an edge whose both endpoints are in, say $V_1$, then this edge can be included in $M_1$, the partial matching produced in the first step, and contribute $-1$. or it can be included in $M_2$ and contribute $\E[X_1^2] = 1$. So the contributions cancel out. On the other hand, if $M$ does not contain edges whose endpoints are from the same part, then it is counted exactly once, for $M_1$ must be empty, and its contribution is $\rho_{12}^a\rho_{13}^b\rho_{14}^c\rho_{23}^d\rho_{24}^{e}\rho_{34}^f$ where $a$ is the number of edges between~$V_1$ and~$V_2$, $b$ is the number of edges between~$V_1$ and~$V_3$, etc.

Now we need to count the number of perfect matchings given the numbers of edges across any two different parts. It is easy to see that if $a$ is the number of edges between~$V_1$ and~$V_2$, $b$ is the number of edges between~$V_1$ and~$V_3$, etc, then the number of such matchings is given by
\[
\binom{i}{a, b, c}\cdot \binom{j}{a,d,e} \cdot \binom{k}{b,d,f}\cdot \binom{l}{c,e,f}\cdot a!b!c!d!e!f! \EQ \frac{i!j!k!l!}{a!b!c!d!e!f!}\;.
\]
Hence the corollary follows.
\end{proof}

We can then use this corollary to obtain an expression for $F_4$. However, unlike the expression for $F_2$, the expression that we obtain does not converge for all valid pairwise biases.

\section{An Example Where \texorpdfstring{$F_4$}{F4} is Negative on Positive Inputs}\label{sub:F4negative}
The idea for this example is to choose vectors $\bv_1,\bv_2,\bv_3,\bv_4$ such that $\bv_2 + \bv_3 + \bv_4 = (2-\delta)\bv_1$ for some small $\delta > 0$ but the components of $\bv_2,\bv_3,\bv_4$ which are orthogonal to $\bv_1$ have negative inner products. We can do this by choosing three orthonormal vectors $\be_1,\be_2,\be_3$ and taking 
\begin{enumerate}
\item $\bv_1 = \be_1$
\item $\bv_2 = \frac{(2 - \delta)}{3}\be_1 + \frac{1}{3}\sqrt{5 + 4\delta - {\delta}^2}\be_2$
\item $\bv_3 = \frac{(2 - \delta)}{3}\be_1 + \frac{1}{3}\sqrt{5 + 4\delta - {\delta}^2}\left(-\frac{1}{2}\be_2 + \frac{\sqrt{3}}{2}\be_3\right)$
\item $\bv_4 = \frac{(2 - \delta)}{3}\be_1 + \frac{1}{3}\sqrt{5 + 4\delta - {\delta}^2}\left(-\frac{1}{2}\be_2 - \frac{\sqrt{3}}{2}\be_3\right)$
\end{enumerate}
This gives the following pairwise biases:
\begin{enumerate}
\item $\forall i \in \{2,3,4\}, b_{1,i} = \bv_1 \cdot \bv_i = \frac{2 - \delta}{3}$.
\item $\forall i < j \in \{2,3,4\}, b_{i,j} = \bv_i \cdot \bv_j = \frac{(2-\delta)^2}{9} - \frac{5 + 4\delta - {\delta}^2}{18} = \frac{1 - 4\delta + {\delta}^2}{6}$.
\end{enumerate}
Note that these pairwise biases are all positive as long as $\delta \in (0,2-\sqrt{3})$.

We now use the following rounding scheme for some $\varepsilon > 0$
\begin{enumerate}
\item Choose a random vector $\bu$.
\item If $|\bv_i \cdot \bu| \in [\varepsilon,1.5\varepsilon)$, take $x_i = sign(\bv_i \cdot \bu)$. Otherwise, choose $x_i$ by flipping a coin.
\end{enumerate}
\begin{lemma}
For this rounding scheme, if $x_1,x_2,x_3,x_4$ are all determined without flipping a coin, which happens with nonzero probability, then ${x_1}{x_2}{x_3}{x_4} = -1$.
\end{lemma}
\begin{proof}
Since $\bv_2 + \bv_3 + \bv_4 = (2-\delta)\bv_1$, we have that $(\bv_2 \cdot \bu) + (\bv_3 \cdot \bu) + (\bv_4 \cdot \bu) = (2-\delta)(\bv_1 \cdot \bu)$. If $x_2,x_3,x_4$ are all determined without flipping a coin then $\forall i \in \{2,3,4\}, |\bv_i \cdot \bu| \in [\varepsilon,1.5\varepsilon)$. We have the following cases for $(\bv_2 \cdot \bu)$, $(\bv_3 \cdot \bu)$, and $(\bv_4 \cdot \bu)$:
\begin{enumerate}
\item If $(\bv_2 \cdot \bu),(\bv_3 \cdot \bu),(\bv_4 \cdot \bu)$ are all in $[\varepsilon,1.5\varepsilon)$ then $(\bv_1 \cdot \bu) \in [\frac{3}{2 - \delta}\varepsilon,\frac{4.5}{2 - \delta}\varepsilon)$ so $x_4$ is determined by a coin flip. 
\item Similarly, if $(\bv_2 \cdot \bu),(\bv_3 \cdot \bu),(\bv_4 \cdot \bu)$ are all in $(-1.5\varepsilon,-\varepsilon]$ then $x_4$ is determined by a coin flip.
\item If two of the inner products $(\bv_2 \cdot \bu),(\bv_3 \cdot \bu),(\bv_4 \cdot \bu)$ are in $[\varepsilon,1.5\varepsilon)$ and the other inner product is in $(-1.5\varepsilon,-\varepsilon]$ then $(\bv_1 \cdot \bu) \in (\frac{.5}{2 - \delta}\varepsilon,\frac{2}{2 - \delta}\varepsilon)$ so either $x_1$ is determined by a coin flip or $x_1 = 1$. If $x_1 = 1$, which happens with nonzero probability, then ${x_1}{x_2}{x_3}{x_4} = -1$. 
\item Similarly, if two of the inner products $(\bv_2 \cdot \bu),(\bv_3 \cdot \bu),(\bv_4 \cdot \bu)$ are in $(-1.5\varepsilon,-\varepsilon]$ and the other inner product is in $[\varepsilon,1.5\varepsilon)$ then either $x_1$ is determined by a coin flip or $x_1 = -1$. If $x_1 = -1$ then ${x_1}{x_2}{x_3}{x_4} = -1$. \qedhere
\end{enumerate}
\end{proof}
\begin{corollary}
For this rounding scheme, taking $b_{1,i} = \frac{2 - \delta}{3}$ for all $i \in \{2,3,4\}$ and taking $b_{i,j} = \frac{(2-\delta)^2}{9} - \frac{5 + 4\delta - {\delta}^2}{18} = \frac{1 - 4\delta + {\delta}^2}{6}$ for all $i < j \in \{2,3,4\}$, $F_4(b_{1,1}, b_{1,2}, b_{1,3}, b_{2,3}, b_{2,4}, b_{3,4}) < 0$.
\end{corollary}
\begin{remark}
The reason why this example works is that although the pairwise biases are all positive, once we consider the components of $\bv_2,\bv_3,\bv_4$ which are orthogonal to $\bv_1$, their inner products are negative. We conjecture that if the inner products remain positive throughout the Gram-Schmidt process then $F_4$ must be non-negative.
\end{remark}

\section{\change{Existence} of \change{a} Minimizer}\label{app:minimizer}

In this appendix, we justify the claim that there exists $f : \mathbb R \to [-1, 1]$ with $\|f\|_{\infty} \le 1$ which minimizes
\[
L(f) \;=\ \frac{1-\alpha}{\alpha} \int_{-\infty}^\infty f(x)^2\phi(x)dx + \int_{\mathbb R^2}f(x)f(y) \phi_{\rho}(x, y)\,dx dy \;,
\]
where $\rho > -1$. The functional we consider for MAX NAE-$\{3\}$-SAT are justified by near-identical logic.

Consider a sequence of functions $g_i$ with $\|g\|_{\infty} \le 1$ such that
\[
\lim_{i \to \infty} L(g_i) \EQ \inf \{L(f): \|f\|_{\infty}\le 1\}\;.
\]
Note that by Claim~\ref{claim:nae3-one}, each $g_i$ can be assumed to be monotone. We shall use the topology of the space of functions to construct a function $g$ for which $L(g)$ equals this limit. In that case, $g$ will be the minimizer we desire. Such proofs are standard in functional analysis (e.g.,~\cite{leoni2017first}).

Consider the Hilbert space $L^2(\phi)$, where $\phi$ is Gaussian measure on $\mathbb R$. Note that every $f$ such that $\|f\|_{\infty} \le 1$ has norm at most $1$ in this space. By the Banach-Alaoglu theorem, we have there is a subsequence $g_{\ell_i}$ of the $g_i$'s which is \emph{weakly convergent} to a function $g$. In particular, for any other $h \in L^2(\phi)$ we have that
\[
\lim_{i \to \infty} \int_{\mathbb R} g_{\ell_i}(x)h(x)\phi(x)\,dx \EQ \int_{\mathbb R}g(x)h(x)\phi(x)\,dx\;.
\]
By taking $h$ to be the \strikeout{the} indicator \strikeout{an} \change{of} an interval $[a, b]$ and letting $a$ approach $b$, we have that $g$ is the pointwise limit of the $g_{\ell_i}$ almost everywhere.\footnote{Note that we need the $g_{\ell_i}$ are bounded and monotone to make this deduction.} Thus, since the $g_{\ell_i}$ are bounded, on every compact interval $[a, b]$, $g_{\ell_i}$ strongly converge to $g$ in $L^2([a, b])$. Therefore, for every interval $[a, b]$, we have that
\begin{align*}
\lim_{i \to \infty} &\left[\frac{1-\alpha}{\alpha} \int_{a}^b g_{\ell_i}(x)^2\phi(x)dx + \int_{[a,b]^2}g_{\ell_i}(x)g_{\ell_i}(y) \phi_{\rho}(x, y)\,dx dy\right]\\&\EQ \frac{1-\alpha}{\alpha} \int_{a}^b g(x)^2\phi(x)dx + \int_{[a,b]^2}g(x)g(y) \phi_{\rho}(x, y)\,dx dy \;.
\end{align*}
For for every $\eps > 0$, $1-\eps$ of the mass of $\phi(x)$ as well as $\phi_{\rho}(x, y)$ is in a bounded rectangle. Thus, we may send $a \to -\infty$ and $b \to -\infty$ to infer that
\[
\lim_{i \to \infty} L(g_i) \EQ  L(g)\;.
\]

\section{Heuristic argument for the optimality of \texorpdfstring{$\pm1$}{+/-1} functions}\label{sub-heuristic}

In this appendix, we give a `heuristic argument' in support of Conjecture~\ref{C-Step}, stating that the optimal rounding function~$f_K$, when $|K|\ge 2$ and $\min K=3$, is a $\pm1$ step function. We stress that the argument, as presented, is not rigorous. It might, however, give some intuition why the conjecture might be true. We focus for simplicity on the case $K=\{3,k\}$, for $k>3$.

We say that a function $f:(-\infty,\infty)\to[-1,1]$ is \emph{locally optimal} for MAX NAE-$K$-SAT if there exists $\varepsilon>0$ such that there is no continuous function $g:(-\infty,\infty)\to[-\eps,\eps]$ such that $f+g:(-\infty,\infty)\to[-1,1]$ and $\alpha_K(f+g)>\alpha_K(f)$. (See Section~\ref{sub-difficult} for the definition of $\alpha_K(f)$.) The optimal function $f_K$ must of course be locally optimal.

Suppose that $f=f_K$ does not assume the values $\pm 1$ almost everywhere. Then, there must exist two disjoint intervals $[a_1,b_1]$ and $[a_2,b_2]$ such that $-1+\delta \le f(x)\le 1-\delta $ for every $x\in [a_1,b_1]\cup [a_2,b_2]$.

Let $g_1(x)=1$, if $x\in[a_1,b_1]$, and $g_1(x)=0$, otherwise. Define $g_2(x)$ similarly for the interval $[a_2,b_2]$.
For $-\delta<\eps_1,\eps_2<\delta$, the function $f+\eps_1g_1+\eps_2g_2$ is a function from $(-\infty,\infty)$ to $[-1,1]$.
Let $d_{k,i}=\frac{d\alpha_k(f+\eps g_i)}{d\eps}$ be the derivative with respect to adding a small multiple of~$g_i$, for $i=1,2$, and where $k$ can also be $3$. For small enough $\eps_1,\eps_2$, we have $\alpha_k(f+\eps_1 g_1 + \eps_2 g_2) \approx \alpha_k(f) + d_{k,1}\eps_1 + d_{k,2}\eps_2$.

We expect $d_{3,1}$ and $d_{k,1}$, and also $d_{3,2}$ and $d_{k,2}$, to have opposite signs, if they are not~$0$, as otherwise $f$ is clearly not locally optimal, as can be seen by adding a small multiple of $g_1$, or a small multiple of $g_2$. 

What is more interesting is the possibility of improving $f$ by adding both a small multiple of $g_1$ and a small multiple of~$g_2$. We thus ask whether there exists small enough $\eps_1,\eps_2$ such that $d_{3,1}\eps_1 + d_{3,2}\eps_2\ge 0$ and $d_{k,1}\eps_1 + d_{k,2}\eps_2\ge 0$. A sufficient condition for the existence of such $\eps_1,\eps_2$ is that the matrix $D=\left(\begin{array}{cc}d_{3,1} & d_{3,2} \\ d_{k,1} & d_{k,2}\end{array}\right)$ is non-singular.

In other words, if $f$ is locally optimal and does not assume $\pm 1$ values almost everywhere, then the matrix $D=D_{[a_1,b_1],[a_2,b_2]}$ must be singular for every two intervals $[a_1,b_1]$ and $[a_2,b_2]$ in which $f$ assume intermediate values. We believe that this condition cannot be satisfied. \strikeout{But, as conceded} \change{However}, we do not have a rigorous proof.

The main reason we believe that this condition cannot be satisfied for $K=\{3,k\}$ is that the functions $\alpha_3(f)$ and $\alpha_k(f)$, for $k>3$, seem to be `pulling' in opposing directions.
If this argument can be made rigorous, the proof would need to rely on the condition $\min K=3$, as if $\min K>3$, the optimal function is \change{a} random assignment, i.e., $f_K(x)=0$, for every $x\in (-\infty,\infty)$.

\section{Equivalence of Uniform and Non-uniform MAX NAE-SAT}\label{S-equiv}

As discussed in the preliminaries, there are two natural definitions of MAX NAE-SAT: \emph{uniform} MAX NAE-SAT, where the clause lengths can grow with the number of variables, and \emph{non-uniform} MAX NAE-SAT, where the clause lengths can be an arbitrarily large constant (that is the limit of MAX NAE-$[k]$-SAT as~$k$ goes to infinity). In this appendix, we show that in terms of approximation, these two problems are essentially equivalent.  

Let $U_k(c)$ be the best soundness one can achieve in polynomial time for MAX NAE-$[k]$-SAT (assuming UGC). Let $U(c)$ be the corresponding quantities for MAX NAE-SAT. Note the following is true.

\begin{proposition}\label{prop:subtle}
For any CSP $\Lambda$, and $c, c', t \in [0, 1]$, such that $1-t+tc\ge c'  \ge tc$, we have $1-t+tU_{\Lambda}(c) \ge U_\Lambda(c')$.
\end{proposition}

\begin{proof}
Consider an instance formed by taking the conjunction of a formula $\Psi_1(x)$ with total weight $t$ and completeness $c$ and a second formula $\Psi_2(y)$ with total weight $1-t$ and completeness $\frac{c'-tc}{1-t}$. Let $A$ be a polynomial time algorithm which takes as input $\Psi_1$ and $\Psi_2$ and outputs an assignment $(x, y)$. By definition of $U_{\Lambda}(c)$, for any $\eta > 0$, there must be an input $(\Psi'_1, \Psi'_2)$ such that $(x, y) := A(\Psi'_1, \Psi'_2)$ satisfies at most $U_{\Lambda}(c) + \eta$ of the clauses of $\Psi'_1$. Thus, the output of $A$ satisfies at most $1-t + t(U_{\Lambda}(c) + \eta)$ fraction of the clauses. This gives the upper bound on $U_\Lambda(1-t + tc)$. 
\end{proof}

The following is the main result of this appendix.

\begin{lemma}
For all $c \in [0, 1]$,
\[
U(c) \EQ \lim_{k\to \infty} U_k(c)\;.
\]  
\end{lemma}

\begin{proof}
It suffices to prove the following two inequalities for all $\epsilon > 0$,
\begin{align}
U(c) &\LE \lim_{k \to \infty} U_k(c) \label{eq:b}\\
\lim_{k \to \infty} U_k(c) &\LE U(c)+\epsilon\;. \label{eq:a}
\end{align}

The inequality (\ref{eq:b}) follows from the fact that any instance of MAX NAE-$[k]$-SAT is an instance of uniform MAX NAE-SAT. 

To prove (\ref{eq:a}), it suffices to describe an algorithm for MAX NAE-SAT for any $\epsilon > 0$. Let $\Phi$ be an instance of MAX NAE-SAT.

Define $P_k(\delta) := 1 - 2(1-\delta)^k + (1-2\delta)^k$ and $Q_k(\delta) := (1-2\delta)^k$. 

Run the following algorithm.

\begin{enumerate}
\item Guess $k = O_{\eps}(1)$ and set $\delta = Q_k^{-1}(1-\epsilon/2) = \frac{1 - (1-\epsilon/2)^{1/k}}{2}$.
\item Let $\Phi' \subset \Phi$ be the instance consisting of all clauses having length at most $k$. Run the optimal polynomial-time approximation algorithm for this instance and get a solution $x$.
\item For each variable $x_i$, $i \in [n]$. With probability $\delta$, set $y_i = 1$ and with probability $\delta$, set $y_i = -1$. Otherwise set $y_i = x_i$.
\item Output $y$.
\end{enumerate}

It suffices to prove that the above algorithm works for some $k = O_{\eps}(1)$ as then we have a polynomial time algorithm by enumerating over all $k$ at most $O_{\eps}(1)$.

First, assume $k$ is fixed. To analyze the algorithm, fix an optimal solution $z$ to $\Phi$. For each $i \ge 1$, let $w_i$ the relative weight of clauses of $\Phi$ with size $i$ (so that $\sum_i w_i = 1$) and let $c_i$ be the fraction of clauses of $\Phi$ satisfied by $x$. Now define
\[
c'_k \EQ \frac{\sum_{i=1}^k w_i c_i}{\sum_{i=1}^k w_i}\;,
\]
that is the completeness of $\Phi'$. By definition, the solution $x$ will be a $U_k(c'_k) -\frac{\epsilon}{4}$ approximate solution to~$\Phi'$.

Observe that for a clause of length $i$, it is satisfied by (3) with probability $1 - 2(1-\delta)^i + (1-2\delta)^i = P_i(\delta).$ Further, the probability that a clause of length $k$ is "untouched" (that is, no variables are changed) by (3) is $(1-2\delta)^k = Q_k(\delta)$.

For a given $k$ and $\delta$ such that $Q_k(\delta) = 1-\epsilon/2$, let $k'$ be the smallest integer such that $P_{k'}(\delta) \ge 1-\epsilon/2$. Thus, every clause of length at least $k'$ is satisfied with probability at least $1-\epsilon/2$. Thus, the fraction of clauses satisfied by our algorithm is at least
\begin{align*}
&\left(1-\frac{\epsilon}{2}\right) \sum_{i = k'}^n w_i + U_k(c'_k)\sum_{i=1}^k w_i - \frac{\epsilon}{4}\\ 
\GE\; & \sum_{i = k'}^n w_i + U_k(c'_k)\sum_{i=1}^k w_i - \frac{\epsilon}{4}\\
\EQ\; & \sum_{i = k+1}^n w_i + U_k(c'_k)\sum_{i=1}^k w_i - \sum_{i=k+1}^{k'}w_i - \frac{\epsilon}{4}\\
\GE\; & U_k\left(\sum_{i=1}^n c_iw_i\right) - \sum_{i=k+1}^{k'}w_i - \frac{\epsilon}{4}\quad(\text{Proposition \ref{prop:subtle}})\\
\EQ\; & U_k(c)-\sum_{i=k+1}^{k'}w_i - \frac{\epsilon}{4}.
\end{align*}

To finish, it suffices to show there exists $k = O_{\epsilon}(1)$ such that $\sum_{i=k+1}^{k'}w_i \le \frac{\epsilon}{4}$. 

To see why, fix $k = 3$ and then consider the sequence $k', k'', k''', \hdots k^{(5/\epsilon)}$. (Where $k^{(i+1)}$ is $k'$ for $k = k^{(i)}$.) Since $\sum_{i=1}^n w_i = 1$, there must exist $i$ such that $\sum_{i=k+1}^{k'}w_i \le \frac{\epsilon}{4}$ for $k = k^{(i)}$. Thus, we have that (\ref{eq:a}) holds. 

\end{proof}

We remark that a similar argument can be used to prove uniform and non-uniform MAX SAT are equivalent to approximate.

\change{
\section{Proof of Lemma~\ref{lem:monotono_nae_diff}}\label{app:proof_monotone}
Given a function $f: \mathbb{R}^k \to [-1, 1]$, we will use $f^{\odd}$ to denote its odd part, defined by $\bv \mapsto \frac{f(\bv) - f(-\bv)}{2}$, and $f^{\even}$ its even part, defined by $\bv \mapsto \frac{f(\bv) + f(-\bv)}{2}$. Clearly, $f^{\even}, f^{\odd}: \mathbb{R}^k \to [-1, 1]$ and $f = f^{\even} + f^{\odd}$.}
\change{The most difficult term to deal with in the performance of $Round_f$ when $f$ is not odd is $F_4[f](\rho, \rho, \rho, 0, 0, 0)$, and our strategy is to show that the gain in other terms more than compensates for what's potentially lost in $F_4[f](\rho, \rho, \rho, 0, 0, 0)$. To implement this, we first write down explicitly how these terms change if we replace $f$ with its odd part.}

\change{\begin{proposition}\label{prop:monotone_nae5_expressions}
    Let $f: \mathbb{R}^k \to [-1, 1]$ and $\rho \in [0, 1]$. We have the following equalities:
    \begin{itemize}
        \item[(a)] $F_4[f](\rho, \rho, \rho, \rho, \rho, \rho) = \E\left[(\U_{\sqrt{\rho}} f^\odd)^4\right] +\E\left[(\U_{\sqrt{\rho}} f^\even)^4\right] + 6\E\left[(\U_{\sqrt{\rho}} f^\odd)^2(\U_{\sqrt{\rho}} f^\even)^2\right]$.
        \item[(b)] $6F_2[f]\left(\rho\right) + 4F_2[f](0) = 6\E\left[(\U_{\sqrt{\rho}} f^\odd)^2\right]  + 6\E\left[(\U_{\sqrt{\rho}} f^\even)^2\right] + 4\E[f]^2$.
        \item[(c)] $F_4[f](\rho, \rho, \rho, 0, 0, 0) = \E[f] \cdot \left(\E\left[(\U_{\sqrt{\rho}} f^{\even}(\bx))^3\right] + 3\E\left[\U_{\sqrt{\rho}} f^{\even}(\bx) \cdot (\U_{\sqrt{\rho}} f^{\odd}(\bx))^2\right]\right).$
    \end{itemize}
\end{proposition}
\begin{proof}
    For (a), by Lemma~\ref{lemma:noisyf},
    \[
     F_4[f](\rho, \rho, \rho, \rho, \rho, \rho) = \E\left[(\U_{\sqrt{\rho}} f)^4\right] = \E\left[(\U_{\sqrt{\rho}} f^\odd + \U_{\sqrt{\rho}}f^\even)^4\right].
    \]
    If we expand the 4-th power, any term that has an odd power of $\U_{\sqrt{\rho}} f^\odd$ will have expectation 0, so we have
    \[
    \E\left[(\U_{\sqrt{\rho}} f)^4\right] = \E\left[(\U_{\sqrt{\rho}} f^\odd)^4\right] +\E\left[(\U_{\sqrt{\rho}} f^\even)^4\right] + \binom{4}{2}\cdot\E\left[(\U_{\sqrt{\rho}} f^\odd)^2(\U_{\sqrt{\rho}} f^\even)^2\right]. 
    \]
    Part (b) follows directly from Lemma~\ref{lemma:noisyf} and Proposition~\ref{prop:F2power}.
    For part (c), by Lemma~\ref{lemma:moment_orthogonal}, we have    
\[
F_4[f](\rho, \rho, \rho, 0, 0, 0) = F_3[f](\rho) \cdot \E[f].
\]
By Lemma~\ref{lemma:noisyf}, we have 
\[
F_3[f](\rho) = \E\left[(\U_{\sqrt{\rho}} f(\bx))^3\right] = \E\left[(\U_{\sqrt{\rho}} f^{\even}(\bx) + \U_{\sqrt{\rho}} f^{\odd}(\bx))^3\right].
\]
If we expand the cubed binomial, any term that has an odd power of $\U_{\sqrt{\rho}} f^{\odd}(\bx)$ will have expectation 0, and therefore we are left with
\[
F_3[f](\rho) = \E\left[(\U_{\sqrt{\rho}} f^{\even}(\bx))^3\right] + 3\E\left[\U_{\sqrt{\rho}} f^{\even}(\bx) \cdot (\U_{\sqrt{\rho}} f^{\odd}(\bx))^2\right]. \qedhere
\]
\end{proof}
We now show that 
\begin{proposition}\label{prop:bound_f4_rho0}
    Let $f: \mathbb{R}^k \to [-1, 1]$ and $\rho \in [0, 1]$. We have
    \begin{itemize}
        \item[(a)] $\E\left[(\U_{\sqrt{\rho}} f^\odd)^2(\U_{\sqrt{\rho}} f^\even)^2\right] + \E[f]^2 \geq \left|2\E[f]\cdot\E\left[(\U_{\sqrt{\rho}} f^\odd)^2 \cdot\U_{\sqrt{\rho}} f^\even\right]\right|$.
        \item[(b)] $\E[f]^2 + F_2[f^\even](\rho)\geq \left|2\E[f]\cdot\E\left[(\U_{\sqrt{\rho}} f^\even)^3\right]\right|$.
    \end{itemize}
\end{proposition}
\begin{proof}
     For part (a), note that for any random variables $X$ and $Y$ we have $\E\left[X^2(Y - \E[Y])^2\right] = \E[X^2Y^2] + \E[X^2]\E[Y]^2 - 2\E\left[X^2Y\right]\E[Y] \geq 0$, and by letting $X = \U_{\sqrt{\rho}} f^\odd$ and $Y = \U_{\sqrt{\rho}} f^\even$ we have
    \begin{align*}
    \E\left[(\U_{\sqrt{\rho}} f^\odd)^2(\U_{\sqrt{\rho}} f^\even)^2\right] + \E[f]^2 & \geq \E\left[(\U_{\sqrt{\rho}} f^\odd)^2(\U_{\sqrt{\rho}} f^\even)^2\right] + \E\left[(\U_{\sqrt{\rho}} f^\odd)^2\right]\E[f]^2 \\
    & \geq 2\E[f]\cdot\E\left[(\U_{\sqrt{\rho}} f^\odd)^2 \cdot\U_{\sqrt{\rho}} f^\even\right].
    \end{align*}
    By considering $\E\left[X^2(Y + \E[Y])^2\right]$, a similar argument shows that 
    \[
    \E\left[(\U_{\sqrt{\rho}} f^\odd)^2(\U_{\sqrt{\rho}} f^\even)^2\right] + \E[f]^2 \geq -2\E[f]\cdot\E\left[(\U_{\sqrt{\rho}} f^\odd)^2 \cdot\U_{\sqrt{\rho}} f^\even\right].
    \]
    Thus, part (a) follows. For part (b), we have
    \begin{align*}
    \left|\E\left[(\U_{\sqrt{\rho}} f^\even)^3\right]\right| \leq \E\left[\left|\U_{\sqrt{\rho}} f^\even\right|^3\right] \leq \E\left[\left(\U_{\sqrt{\rho}} f^\even\right)^2\right]= F_2[f^\even](\rho).
    \end{align*}
    It follows that
    \[
    \left|2\E[f]\cdot\E\left[(\U_{\sqrt{\rho}} f^\even)^3\right]\right| \leq \left|2\E[f]\cdot F_2[f^\even](\rho)\right| \leq \E[f]^2 + (F_2[f^\even](\rho))^2 \leq \E[f]^2 + F_2[f^\even](\rho). \qedhere
    \]
\end{proof}
We are now finally ready to prove  Lemma~\ref{lem:monotono_nae_diff}
\begin{proof}[Proof of Lemma~\ref{lem:monotono_nae_diff}]
The first part follows from Proposition~\ref{prop:F2power}:
\begin{equation*}
\frac{3 - 3F_2[f]\left(-\rho\right)}{4} = \frac{3 - 3F_2[f^{\even}]\left(\rho\right) + 3F_2[f^{\odd}]\left(\rho\right)}{4} \leq  \frac{3 + 3F_2[f^{\odd}]\left(\rho\right)}{4}. 
\end{equation*}
For the second part, we have by Proposition~\ref{prop:monotone_nae5_expressions}
\begin{align*}
    &\frac{15 - 6F_2[f^\odd]\left(\rho\right) - F_4[f^\odd](\rho)}{16} \\
    &\quad- \frac{15 - 6F_2[f]\left(\rho\right) - 4F_2[f](0) - F_4[f](\rho) - 4F_4[f](\rho, \rho, \rho, 0, 0, 0)}{16} \\
    =\,\,& \frac{6\E\left[(\U_{\sqrt{\rho}} f^\even)^2\right] + 4\E[f]^2 + \E\left[(\U_{\sqrt{\rho}} f^\even)^4\right] + 6\E\left[(\U_{\sqrt{\rho}} f^\odd)^2(\U_{\sqrt{\rho}} f^\even)^2\right]}{16} \\
    & \quad + \frac{4\E[f] \cdot \left(\E\left[(\U_{\sqrt{\rho}} f^{\even}(\bx))^3\right] + 3\E\left[\U_{\sqrt{\rho}} f^{\even}(\bx) \cdot (\U_{\sqrt{\rho}} f^{\odd}(\bx))^2\right]\right)}{16}\\
    \geq\,\,&\frac{2\E\left[(\U_{\sqrt{\rho}}
    f^\even)^2\right] + 8\E[f]^2 + 6\E\left[(\U_{\sqrt{\rho}} f^\odd)^2(\U_{\sqrt{\rho}} f^\even)^2\right]}{16} \\
    & \quad + \frac{4\E[f] \cdot \left(\E\left[(\U_{\sqrt{\rho}} f^{\even}(\bx))^3\right] + 3\E\left[\U_{\sqrt{\rho}} f^{\even}(\bx) \cdot (\U_{\sqrt{\rho}} f^{\odd}(\bx))^2\right]\right)}{16} \\
    =\,\,&\frac{2\E\left[(\U_{\sqrt{\rho}}
    f^\even)^2\right] + 2\E[f]^2 + 4\E[f] \cdot\E\left[(\U_{\sqrt{\rho}} f^{\even}(\bx))^3\right] }{16} \\
    & \quad + \frac{6\E\left[(\U_{\sqrt{\rho}} f^\odd)^2(\U_{\sqrt{\rho}} f^\even)^2\right] + 12\E[f] \cdot\E\left[\U_{\sqrt{\rho}} f^{\even}(\bx) \cdot (\U_{\sqrt{\rho}} f^{\odd}(\bx))^2\right] + 6\E[f]^2 }{16} \\
    \geq\,\,& 0.
\end{align*}
    Here in the first inequality we used $\E\left[(\U_{\sqrt{\rho}}
    f^\even)^2\right] \geq \E[f]^2$ and $\E\left[(\U_{\sqrt{\rho}} f^\even)^4\right] \geq 0$, while in the second inequality we used Proposition~\ref{prop:bound_f4_rho0}.
\end{proof}}

\change{
\section{Analysis of the Explicit Integrality Gap Instance}\label{appendix:explicit_gap}
This appendix is devoted to the proof of the following theorem for the explicit gap instance $\Phi_n$ constructed in Section~\ref{sub-explicit}.
\begin{theorem}
For any integral solution to $\Phi_n$, the weight of the satisfied clauses is at most $\frac{3(\sqrt{21}-4)}{2} + O(\frac{1}{n})$.
\end{theorem}
To analyze the weight of the satisfied constraints for a given solution, we consider the following distributions.
\begin{definition}
For every $k < n/2$, we define $\mathcal{D}_k$ to be the following distribution over $(V_n)^k$:
\begin{enumerate}
    \item Sample $2k+1$ distinct indices $i_1, i_2, \ldots, i_{2k+1} \in [n]$ uniformly at random.
    \item Sample $2k + 1$ independent random coin flips $b_1, \ldots, b_{2k + 1} \in \{-1, 1\}$.
    \item For every $j \in [k]$, let $\bv_j = \frac{1}{\sqrt{3}}(b_1\be_{i_1} + b_{2j}\be_{i_{2j}} + b_{2j +1}\be_{i_{2j+1}})$. Return the $k$-tuple $(\bv_1, \bv_2, \ldots, \bv_{k})$.
\end{enumerate}
\end{definition}
Informally speaking, this distribution samples $k$ vectors from $V_n$ of ``sunflower shape'' in the sense that all of them share exactly one index on which they are nonzero.
\begin{definition}
Given an assignment, we let 
\[F_2 \EQ \underset{(\bv_1, \bv_2) \sim \mathcal{D}_2}{\E}[x_{\bv_1}x_{\bv_2}] \quad,\quad
F_4 \EQ \underset{(\bv_1, \bv_2, \bv_3, \bv_4) \sim \mathcal{D}_4}{\E}[x_{\bv_1}x_{\bv_2}x_{\bv_3}x_{\bv_4}] \;. \]
\end{definition}
\begin{remark}
Here $F_2$ and $F_4$ come from an actual assignment rather than a rounding scheme, but they play the same role in the argument.
\end{remark}
\begin{proposition}
Given an assignment, the proportion of $3$-clauses which are satisfied is $\frac{3+3F_2}{4}$ and the proportion of $5$-clauses which are satisfied is $\frac{15-6F_2-F_4}{16}$ 
\end{proposition}
\begin{proof}[Proof sketch]
This can be shown by expanding out each constraint as a polynomial.
\end{proof}
By Lemma \ref{NAE35ratiolemma}, if we had that $F_4 \geq F_2^2$ then we would have that the total weight of the satisfied clauses is at most $\frac{3(\sqrt{21}-4)}{2}$. Instead, we show that $F_4 \geq F_2^2 - O(\frac{1}{n})$. Adapting the argument in Lemma \ref{NAE35ratiolemma} accordingly, this implies that the total weight of the satisfied clauses is at most $\frac{3(\sqrt{21}-4)}{2} + O(\frac{1}{n})$.
\begin{lemma}\label{lem:f2f4ineq}
For any assignment,
\[
F_4 \GE F_2^2 - O\Bigl(\frac{1}{n}\Bigr)\;.
\]
\end{lemma}
\begin{proof}
Let $k = \lfloor n/2 \rfloor - 1 < n / 2$. Sample $(\bv_1, \ldots, \bv_k) \sim \mathcal{D}_k$. Note that the marginal distribution of any pair of these vectors is exactly $\mathcal{D}_2$ and any 4 vectors exactly $\mathcal{D}_4$. Now let $X = \sum_{i = 1}^k x_{\bv_i}$. We have the inequality
\[
\Var[X^2] \EQ \E[X^4] - \left(\E[X^2]\right)^2 \GE 0\;.
\]
We have that
\[
\E\left[X^2\right] \EQ \E\left[\left(\sum_{i = 1}^k X_{\bv_i}\right)^2\right] \EQ \sum_{i=1}^k \E[X_{\bv_i}^2] + \sum_{i \neq j} \E[X_{\bv_i}X_{\bv_j}] \EQ k + k(k-1)F_2\;.
\]
Here we used the fact that $X_{\bv_i} \in \{-1, 1\}$ and $X_{\bv_i}^2 = 1$. Similarly we can compute
\[
\E\left[X^4\right] \EQ \E\left[\left(\sum_{i = 1}^k X_{\bv_i}\right)^4\right] \EQ 3k^2 - 2k + k(k-1)(6k-8)F_2 + k(k-1)(k-2)(k-3)F_4\;.
\]
Plugging in these two expressions to the inequality above, we get
\[
3k^2 - 2k + k(k-1)(6k-8)F_2 + k(k-1)(k-2)(k-3)F_4 - (k + k(k-1)F_2)^2 \GE 0.
\]
Our lemma follows by shifting terms, dividing both sides by $k(k-1)(k-2)(k-3)$, and using the fact that $k = \Theta(n)$.
\end{proof}
}

\end{document}